\documentclass[graphicx,12pt]{article}
\usepackage[dvipdfm]{geometry}
\usepackage{caption,natbib}
\usepackage{epsfig,subfigure}
\usepackage[center]{titlesec}
\usepackage{booktabs}
\usepackage{epstopdf}
\usepackage{algorithm}
\usepackage{amsmath,amsthm,amssymb}
\usepackage{algorithmic}
\usepackage{color}
\usepackage[perpage,symbol*]{footmisc}
\newtheorem{theorem}{Theorem}
\newtheorem{lemma}{Lemma}
\newtheorem{condition}{Condition}

\newcommand{\tauo}{\tau_0}

\theoremstyle{definition}

\renewcommand{\baselinestretch}{1.6}

\begin{document}

\def\spacingset#1{\renewcommand{\baselinestretch}%
{#1}\small\normalsize} \spacingset{1.5}

\begin{center}
{\Large\bf \textsf{Linear spline index regression model: Interpretability, nonlinearity and dimension reduction}}

\vspace*{0.1in}

Lianqiang Qu$^1$, Long Lv$^1$, Meiling Hao$^2$ and Liuquan Sun$^{3}$\\
{\sl $^1$School of Mathematics and Statistics,
Central China Normal University,\\
Wuhan, Hubei, 430079, China}\\
{\sl $^2$ School of Statistics,
University of International Business and Economics,\\
Beijing, 100029, China}\\ 
{\sl $^3$ Academy of Mathematics and Systems Science, Chinese Academy of Sciences, \\
and School of Mathematical Sciences, University of Chinese Academy of Sciences,\\
Beijing 100190, P.R.China, slq@amt.ac.cn}
\end{center}

\noindent {\bf Abstract.}
Inspired by the complexity of certain real-world datasets,
this article introduces a novel flexible linear spline index regression model.
The model posits piecewise linear effects of an index on the response, with continuous changes occurring at knots.
Significantly, it possesses the interpretability of linear models, captures nonlinear effects similar to nonparametric models,
and achieves dimension reduction like single-index models.
In addition, the locations and number of knots remain unknown,
which further enhances the adaptability of the model in practical applications.
Combining the penalized approach and convolution techniques,
we propose a new method to simultaneously estimate the unknown parameters and the number of knots.
The proposed method allows the number of knots to diverge with the sample size.
We demonstrate that the proposed estimators can identify the number of knots with a probability approaching one
and estimate the coefficients as efficiently as if the number of knots is known in advance.
We also introduce a procedure to test the presence of knots.
Simulation studies and two real datasets are employed to assess the finite sample performance of the proposed method.


\vspace*{0.05in}

\noindent {\it  Keywords:} Convolution smoothing; Knot detection; Nonlinear effects; Piecewise linear model.

\spacingset{2.0}

\section{Introduction}
In the era of information explosion, vast amounts of data have been accumulated across various fields, including economics, medicine, and sociology.
In this context, we have access to
a real estate valuation dataset that includes two districts in Taipei City and two districts in New Taipei City.
The dataset consists of 414 records of real estate transactions,
with the variable of interest being the price per unit area of residential housing.
In line with the analysis conducted by \cite{yeh2018building}, we consider the following covariates:
the distance to the nearest MRT station (Meter), house age (Year), transaction date (Date),
and the number of convenience stores within walking distance in the vicinity (Number).

The aim is to investigate the effects of various covariates on the housing prices.
Figure \ref{Figure:real data}  provides insights into the relationship between the housing prices and an index
derived from a linear combination of Meter, Year, and Date.
The details of what  are included in the index are provided in Section \ref{section:real:data}.
The left panel of Figure \ref{Figure:real data} displays the estimated curve derived from the partial single-index regression model \citep{wang2010estimation}.
Meanwhile, the right panel showcases the fitted curve using the proposed method outlined in Section \ref{section:Estimate}.
The results indicate that the effects of the index on housing prices are nonlinear.
Specifically, the effects exhibit a significant change but remain continuous up to a certain point (knot).
More precisely, housing prices tend to remain stable when the index values are {less than} the knot.
However, once the index values surpass this point, the housing prices experience a rapid increase.
Identifying of this knot is crucial, as it facilitates a comprehensive analysis of the determinants
that influence the stability or rapid escalation in residential property prices.

In summary, for effective analysis of such data, a suitable model should be capable of capturing nonlinear effects of multidimensional covariates,
providing a clear interpretation of these effects, mitigating the curse of dimensionality, and accurately determining the locations and number of knots.
However, we have not yet identified a model that can simultaneously meet all the desired features.
Therefore, developing such a model remains to be an urgent but challenging task.

\begin{figure}
\centering
\subfigure{\includegraphics[width=0.45\textwidth]{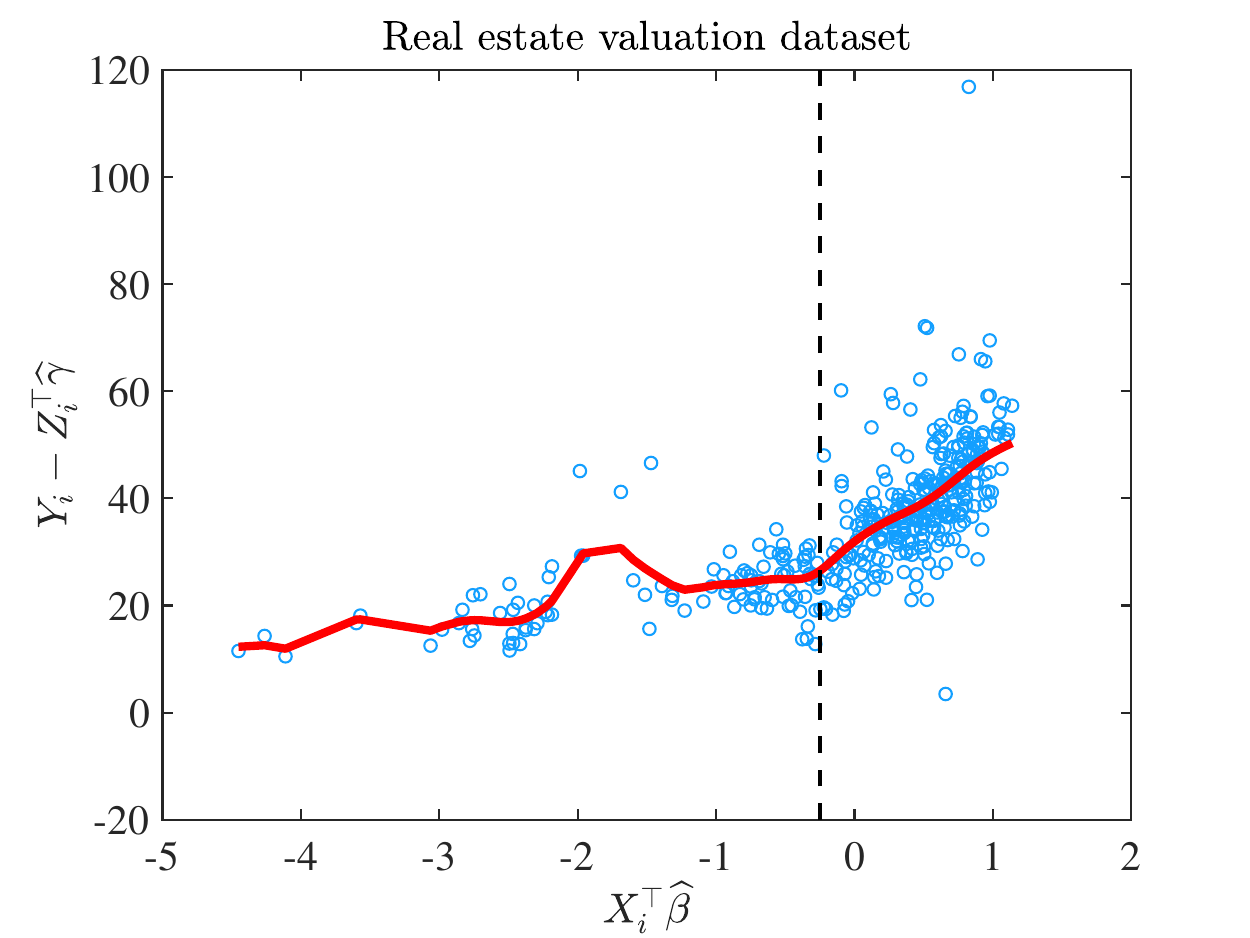}}
\subfigure{\includegraphics[width=0.45\textwidth]{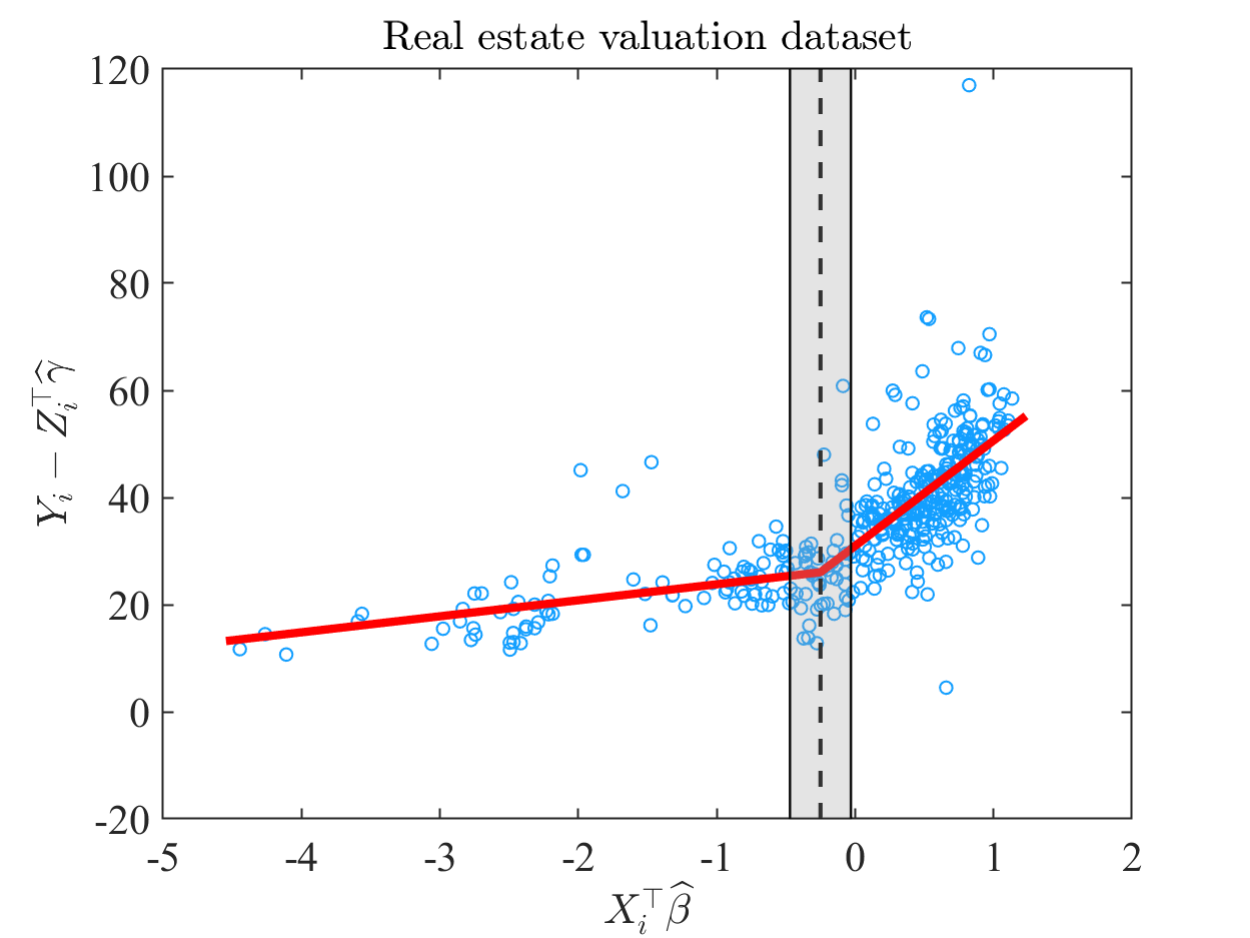}}
\caption{{\small The analysis of the real estate valuation dataset.
The red lines in the left and right panels are the estimated curves via the single-index regression model \citep{wang2010estimation}
and model \eqref{SIR-model}, respectively.
The grey region represents the $95\%$ confidence intervals of the estimates of the knot locations.}}\label{Figure:real data}
\end{figure}

In this article, we propose a novel linear spline index regression (LSIR) model for data analysis:
\begin{align}\label{SIR-model}
Y_i=\gamma_0+\alpha_0X_i^\top\beta+\sum_{m=1}^{M_n^*}\alpha_m f(X_i^\top\beta,\tau_m)+Z_i^\top\gamma+\epsilon_i,
\end{align}
where $Y_i$ denotes the response variable, $X_i=(X_{i1},\ldots,X_{id_1})^\top$
and $Z_i=(Z_{i1},\ldots,Z_{id_2})^\top$ denote the vectors of covariates, $\epsilon_i$ is the random error
and $f(x,\tau)=(x-\tau)I(x\ge \tau)$ with $I(\cdot)$ being the indicator function.
Here and in what follows, we call $X_i^\top\beta$ an index that aggregates the dimension of $X_i$.
The model allows for changes in slopes at unknown knot locations
$-\infty<\tau_1<\tau_2<\dots<\tau_{M_n^*}<\infty$,
and $M_n^*$ denotes the number of the knots.
The slope with respect to the index $X_i^\top\beta$ in the $m$th segment is denoted by $\mu_m=\sum_{k=0}^m\alpha_k$.
The nonzero coefficients $\alpha_m\neq 0~(m=0,1,\dots,M_n^*)$ capture differences in slopes between adjacent $m$th and $(m+1)$th segments.
The vectors $\beta=(\beta_1,\dots,\beta_{d_1})^\top\in\mathbb{R}^{d_1}$
and $\gamma=(\gamma_1,\dots,\gamma_{d_2})^\top\in\mathbb{R}^{d_2}$
represent the unknown coefficients for the covariates  $X_i$ and $Z_i,$ respectively,
and $\gamma_0$ denotes the intercept.

The proposed model is appealing since it melds the chief advantages of linear regression models, nonparametric models and single-index models.
Specifically, model \eqref{SIR-model} maintains interpretability similar to linear regression models,
where nonzero components of $\beta$ indicate significant predictors of the response variable.
The coefficients $\mu_m\beta_j$ of $X_{ij}$ in the $m$th segment can be interpreted as the rate of change in the response $Y_i$ associated with a unit increase in $X_{ij}$,
holding all other covariates fixed.
In addition, the locations and number of knots are not predetermined but can be data-driven estimated.
This feature enhances the flexibility and adaptability of model \eqref{SIR-model},
allowing it to effectively capture the underlying nonlinear effects of multidimensional covariates.
Furthermore, model \eqref{SIR-model} achieves dimension reduction by collapsing the influence of the covariates $X_i$ into a single-index $X_i^\top \beta$,
thus avoiding the curse of dimensionality.
In summary, the LSIR model incorporates all the aforementioned desired features.

The proposed method is closely related to two important models: (partial-linear) single-index models and linear spline regression models.
Single-index models have been widely used to analyze the nonlinear effects of an index on the response over the last few decades
because of their convenient in dimension reduction;
see e.g., \cite{powell1989semiparametric}, \cite{hardle1993optimal}, \cite{carroll1997generalized}, \cite{hristache2001direct}, \cite{wang2010estimation},
\cite{ma2013doubly} and the references therein.
However, single-index models assume that the slopes on the adjacent segments change continuously,
therefore they cannot be used for estimating and inferring the locations and number of knots.
This distinction sets our model apart from single-index models.
In addition, single-index models may lack interpretability within each segment due to the presence of an unknown link function; see the left panel of Figure \ref{Figure:real data}.
Linear spline models are alternatively known as kink regression models (\citealp{card2012nonlinear,hansen2017regression}),
bent-line models (\citealp{li2011bent}) and broken-line/stick models (\citealp{muggeo2003estimating}).
It is a piecewise linear regression model,
where the regression function is continuous but the slope exhibits discontinuities at various knots. 
This model, with a known number of knots, has been extensively studied in the existing literature
(e.g., \citealp{tishler1981new, muggeo2003estimating,li2011bent,card2012nonlinear,das2016fast,hansen2017regression,zhang2017panel,yang2023estimation}).
Studies on estimating the number of knots are still limited.
\cite{muggeo2011efficient} proposed utilizing the segmented and Lars algorithms to estimate the number of knots,
while \cite{zhong2022estimation} developed estimation and inference methods for multi-kink quantile regression
and introduced a BIC-type procedure for determining the number of knots.
Recently, the multi-kink quantile regression method has been developed in various contexts
\citep{wan2023multikink,SUN2024105304}.
Nevertheless, current methods primarily focus on incorporating the evolving effects of a single continuous variable,
such as time or age, on the response variable, rather than multidimensional covariates.
Furthermore, the aforementioned methods assume a fixed number of knots,
and are not suitable for situations where the number of knots increases with the sample size.
This is exactly the case we are considering in our work.

Hence, the estimation and inference of unknown parameters in the LSIR model present significant challenges that necessitate the development of new techniques,
particularly when the number of knots is unknown and can scale with the sample size.
The main contributions of our work are outlined as follows.
\begin{itemize}
\item We provide a general framework for analyzing the knot effects of multidimensional covariates.
The major technical challenges in fitting model \eqref{SIR-model} include the fact that the objective function
is nondifferentiable with respect to $\tau_m$ and $\beta$, and that the number of knots is unknown.
In addition, the presence of an unknown parameter $\beta$ in the index leads to the scale unknown,
which further complicates the identification of the  knot locations compared to kink regression models.
To tackle these challenges, we develop a new method that combines penalized approaches
and convolution techniques to simultaneously estimate the unknown parameters, as well as the locations and number of knots.
Here, the convolution is a local smoothing procedure,
which shares the similar spirit of knot estimation in kink regression models
(e.g., \citealp{tishler1981new,das2016fast}).
However, our method differs from the existing methods in several aspects.
Firstly, \cite{das2016fast} approximated the non-smoothing function with a quadratic function,
which can be regarded as a specific instance of the convolution smoothing approach.
In addition, we need a shrinkage procedure for the estimation of the number of knots.
This requires a proper decay of the smoothing parameter (bandwidth) in the penalized procedure.
Large bandwidths cannot achieve the optimal statistical rate, while small bandwidths are ineffective in the smoothing approximation stage.
In contrast to the existing literature,
our key observation is that the bandwidth should adapt not only to the sample size
but also to the number of knots to strike a balance between the bias and smoothness.

\item We establish the consistency and the asymptotic normality for the proposed estimators.
Specifically, we thoroughly analyze the size of the number of knots, allowing it to diverge with the sample size.
This distinguishes our theoretical analysis from previous studies,
such as \cite{zhong2022estimation}, \cite{yang2023estimation} and \cite{das2016fast},
which are restricted to finite-knot scenarios.
However, it poses challenges in determining the vanishing rate of the bandwidth.
To address this issue, we derive a approximate inequality between the smoothed and unsmoothed functions
to enhance our theoretical analysis.
Combining the penalized method with the approximate inequality,
we establish that our estimators can identify $M_n^*$ with probability tending to one,
and the unknown parameters can be estimated as efficiently as when $M_n^*$ is known.

\item We develop a test procedure to verify the existence of knots in the LSIR model,
and establish its asymptotic distributions under the null and local alternative hypotheses.
Our findings demonstrate that the power of the test is affected by both the signal's magnitude and the number of knots.
In particular, the power converges to one if there is at least one
$\alpha_m~(1\le m\le M_n^*)$ that diverges with the sample size $n,$ or if the number of knots tends to infinity as $n$ approaches infinity.
 This advancement in our results surpasses the previous works of \cite{hansen2017regression} and \cite{zhong2022estimation},
where the power is only determined by the magnitude of the signal.
\end{itemize}

The paper is structured as follows. Section \ref{section:Estimate} presents the LSIR model and estimation approach.
In Section \ref{section:theorem}, we derive the oracle properties of the proposed estimators.
Section \ref{section:testing} introduces a procedure for testing the existence of knots.
In Section \ref{section:simulations}, we report the results of simulation studies conducted to evaluate the performance of the proposed method.
Section \ref{section:real:data} illustrates the application of our method on two real datasets.
Section \ref{section:discussion} gives concluding remarks. The technical details are included in the online Supplementary Material.

\section{Estimation methods}\label{section:Estimate}

In this section, we develop an estimation procedure for the parameters presented in model \eqref{SIR-model}.
To ensure the identification of $\alpha_m, \tau_m$ and $\beta$,
we assume that the components of ${X_i}$ are continuously distributed random variables.
Furthermore, we assume that the support of ${X_i}$ does not lie within any proper linear subspace of $\mathbb{R}^{d_1}$.
In addition, we set the first component $\beta_1$ of ${\beta}$ to be one.
Such identification conditions  are commonly required under single-index models (\citealp{ichimura1993semiparametric}).
Then, we write ${\beta}=(1,{\beta}_{(-1)}^\top)^\top$ and ${X_i}=(X_{i1},{\widetilde X_{i}}^\top)^\top$,
where ${\beta}_{(-1)}=(\beta_2,\ldots,\beta_{d_1})^\top$ and ${\widetilde X}_i=(X_{i2},\ldots,X_{id_1})^\top$.
Let $\widetilde Z_i=(1,Z_i^\top)^\top$, $\eta=(\gamma_0,\gamma^\top)^\top$,
 ${\alpha}(M)=(\alpha_0, {\alpha}_{(-0)}(M)^\top)^\top$ with ${\alpha}_{(-0)}(M)=(\alpha_1,\dots,\alpha_{M})^\top$
and ${\tau}(M)=(\tau_1,\dots,\tau_{M})^\top.$
Define ${\theta}(M)=({\alpha}_{(-0)}(M)^\top,{\tau}(M)^\top,\alpha_0,{\beta}_{(-1)}^\top,\eta^\top)^\top$.
Let the superscript ``$*$" denote the true parameter values under which the data are generated and subscript ``$o$"  denote the oracle case,
namely the number of knots is known. Define  $M_n^*$ as the true number of knots,
and write ${\theta}_o={\theta}(M_n^*)$,
${\alpha}_{o}^*={\alpha}^*(M_n^*),$ ${\tau}_{o}^*={\tau}^*(M_n^*)$ and ${\theta}_{o}^*={\theta}^*(M_n^*).$

The observations consist of $n$ independent and identically distributed samples from model \eqref{SIR-model},
denoted by $\{(X_i, Z_i, Y_i): 1\le i\le n\}$.
When $M_n^*$ is known, we can obtain an estimate of ${\theta}_o^*$ using the least squares method, that is,
\begin{align}\label{Oracle:est}
{\widehat\theta}_o=\text{arg}\min_{{\theta}_o}\frac{1}{2}\sum_{i=1}^n\Big\{Y_i-\widetilde{Z}_i^\top{\eta}-\alpha_0{X}_i^\top{\beta}-
\sum_{m=1}^{M_n^*}\alpha_mf({X}_i^\top{\beta},\tau_m)\Big\}^2.
\end{align}
We refer to ${\widehat\theta}_o$ as the oracle estimator because it is obtained when $M_n^*$ is known.
However, there are practical challenges in obtaining ${\widehat\theta}_o$.
First, it requires $M_n^*$ known in advance,  which is often impractical.
Second, the function $f(x,\tau)$ is not differentiable at point $\tau,$
 making the minimization of \eqref{Oracle:est} time-consuming.

To overcome the aforementioned challenges, we develop a novel method to simultaneously estimate ${\theta}_o^*$ and $M_n^*$,
which integrates penalized methods with convolution techniques.
Specifically, we approximate $f(x,\tau)$ by
\begin{align*}
q_n(x,\tau)=\int_{-\infty}^{+\infty}f(v,\tau)\mathcal{K}_{\delta_n}(v-x)dv,~~~x\in\mathbb{R},
\end{align*}
where $\mathcal{K}_{\delta_n}(x)=\delta_n^{-1}\mathcal{K}(x/\delta_n)$ with
$\mathcal{K}(x)$ being a kernel function and $\delta_n$ being the smoothing parameter.
Convolution plays a role of random smoothing
in the sense that $q_n(x,\tau)=\mathbb{E}[f(x+Q\delta_n,\tau)]$, where $Q$ denotes a random variable with density function $\mathcal{K}(x)$.
To better understand this smoothing mechanism, we compute the smoothed function $q_n(x,\tau)$
explicitly for several widely used kernel functions.
In what follows, we write $q_n(x)=q_n(x,0)$ for short.
Please see Figure \ref{figure:different-qn} for a visualization of convolution smoothing methods.

\begin{figure}
\centering
\subfigure{\includegraphics[width=0.45\textwidth]{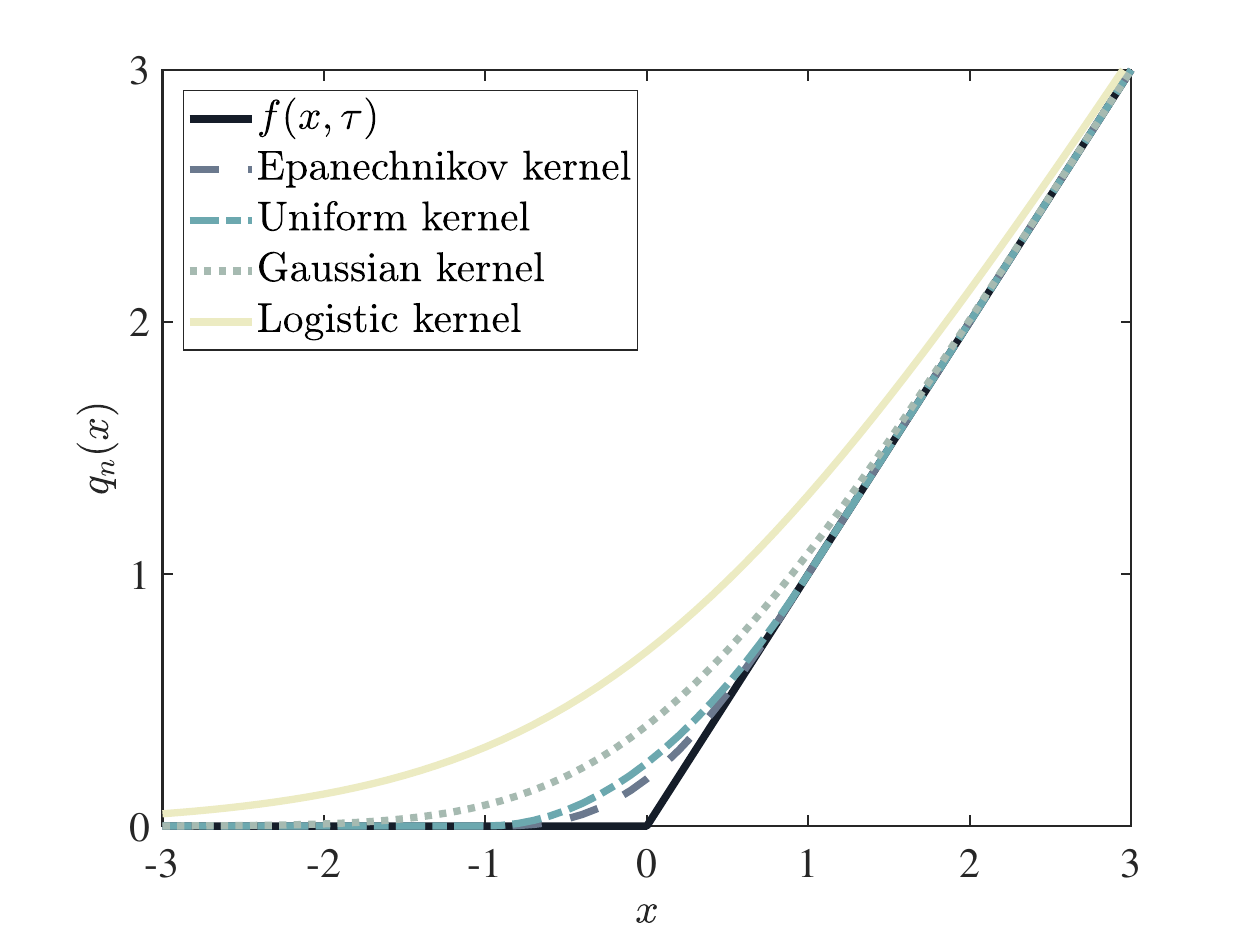}}
\subfigure{\includegraphics[width=0.45\textwidth]{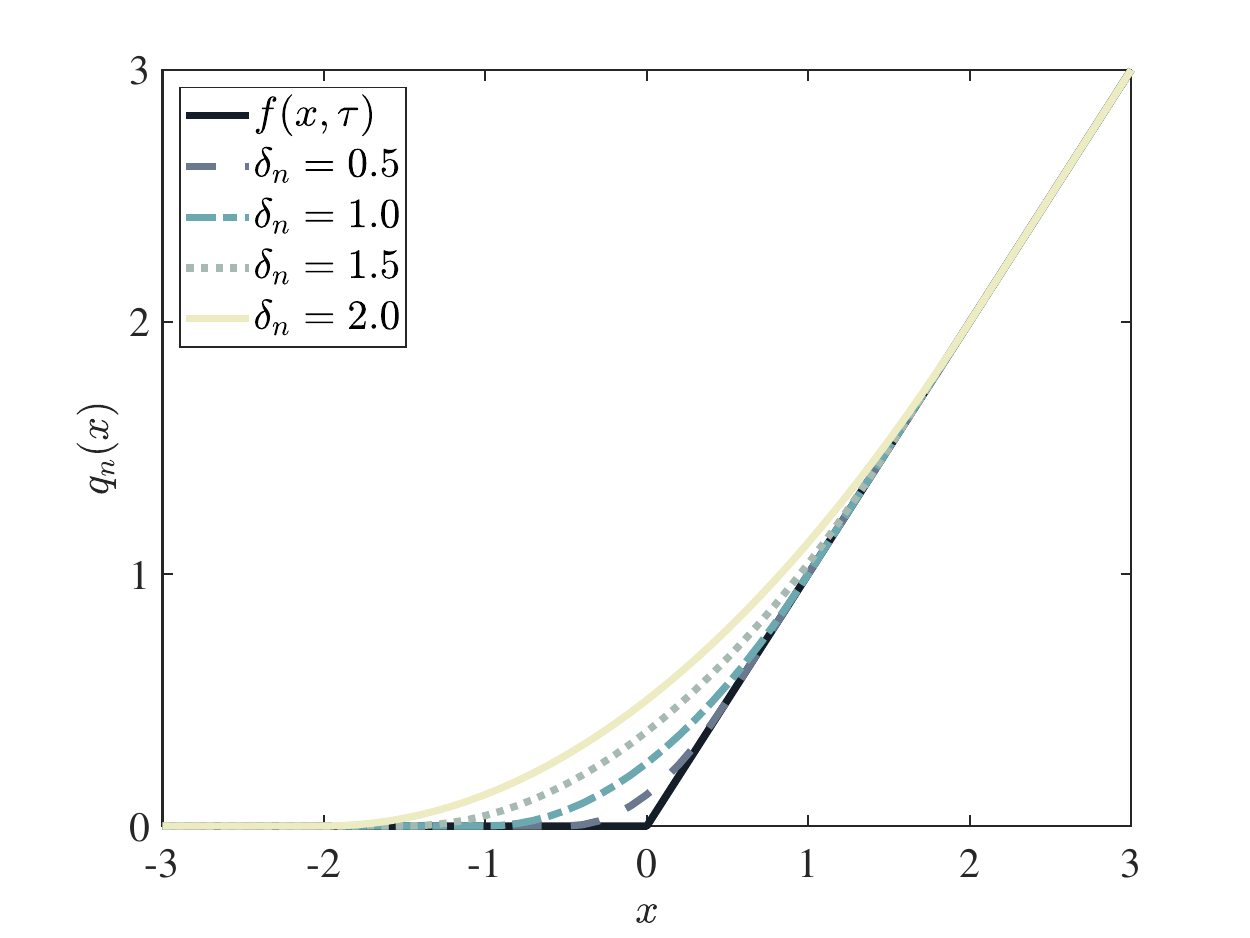}}
\caption{{\small Illustration of the smoothed approximations of the function $f(x,\tau)$.
The left panel displays the smoothed functions $q_n(x)$ using the Uniform kernel, Logistic kernel, Gaussian kernel, and Epanechnikov kernel with $\delta_n=1$ and $\tau=0$.
The right panel showcases the approximation function based on the Uniform kernel for various values of $\delta_n$.}}\label{figure:different-qn}
\end{figure}

\begin{itemize}
\item [1.] (Uniform kernel) Let $\mathcal{K}(x)=(1/2)I(|x| \le 1)$,
which is the density function of the uniform distribution on $[-1,1]$.
Then, we obtain $q_n(x)=(x+\delta_n)^2I(|x|\le \delta_n)/(4\delta_n) +xI(x>\delta_n)$,
which is exactly the  the smoothed function employed to approximate $f(x,\tau)$ in \cite{das2016fast}.

\item [2.] (Epanechnikov kernel) When $\mathcal{K}(x)=(3/4)(1-x^2)I(|x|\le 1)$,
the resulting smoothed function is $q_n(x)=(-x^4/\delta_n^3+6x^2/\delta_n+8x+3\delta_n)I(|x|\le \delta_n)/16+xI(x>\delta_n)$.

\item[3.] (Logistic kernel) In the case of the logistic kernel $\mathcal{K}(x)=e^{-x}/(1+e^{-x})^2$,
the resulting smoothed function is $q_n(x)=\delta_n\ln(1+e^{x/\delta_n})$,
which is the softplus activation function
and has been widely used as a smooth activation function in the context of artificial neural networks \citep{glorot2011deep}.
\item [4.] (Gaussian kernel) Let  $\phi(x)$ and $\Phi(x)$ denote the density function
and the cumulative distribution function of the standard normal distribution, respectively.
When $\mathcal{K}(x)=\phi(x)$, $q_n(x)=x\Phi(x/\delta_n)+\delta_n\phi(x/\delta_n)$.
\end{itemize}

The parameter $\delta_n$ controls the approximation level of $q_n(x,\tau)$,
and choosing an appropriate $\delta_n$ is vital for accurately estimating unknown parameters.
As to the approximation level $\delta_n$,
it decreases with $n$ so that the bias of the estimators diminishes to zero as $n\rightarrow \infty$.
However, minimizing \eqref{Oracle:est} becomes time-consuming with a very small $\delta_n$.
Thus, there is a trade-off between the bias of estimations and the computational cost.
We delve into the decreasing rate of $\delta_n$ in more detail in Section \ref{section:theorem}.


Next, we consider to estimate $M_n^*$. Intuitively, adding more knots allows for greater flexibility, but it also increases the risk of overfitting the data.
However, too few knots can result in an underfitting model that fails to capture the underlying relationships adequately.
Therefore, there exists a balance between model flexibility and complexity.
To be more specific, we consider a simple example of model \eqref{SIR-model}:
\begin{align}\label{Example}
 Y_i=& \widetilde{Z}_i^\top{\eta}+\alpha_0{X}_i^\top{\beta}+\alpha_1f({X}_i^\top{\beta},\tau_1)+\alpha_2f({X}_i^\top{\beta},\tau_2)+\epsilon_i\nonumber\\
=& \widetilde{Z}_i^\top{\eta}+\alpha_0{X}_i^\top{\beta}I({X}_i^\top{\beta}\leq \tau_1)+\big\{(\alpha_0+\alpha_1){X}_i^\top{\beta}-\alpha_1\tau_1\big\}I(\tau_1< {X}_i^\top{\beta}\leq\tau_2)\nonumber \\
 &+\big\{(\alpha_0+\alpha_1+\alpha_2){X}_i^\top{\beta}-(\alpha_1\tau_1+\alpha_2\tau_2)\big\}I({X}_i^\top{\beta}>\tau_2)+\epsilon_i \nonumber\\
=& \widetilde{Z}_i^\top{\eta}+\alpha_0{X}_i^\top{\beta}I({X}_i^\top{\beta}\leq \tau_1)+\big\{(\alpha_0+\alpha_1){X}_i^\top{\beta}-\alpha_1\tau_1\big\}I({X}_i^\top{\beta}>\tau_1)+\epsilon_i,
\end{align}
where the last equality holds if $\alpha_2=0$.
We see that if $\tau_2$ is not a knot,  the slope of ${X}_i^\top{\beta}$ on $\tau_1< {X}_i^\top{\beta}\leq\tau_2$ must be the same as that on ${X}_i^\top{\beta}> \tau_2$,
which implies that $\alpha_2=0.$
However, if $\alpha_2=0,$ then $\tau_2$ is unidentifiable.
But the last equality in \eqref{Example} still holds by setting $\tau_2=\tau_{\infty},$
where $\tau_{\infty}$ is chosen such that $I({X}_i^\top{\beta}>\tau_\infty)=0$ almost surely.
It is feasible in practice since we can set, for example, $\tau_{\infty}=\max_{1\le i\le n}|{X}_i^\top{\beta}|+1$.
These facts allow us to consider a penalized method to estimate $M_n^*.$
Specifically, let $M_n$ be a prespecified sequence that can diverge with $n$.
We propose the following penalized smoothing least squares method to simultaneously estimate ${\theta}_o^*$ and $M_n^*:$
\begin{align}\label{SP:est}
{\widehat\theta}_{\lambda_n}=\text{arg}\min_{{\theta}}\frac{1}{2}\sum_{i=1}^n\Big\{Y_i-\widetilde{Z}_i^\top{\eta}-\alpha_0{X}_i^\top{\beta}-
\sum_{m=1}^{M_n}\alpha_mq_n({X}_i^\top{\beta},\tau_m)\Big\}^2+n\sum_{m=1}^{M_n} p_{\lambda_n,t}(|\alpha_m|),
\end{align}
where ${\theta}={\theta}(M_n)$ and ${\widehat\theta}_{\lambda_n}=({\widehat\alpha}_{(-0),\lambda_n}^\top,{\widehat\tau}_{\lambda_n}^\top,
\widehat\alpha_{0,\lambda_n},{\widehat\beta}_{(-1),\lambda_n}^\top,{\widehat\eta}_{\lambda_n}^\top)^\top$
with $\widehat\tau_{m,\lambda_n}=\tau_{\infty}$ if $\widehat\alpha_{m,\lambda_n}=0.$
Here, $\widehat\tau_{m,\lambda_n}$ and $\widehat\alpha_{m,\lambda_n}$ are the $m$th component of ${\widehat\tau}_{\lambda_n}$ and ${\widehat\alpha}_{(-0),\lambda_n}.$
In addition, $p_{\lambda_n,t}(u)$ denotes a penalty function, $\lambda_n$ is a tuning parameter,
and $t$ is a parameter that controls the concavity of the penalty function.
The magnitude of $\lambda_n$ controls the complexity of the model.
A larger value of $\lambda_n$ indicates heavier shrinkage for ${\alpha}.$

We develop an iterative algorithm to obtain  $\widehat \theta_{\lambda_n}$.
To conserve space, we have provided the detailed information in the Supplementary Material.
Let $\widehat{\mathcal{S}}_{\lambda_n}=\{m:\widehat \alpha_{m,\lambda_n}\neq 0, 1\le m\le M_n\}$.
Denote $\widehat M_{n,\lambda_n}=\text{Card}(\widehat{\mathcal{S}}_{\lambda_n})$  as an estimator of $M_n^*,$
where $\text{Card}(A)$ denotes the cardinality of any set $A.$
In what follows, we omit the dependence of ${\widehat\theta}_{\lambda_n}, \widehat{\mathcal{S}}_{\lambda_n}$
and $\widehat M_{n,\lambda_n}$ on the tuning parameter $\lambda_n$  for the sake of brevity.

In this study, our focus is primarily on
the smoothly clipped absolute deviation (SCAD) penalty and the minimax concave penalty (MCP).
The SCAD penalty is introduced in \cite{fan2001variable} and defined by
\begin{align*}
    p_{\lambda_n,t}(u)=\lambda_n\int_0^{|u|}\min\big\{1,(t-x/\lambda_n)_+/(t-1)\big\}dx,~~t>2.
\end{align*}
The MCP penalty is introduced in \cite{zhang2010nearly} and defined by
\begin{align*}
    p_{\lambda_n,t}(u)=\lambda_n\int_0^{|u|}\big(1-x/(t\lambda_n)\big)_+dx,~~t>1.
\end{align*}
Here, $(x)_+=\max\{x,0\}$. Following \cite{fan2001variable} and \cite{zhang2010nearly}, we treat $t$ as a fixed constant.

\section{Theoretical results}\label{section:theorem}
In this section, we establish the consistency and asymptotic normality of the proposed estimators.
The proofs are given in the Supplementary Material.
We consider the following conditions.

\begin{condition}  \label{C1} There are $M_n^*$ distinct knots that satisfy $\tau_{1}^*<\tau_{2}^*<...<\tau_{M_n^*}^*$
and $\mathbb{P}(\tau_m^*<{X}_i^\top{\beta}^*\le \tau_{m+1}^*)>0$ for $m=0,\dots,M_n^*$,
with $\tau_0^*=-\infty$ and $\tau_{M_n^*+1}^*=\infty$.
\end{condition}

\begin{condition}  \label{C2}
There exist some positive constants $\kappa_0$ and $\kappa_1$ such that $|\alpha_m|<\kappa_0$
and $|\sum_{j=0}^m\alpha_{j}^*|<\kappa_1$ for $0\le m\le M_n^*.$
Additionally, $\min_{1\leq j\leq M_n^*}|\alpha_{j}^*|/\lambda_n\to\infty.$
\end{condition}

\begin{condition}  \label{C3}
$\mathbb{E}(X_{ij}^4|{X}_i^\top{\beta}^*)<\infty$ and $\mathbb{E}(Z_{ij}^4|{X}_i^\top{\beta}^*)<\infty.$
\end{condition}

\begin{condition}  \label{C4} $\mathbb{E}(\epsilon_i|{X}_i,{Z}_i)=0$,
$\mathbb{E}(\epsilon_i^2|{X}_i,{Z}_i)=\sigma^2<\infty$
and $\mathbb{E}(\epsilon_i^4|{X}_i,{Z}_i)<\infty$.
\end{condition}

\begin{condition}  \label{C5}
The kernel function $\mathcal{K}(x)$ has bounded support and satisfies that $\mathcal{K}(x)=\mathcal{K}(-x)$
and $\int \mathcal{K}(x)dx=1$.
\end{condition}

\begin{condition}  \label{C6} (i) $p_{\lambda_n,t}(u)$ is a symmetric function of $u$,
and it is nondecreasing and concave in $u$ for $u\in[0,\infty)$;
(ii) $p_{\lambda_n,t}(u)$ is differentiable in $u\in (0,\infty)$ with $\lambda_n^{-1}p'_{\lambda_n,t}(0+)>0;$
and (iii) there exist some positive constants $\kappa_2$ and $\kappa_3$ such that
$|p''_{\lambda_n,t}(u_1)-p''_{\lambda_n,t}(u_2)|\leq \kappa_2|u_1-u_2|$
for any $u_1$ and $u_2>\kappa_3\lambda_n.$
Here, $p'_{\lambda_n,t}(u)$ and $p''_{\lambda_n,t}(u)$ respectively represent the first and second derivatives of the penalty function $p_{\lambda_n,t}(u)$.
\end{condition}

\begin{condition}  \label{C7}
$\lambda_n \to 0$ and $\sqrt{n\lambda_n^2/s_n} \to \infty$ as $n\to\infty$,
where $s_n=1+d_1+d_2+2M_n.$
\end{condition}

Condition \ref{C1} is mild and assumes that the knots are ordered and distinct.
The first part of condition \ref{C2} imposes a bound on the slope of ${X}_i^\top{\beta}^*$ within each interval $(\tau_{m}^*,\tau_{m+1}^*]$ for $0\le m\le M_n^*$.
This condition is necessary to control the growth of the estimators.
The second part of condition \ref{C2}  places a requirement on the size of the non-zero coefficients $\alpha_m^*$.
It ensures that the non-zero coefficients do not converge to zero too rapidly, which is necessary for achieving the oracle property.
Intuitively, if some non-zero coefficients converge to 0 too fast,
it becomes challenging to estimate them accurately.
This condition relaxes the requirement in \cite{das2016fast},
where $|\alpha_m^*|$ is bounded below by a constant independent of the sample size $n$ for $0\le m\le M_n^*$.
Conditions \ref{C3} and \ref{C4} are standard assumptions in linear regression models.

Condition \ref{C5} specifies that $\mathcal{K}(x)$ is a symmetric density function,
and is satisfied by many kernel functions, such as the uniform kernel and the Epanechnikov kernel.
Under condition \ref{C5}, the convolution gives a smooth approximation to the function $f(x,\tau)$.
To see this, we define the support of $\mathcal{K}(x)$ as $[-1,1]$ and set $\tau=0$.
A direct calculation yields $q_n(x)=\int_{-1}^{x/\delta_n} (x-u\delta_n)\mathcal{K}(u)du$ if $|x|\le \delta_n$
and $q_n(x)=f(x,0)$ otherwise. Moreover, $q_n'(x)=\int_{-1}^{x/\delta_n}\mathcal{K}(u)du$ and $q_n''(x)=\mathcal{K}(x/\delta_n)/\delta_n$.
These imply that $q_n(x)$ is Lipschitz continuous, that is, $|q_n(x_1)-q_n(x_2)|\le |x_1-x_2|$ for any $x_1, x_2\in \mathbb{R}$.
Furthermore, it can be shown that $|q_n(x) - f(x,0)| \le O(\delta_n)$.
Therefore, a smaller bandwidth results in a more accurate approximation between $q_n(x)$ and $f(x,0)$.
This approximate inequality is crucial for our theoretical analysis.
Note that in condition \ref{C5}, we assume the boundedness of the support of $\mathcal{K}(x)$,
which is primarily made to simplify the proof of subsequent theorems.

Condition \ref{C6} describes the properties of a folded-concave penalty function  (\citealp{fan2004nonconcave}).
It consists of three parts.
The first part states that the penalty function is symmetric, non-decreasing, and concave in its argument.
The second part requires differentiability of the penalty function with a positive derivative at zero.
This condition ensures that the penalty function is singular at the origin,
leading to sparsity in the estimated coefficients. The third part imposes a smoothness condition on the penalty function,
controlling the  change rate of its second derivative.
Many popular folded-concave penalty functions, such as SCAD and MCP, satisfy condition \ref{C6}.

Condition \ref{C7} is a mild condition that determines the vanishing rate of the tuning parameter $\lambda_n$.
These conditions collectively ensure the consistency and asymptotic normality of the proposed estimators,
providing theoretical guarantees for their performance as the sample size increases.

Let $a_n=\max_{1\leq m\leq M_n}\big\{p'_{\lambda_n,t}(|\alpha_{m}^*|):\alpha_{m}^*\neq0\big\}$
and $b_n=\max_{1\leq m\leq M_n}\big\{p''_{\lambda_n,t}(|\alpha_{m}^*|):\alpha_{m}^*\neq0\big\}.$
Define $\mathcal{S}^*=\{m: \alpha_m^*\neq 0,~ 1\le m\le M_n\}.$
The following theorem establishes the existence and consistency of the penalized smoothing least squares estimator.
\begin{theorem}\label{Theorem:Consistenncy}
Suppose that conditions \ref{C1}-\ref{C7} hold.
If $b_{n}\rightarrow 0$, $\sqrt{n}s_n^2\delta_n\to0$ and $s_n^3/(n\delta_n)\to0$ as $ n \rightarrow \infty$,
then there exists a local minimizer  ${\widehat\theta}$ defined in \eqref{SP:est} such that
\begin{itemize}
    \item[](i)
    $\mathbb{P}(\widehat{\mathcal{S}}=\mathcal{S}^*)\rightarrow 1$ as $n\rightarrow\infty$.
    \item[](ii) $\|{\widehat\theta}_1-{\theta}_o^*\|=O_p\big(s_n^{1/2}\big(n^{-1/2}+a_n\big)\big)$, where  ${\widehat\theta}_1=\big(\widehat{{\alpha}}_{(-0)}(M_n^*)^\top,\widehat{{\tau}}(M_n^*)^\top,\widehat\alpha_0,{\widehat\beta}_{(-1)}^\top,{\widehat\eta}^\top\big)^\top$ is the subvector of ${\widehat\theta}$,
    and $\|\cdot\|$ denotes the Euclidean norm.
\end{itemize}
\end{theorem}
The first part of Theorem \ref{Theorem:Consistenncy} states that the proposed estimators can identify the true model with probability tending to one.
It also implies that $\widehat M_n$ is a consistent estimator of $M_n^*$.
The second part demonstrates that if $a_n$ is of order $n^{-1/2},$ then
there exists a root-$(n/s_n)$-consistent estimator of ${\theta}^*$.
If the penalty function is SCAD or MCP and condition \ref{C2} holds,
then $a_n=0$ when $n$ is sufficiently large.
This implies that the scaling factor becomes negligible as the sample size increases,
indicating that the proposed estimators achieves the optimal convergence rate.

Theorem \ref{Theorem:Consistenncy} requires $\sqrt{n}s_n^2\delta_n\to0$ and $s_n^3/(n\delta_n)\to0$ as $n\rightarrow \infty$.
If $d_1$ and $d_2$ are fixed, then we have $s_n=O(M_n)$.
Therefore,  it implies that the number of knots $M_n$ can increase with the sample size,
but the increase rate can not be too fast in order to ensure the optimal convergence rate.
Moreover, the condition $\sqrt{n}s_n^2\delta_n\to0$ indicates that
a small $\delta_n$ is needed to ensure that the bias of the estimators becomes negligible as $s_n$ increases.
In contrast, $s_n^3/(n\delta_n)\to0$ indicates that
a large $\delta_n$ is necessary for maintaining the smoothness of the loss function.
It further indicates that the smoothing parameter $\delta_n\rightarrow 0$ should adapt to the sample size and the number of knots
to achieve a balanced trade-off between the bias and smoothness.
Intuitively, a large value of $M_n$ tends to make the locations $\tau_m$ and $\tau_{m+1}$ of knots close to each other.
Therefore, a slightly smaller value of $\delta_n$ may be more popular in order to better distinguish these knots.
This finding is different from the results of \cite{das2016fast},
where $\delta_n$ is free of the number of knots.


Next, we establish the normality of the proposed estimator.
Define
\begin{align*}
    {\Sigma}_{\lambda_n}({\theta}_o)=&\text{diag}\big\{p''_{\lambda_n,t}(|\alpha_{1}|),\ldots,p''_{\lambda_n,t}(|\alpha_{M_n^*}|),0,\dots,0\big\}, \\
     {B}({\theta}_o)=&\big(p'_{\lambda_n,t}(|\alpha_{1}|)\text{sgn}(\alpha_{1}),\ldots,p'_{\lambda_n,t}(|\alpha_{M_n^*}|)\text{sgn}(\alpha_{M_n^*}),0,\dots,0\big)^\top,\\
    \mbox{and}~~~{V}({\theta}_o)=&\mathbb{E}\big\{{H}_i({\theta}_o){H}_i({\theta}_o)^\top\big\},
\end{align*}
where
\begin{align*}
{H}_i({\theta}_o)=&\Big(f({X}_i^\top{\beta},\tau_{1}),\ldots,f({X}_i^\top{\beta},\tau_{M_n^*}),-\alpha_{1}I{({X}_i^\top{\beta}>\tau_{1})},\dots,-\alpha_{M_n^*}I{({X}_i^\top{\beta}>\tau_{M_n^*})},\\
    &\hspace{2in}{X}_i^\top{\beta},\Big[\alpha_{0}+\sum_{m=1}^{M_n^*}\alpha_{m} I({X}_i^\top{\beta}>\tau_{m})\Big]{\widetilde X_i}^\top,\widetilde{Z}_i^\top\Big)^\top.
\end{align*}
Furthermore, we  write ${\Sigma}_{\lambda_n}^*={\Sigma}_{\lambda_n}({\theta}_o^*),$ ${B}^*={B}({\theta}_o^*)$ and ${V}_*={V}({\theta}_o^*)$ for short.
\begin{theorem}\label{Theorem:Oracle}
Suppose that conditions \ref{C1}-\ref{C7} hold. If $\sqrt{ns_n^5}\delta_n\to0$ and  $s_n^4/(n\delta_n)\to0$ as $n\rightarrow \infty$,
then for any $q\times s_n^*$ matrix $A_n$, we have
$$\sqrt{n}{A}_n{V}_*^{-1/2}\big({V}_*+{\Sigma}_{\lambda_n}^*\big)\big\{({\widehat \theta}_1-{\theta}_o^*)+\big({V}_*+{\Sigma}_{\lambda_n}^*\big)^{-1}{B}^*\big\}
        \xrightarrow{~d~} \mathcal{N}({0}_q,\sigma^2{G}),$$
        where $ \xrightarrow{~d~}$ means convergence in distribution, and $G$ is a $q \times q$ positive definite matrix such that ${A}_n {A}_n^\top\to {G}$ as $n\rightarrow\infty$,
        with $\|{G}\|_\mathrm{F}=O(1)$.
        Here, $s_n^*=1+d_1+d_2+2M_n^*,$ $q<s_n^*$ is fixed, and $\|\cdot\|_\mathrm{F}$ denotes the Frobenius norm of a matrix.
\end{theorem}

For the SCAD and MCP, ${\Sigma}_{\lambda_n}^*$ and ${B}^*$ become zero with a sufficiently large $n$.
Theorem \ref{Theorem:Oracle} can be simplified to state that
$\sqrt{n}{A}_n{V}_*^{1/2}({\widehat \theta}_1-{\theta}_o^*)\xrightarrow{~d~} \mathcal{N}({0}_q,\sigma^2{G}),$
which has the same efficiency as the estimator ${\widehat \theta}_o$ based on $M_n^*$ known in advance.
In this sense, our estimator achieves the oracle property.

In practice, the covariance matrix of ${\widehat\theta}_1$ needs to be estimated.
Following the conventional technique,
we consider the following sandwich formula to estimate the covariance of ${\widehat \theta}_1$:
\begin{align*}
   {\widehat\Xi}_n=n^{-1}\widehat\sigma^2\big\{{V}_n({\widehat\theta}_{1})+{\Sigma}_{\lambda_n}({\widehat\theta}_{1}) \big\}^{-1}  {V}_n({\widehat\theta}_{1})
   \big\{{V}_n({\widehat\theta}_{1})+{\Sigma}_{\lambda_n}({\widehat\theta}_{1}) \big\}^{-1},
\end{align*}
where
\[
\widehat\sigma^2=\frac{1}{n} \sum_{i=1}^n\Big\{Y_i-\widetilde{Z}_i^\top{\widehat\eta}-\widehat\alpha_{0}{X}_i^\top{\widehat\beta}-\sum_{m=1}^{\widehat M_n}\widehat\alpha_{m} q_n({X}_i^\top{\widehat\beta},\widehat\tau_{m})\Big\}^2,
\]
and ${V}_n({\widehat\theta}_{1})=n^{-1}\sum_{i=1}^{n}{H}_{ni}({\widehat\theta}_{1}){H}_{ni}({\widehat\theta}_{1})^\top$ with
\begin{align*}
    {H}_{ni}({\widehat\theta}_{1})=&\Big(q_n({X}_i^\top{\widehat\beta},\widehat\tau_{1}),\ldots,q_n({X}_i^\top{\widehat\beta},\widehat\tau_{\widehat M_n}),\widehat\alpha_{1}\frac{\partial}{\partial \tau_1}q_n({X}_i^\top{\widehat\beta},\widehat\tau_{1}),\ldots,
    \widehat\alpha_{\widehat M_n}\frac{\partial}{\partial \tau_{\widehat M_n}}q_n({X}_i^\top{\widehat\beta},\widehat\tau_{\widehat M_n}),\\
     &\hspace{1.95in}{X}_i^\top{\widehat\beta},\Big[\widehat\alpha_{0}{\widetilde X}_i+\sum_{m=1}^{\widehat M_n}\widehat\alpha_{m}\frac{\partial}{\partial{\beta}_{(-1)}}q_n({X}_i^\top{\widehat\beta},\widehat\tau_{m})\Big]^\top, \widetilde{Z}_i^\top\Big)^\top.
\end{align*}
The following theorem establishes the consistency of the sandwich-type estimator.
\begin{theorem}\label{ConS}
Suppose that conditions \ref{C1}-\ref{C7} hold. If $\sqrt{ns_n^5}\delta_n\to0$ and $s_n^4/(n\delta_n)\to0$ as $n\rightarrow \infty$, then we have
    \begin{align*}
        n[{A}_n{\widehat\Xi}_n{A}_n^{\top}-{A}_n{\Xi} {A}_n^{\top}]=o_p(1),
    \end{align*}
    where ${A}_n$ is defined in Theorem \ref{Theorem:Oracle}, and ${\Xi}=n^{-1}\sigma^2\big({V}_*+{\Sigma}_{\lambda_n}^* \big)^{-1}  {V}_* \big({V}_*+{\Sigma}_{\lambda_n}^* \big)^{-1}.$
\end{theorem}

The consistent result offers a way to construct a confidence interval for the estimates of parameters.
For example, let $e_j\in \mathbb{R}^{\widehat M_n}$ be a row vector with one in the $j$th component and zero otherwise,
and set $A_n=e_{j}\{{V}_n({\widehat\theta}_{1})+{\Sigma}_{\lambda_n}({\widehat\theta}_{1})\}^{-1}{V}_n({\widehat\theta}_{1})^{1/2}/\widehat\sigma$ with $q=1$.
Then, we can construct a confidence interval with level $\varsigma$ for $\theta_{oj}^*$ by
$[\widehat\theta_{1j}-z_{\varsigma/2}\sqrt{e_j\widehat\Xi e_j^\top},\widehat\theta_{1j}+z_{\varsigma/2}\sqrt{e_j\widehat\Xi e_j^\top}]$,
where $\theta_{oj}^*$ and $\widehat \theta_{1j}$ are the $j$th components of $\theta_o^*$ and $\widehat\theta_1$, respectively,
and $z_{\varsigma/2}$ denotes the upper $\varsigma/2$-quantile of the standard normal distribution.
In addition, it allows for simultaneous testing of whether a group of variables are significant by using a specific matrix $A_n.$
For example, consider the following hypothesis:
\[
H_0: \beta_{(-1)}^*=0 ~~~\text{v.s.}~~~ H_1: \beta_{(-1)}^*\neq 0.
\]
If the null hypothesis holds, then model \eqref{SIR-model} reduces to linear spline regression models \citep{card2012nonlinear, hansen2017regression, li2011bent, muggeo2003estimating}.
Hence, it is  interesting to test whether $\beta_{(-1)}^*$ is zero.
To aim  this, we set $A_n=e_{\beta}\{{V}_n({\widehat\theta}_{1})+{\Sigma}_{\lambda_n}({\widehat\theta}_{1})\}^{-1}{V}_n({\widehat\theta}_{1})^{1/2}/\widehat\sigma$ with $q=d_1-1$,
where $e_{\beta}=(e_{2(1+\widehat M_n)}^\top,\dots,e_{1+d_1+2\widehat M_n}^\top)^\top$.
For this, we define the following statistic
\[
\mathcal{W}= \widehat{\beta}_{(-1)}^\top [e_{\beta}\widehat \Xi e_{\beta}^\top]^{-1}  \widehat{\beta}_{(-1)}.
\]
By Theorems \ref{Theorem:Oracle} and \ref{ConS}, we can show that $\mathcal{W} \xrightarrow{~d~} \chi^2(d_1-1)$,
where $\chi^2(m)$ denotes a Chi-square distribution with $m$ degrees of freedom.

\section{Testing for knot effects}\label{section:testing}
In this section, we investigate the presence or existence of knots.
To this end, we focus on examining the following hypotheses:
\[
H_0: \alpha_{1}^*=\dots=\alpha_{M_n^*}^*=0~~~\text{v.s.}~~~ H_1: \alpha_{m}^*\neq 0~\text{for some}~m \in \{1,\dots,M_n^*\}.
\]
If the null hypothesis $H_0$ holds, then the LSIR model is identical to the traditional linear regression model, without any knots.
Conversely, there exists at least one knot in the LSIR model under the alternative hypothesis $H_1$.
Define
\begin{align*}
\psi(\tauo,\alpha_0,{\beta},{\eta})
=\frac{1}{\sqrt{n}}\sum_{i=1}^n q_n({X}_i^\top{\beta},\tauo)\big(Y_i-\widetilde{Z}_i^\top{\eta}-\alpha_0{X}_i^\top{\beta}\big).
\end{align*}
Under $H_0$, $\tau_m^*$ vanishes and becomes unidentified.
But it still holds that $\mathbb{E}\{\psi(\tauo,\alpha_0^*,{\beta}^*,{\eta}^*)\}=0$ for any given $\tauo$.
This implies that the value of $\psi(\tauo,\alpha_0^*,\beta^*,\eta^*)$ is close to zero under $H_0$.
Additionally, we have $\mathbb{E}\{\psi(\tauo,\alpha_0^*,\beta^*,\eta^*)\}\neq 0$ for some $\tauo$ under $H_1$.
To address the unidentifiable issue of $\tauo$,
we consider a supremum-type test statistic for the hypothesis:
\begin{align*}
\mathcal{T}=\sup_{\tauo\in \Theta_{\tau}}\frac{\big\{n^{-1/2}\sum_{i=1}^n\widehat\psi_i(\tauo)\big\}^2}{\widehat\varrho(\tauo)},
\end{align*}
where $\Theta_{\tau}\subset \mathbb{R}$ is a compact set,
$\widehat\psi_i(\tauo)=q_n({X}_i^\top{\widehat\beta},\tauo)\big(Y_i-\widetilde{Z}_i^\top{\widehat \eta}- \widehat\alpha_0{X}_i^\top{\widehat \beta}\big)$,
$\widehat\varrho(\tauo)$ is a consistent estimator for the asymptotic variance of $\psi(\tauo,\widehat\alpha_0, {\widehat\beta}, {\widehat\eta})$ under $H_0$,
and $\widehat\alpha_0$, ${\widehat\beta}$ and ${\widehat\eta}$ are the estimates of $\alpha_0^*$, ${\beta}^*$ and ${\eta}^*$ under $H_0$.
Specifically, $\widehat\alpha_0,~{\widehat\beta}$ and ${\widehat\eta}$ can be obtained by
\begin{align}\label{H0:est}
(\widehat\alpha_0, {\widehat\beta}^\top, {\widehat\eta}^\top)^\top
=\text{arg}\min_{(\alpha_0, {\beta}^\top, {\eta}^\top)^\top} \frac{1}{2}\sum_{i=1}^n\Big(Y_i-\widetilde{Z}_i^\top{\eta}-\alpha_0{X}_i^\top{\beta}\Big)^2.
\end{align}
Define $\Psi(\varsigma)=I(\mathcal{T}>c(\varsigma))$,
where $c(\varsigma)$ is the critical value and $\varsigma$ denotes the significant level.
The test statistic $\mathcal{T}$ is close to zero under  $H_{0}$.
Therefore, we reject $H_{0}$ if and only if $\Psi(\varsigma)=1.$
Taking another perspective on the test statistic $\psi(\tauo,\alpha_0,{\beta},{\eta})$ is that
$\psi(\tau_m^*,\alpha_0^*,{\beta}^*,{\eta}^*)$ is a smoothing score function for $\alpha_m^*$ evaluated at $\alpha_1^*=\dots=\alpha_{M_n^*}^*=0$.
Thus, $\psi(\tau_m^*,\alpha_0^*,{\beta}^*,{\eta}^*)$ can be viewed as a score-type test statistic for the hypothesis.
Such type test statistics are also considered in different literature (e.g., \citealp{andrews2001testing,fan2017change,zhong2022estimation}).

To analyze the asymptotic behaviour of the test statistic $\mathcal{T}$, we introduce a sequence of local alternatives denoted as $H_{1n}$, where we assume that $\alpha_m^*=\varpi_m/n^{1/2}$ for some $m\in \{1,\ldots, M_n^*\}$. This means that under $H_{1n}$, there exists at least one knot, and the underlying signals are decreasing at a rate of $n^{-1/2}$.

We define ${D}(\tauo)=\mathbb{E}\{{\xi}_i f({X}_i^\top{\beta}^*,\tauo)\},$
$\psi_{*i}(\tauo)=\{f({X}_i^\top{\beta}^*,\tauo)-{D}(\tauo)^\top {\Omega}^{-1}{\xi}_i\}(Y_i-\widetilde{Z}_i^\top{\eta}^*-\alpha_0^*{X}_i^\top{\beta}^*),$
and $\varrho(\tauo)=\mathbb{E}\big\{\psi_{*i}(\tauo)^2\big\}$,
where ${\xi}_i=({X}_i^\top{\beta}^*,\alpha_{0}^*{\widetilde X}_i^\top,\widetilde{Z}_i^\top)^\top$ and ${\Omega}=\mathbb{E}({\xi}_i{\xi}_i^\top)$.

\begin{theorem}\label{Theorem:Test}
Suppose that conditions \ref{C1}-\ref{C5} hold.
If $\sqrt{n}\delta_n^2\to0$, $n\delta_n\to\infty$ and $s_n^2/n\to0$ as $n\rightarrow \infty$,
then, under either the null hypothesis $H_0$ or the local alternative $H_{1n}$, we have that $\mathcal{T}$
converges in distribution to $\sup_{\tauo\in\Theta_{\tau}} \mathcal{G}(\tauo)^2$,
where $\{\mathcal{G}(\tauo): \tauo\in\Theta_{\tau}\}$ is a Gaussian process
with mean function $\Delta(\tauo)=\mathbb{E}[\{f({X}_i^\top{\beta}^*,\tauo)-{D}(\tauo)^\top {\Omega}^{-1}{\xi_i}\}\sum\nolimits_{m=1}^{M_n^*}\varpi_mf({X}_i^\top{\beta},\tau_m^*)]/\sqrt{\varrho(\tauo)}$
and covariance function ${\Gamma}(\tau_1,\tau_2)=\mathbb{COV}\big(\psi_{*i}(\tau_1),\psi_{*i}(\tau_2)\big)/\sqrt{\varrho(\tau_1)\varrho(\tau_2)}$ for any given $\tau_1,\tau_2\in\Theta_{\tau}$.
\end{theorem}

Theorem \ref{Theorem:Test} states that under the null hypothesis, $\Psi(\varsigma)$ can achieve a level of significance of $\varsigma$.
Meanwhile,  Theorem \ref{Theorem:Test} also indicates that the power of the proposed test  statistic is essentially controlled by the signal-to-noise ratio term $\Delta(\tauo)$.
When all $\varpi_m~(1\le m\le M_n^*)$ in $\Delta(\tauo)$ converge to zero, the power diminishes to $\varsigma$.
In this case, the proposed test can not distinguish the null hypothesis from the local alternatives.
If there exists at least one $\varpi_m$ in $\Delta(\tauo)$ diverging with $n$
or $M_n^*\rightarrow \infty$ as $n\rightarrow\infty$,
the power converges to 1, which implies that the proposed method is consistent.

In practice, the critical value $c(\varsigma)$ can be obtained using a resampling approach.
For this purpose, we define $
    \widehat\psi_{*i}(\tauo)=\{f({X}_i^\top{\widehat\beta},\tauo)-{\widehat D}(\tauo)^\top {\widehat\Omega}^{-1}{\widehat\xi}_i\}(Y_i-\widetilde{Z}_i^\top{\widehat\eta}-\widehat\alpha_0{X}_i^\top{\widehat\beta}),$
where ${\widehat D}(\tauo)=n^{-1}\sum_{i=1}^n{\widehat\xi}_if({X}_i^\top{\widehat\beta},\tauo)$,
${\widehat\Omega}=n^{-1}\sum_{i=1}^n{\widehat\xi}_i{\widehat\xi}_i^\top,$
and ${\widehat\xi}_i=\big({X}_i^\top{\widehat\beta},\widehat\alpha_{0}{\widetilde X}_i^\top,\widetilde{Z}_i^\top
\big)^\top$.
Let $\widehat\varrho(\tauo)=n^{-1}\sum_{i=1}^n\widehat\psi_{*i}(\tauo)^2.$
Then, define
\[
\mathcal{T}^*=\sup_{\tauo\in\Theta_{\tau}} \frac{\big\{n^{-1/2}\sum_{i=1}^n\widetilde G_i\widehat\psi_{*i}(\tauo)\big\}^2}{\widehat\varrho(\tauo)},
\]
where $\widetilde G_{i}~(i=1,\dots,n)$ are independent standard normal variables which are independent of the observed data.
By repeatedly generating normal random samples $\widetilde G_i$, the distribution of $\mathcal{T}$ can be approximated by the conditional distribution of $\mathcal{T}^*$ given the observed data.
Then, the critical value $c(\varsigma)$ can be obtained from the upper $(1-\varsigma)$-percentile of the conditional distributions of $\mathcal{T}^*$.

\section{Simulation studies}\label{section:simulations}
In this section, we conduct simulation studies to evaluate the finite sample performance of the proposed estimators.
To select the tuning parameter $\lambda_n$, we employ  a Bayesian information criterion (BIC):
\begin{align*}
\text{BIC}(\lambda_n)=&\log\bigg[\frac{1}{n}\sum_{i=1}^n\Big\{Y_i-\widetilde{Z}_i^\top{\widehat\eta}-\widehat\alpha_{0}{X}_i^\top{\widehat\beta}-\sum_{m=1}^{\widehat M_n}\widehat\alpha_{m} f({X}_i^\top{\widehat\beta},\widehat\tau_{m})\Big\}^2 \bigg]\\
&~~+(2\widehat M_n+2+d_1+d_2)\times \frac{C_n\log (n)}{2n},
\end{align*}
where $C_n$ is a  predetermined constant.
When $C_n=1$, it reduces to the traditional BIC.
When $C_n=\log\{\log(n)\}$, it corresponds to the modified BIC. The modified BIC has been shown by  \cite{wang2009shrinkage} to consistently identify the true model,
when the dimension of the unknown parameters diverges with the sample size $n$.

\subsection{Consistency and normality}\label{section:6.1}
We generate the covariates ${X}_i$ and $Z_{i1}~(i=1,\dots,n)$ as follows.
First, we generate $\breve{{X}}_i=(\breve{X}_{i1},\breve{X}_{i2},\breve{X}_{i3})^\top$ independently from a multivariate normal distribution
with mean zero and covariance ${\Upsilon}=(\sigma_{sk}: 1\le s, k\le 3).$ We set $\sigma_{ss}=1$ and $\sigma_{sk}=0.5$ for $s\neq k.$
Then we define $X_{i1}=\breve{X}_{i1}$, $X_{i2}=3.5\{2\Phi(\breve{X}_{i2})-1\}$, and $Z_{i1}=\breve{X}_{i3}$.
The random errors $\epsilon_i~(i=1,\dots,n)$ are independently generated from three different distributions: the standard normal distribution ($\mathcal{N}(0,1)$),
a standardized Chi-square distribution with 2 degrees of freedom ($\text{Schi}^2(2)$), and a Student's $t$-distribution with 4 degrees of freedom ($t(4)$).
The standardized Chi-square distribution is skewed but still has a zero mean,
while the $t(4)$ distribution represents moderately heavy-tailed errors.
We set $\gamma_1^*=0.5$, and the responses $Y_i$ are then generated under the following three cases:

Case 1: $Y_i=Z_{i1}\gamma_1^*+\alpha_0^*{X}_i^\top{\beta}^*+\alpha_1^*f({X}_i^\top{\beta}^*,\tau_1^*)+\epsilon_i$,
where ${\beta}^*=(1,-1)^\top$, $\tau_1^*=0$, and ${\alpha}^*=(-1,1.5)^\top$.

 Case 2: $Y_i=Z_{i1}\gamma_1^*+\alpha_0^*{X}_i^\top{\beta}^*+\sum_{m=1}^2\alpha_m^*f({X}_i^\top{\beta}^*,\tau_m^*)+\epsilon_i$,
where ${\beta}^*=(1,-1)^\top$, ${\tau}^*=(-1,1)^\top$, and ${\alpha}^*=(1,-2,2)^\top$.

 Case 3: $Y_i=Z_{i1}\gamma_1^*+\alpha_0^*{X}_i^\top{\beta}^*+\sum_{m=1}^4\alpha_m^*f({X}_i^\top{\beta}^*,\tau_m^*)+\epsilon_i$,
where ${\beta}^*=(1,-2)^\top$, ${\tau}^*=(-4,-2,2,4)^\top$, and ${\alpha}^*=(-1,3,-2,-2,3)^\top$.

\begin{table}[h]\centering
{\footnotesize												
\caption{Percentages (\%) of correctly selecting $\widehat M_n=M_n^*$ for Cases 1-3.}\label{Table:Simulations:CS}					
\begin{tabular}{cccccrccrccr}				
\hline
& & & &\multicolumn{2}{c}{$\mathcal{N}(0,1)$} & & \multicolumn{2}{c}{$\text{Schi}^2(2)$} & & \multicolumn{2}{c}{$t(4)$}\\
\cline{5-6} \cline{8-9} \cline{11-12} 			
Case	&	$\nu$	&	$n$	&		&	$C_n=1$ 	&	$\log \log n$	&		&	1 	&	$\log \log n$	&		&	1 	&	$\log \log n$	\\
\hline																							
1	&	0.6 	&	1000 	&	SCAD	&	95.8	&	99.8	&		&	95.2	&	99.2	&		&	94.6	&	98.5	\\
	&		&		&	MCP	&	95.3	&	99.9	&		&	94.6	&	99.3	&		&	93.7	&	98.9	\\
	&		&	2000 	&	SCAD	&	97.4	&	99.9	&		&	97.0	&	99.9	&		&	97.0	&	99.1	\\
	&		&		&	MCP	&	96.8	&	99.9	&		&	96.7	&	99.8	&		&	96.4	&	99.4	\\
	&	0.8 	&	1000 	&	SCAD	&	96.1	&	99.9	&		&	95.1	&	99.0	&		&	94.7	&	98.5	\\
	&		&		&	MCP	&	95.6	&	99.9	&		&	94.9	&	99.3	&		&	93.9	&	99.1	\\
	&		&	2000 	&	SCAD	&	97.5	&	99.9	&		&	96.7	&	99.8	&		&	97.3	&	99.2	\\
	&		&		&	MCP	&	97.1	&	99.9	&		&	96.8	&	99.8	&		&	96.3	&	99.3	\\
\hline																							
2	&	0.6 	&	1000 	&	SCAD	&	96.4	&	99.7	&		&	95.6	&	99.0	&		&	94.5	&	98.1	\\
	&		&		&	MCP	&	96.0	&	99.6	&		&	93.8	&	99.0	&		&	93.1	&	99.0	\\
	&		&	2000 	&	SCAD	&	97.5	&	100	&		&	96.7	&	99.9	&		&	95.9	&	99.4	\\
	&		&		&	MCP	&	97.5	&	99.9	&		&	96.4	&	99.8	&		&	95.8	&	99.7	\\
	&	0.8 	&	1000 	&	SCAD	&	96.5	&	99.7	&		&	95.1	&	99.0	&		&	94.1	&	98.1	\\
	&		&		&	MCP	&	96.1	&	99.6	&		&	94.1	&	99.2	&		&	93.0	&	98.9	\\
	&		&	2000 	&	SCAD	&	97.2	&	99.9	&		&	96.6	&	99.9	&		&	95.9	&	99.4	\\
	&		&		&	MCP	&	97.3	&	99.9	&		&	96.2	&	99.9	&		&	96.0	&	99.7	\\
\hline																							
3	&	0.6 	&	1000 	&	SCAD	&	97.0	&	99.9	&		&	96.8	&	99.7	&		&	96.9	&	99.5	\\
	&		&		&	MCP	&	97.0	&	99.9	&		&	96.9	&	99.7	&		&	96.9	&	99.5	\\
	&		&	2000 	&	SCAD	&	98.7	&	100	&		&	99.1	&	99.9	&		&	97.9	&	100	\\
	&		&		&	MCP	&	98.7	&	100	&		&	99.1	&	99.9	&		&	97.9	&	100	\\
	&	0.8 	&	1000 	&	SCAD	&	97.0	&	99.9	&		&	97.1	&	99.8	&		&	97.4	&	99.5	\\
	&		&		&	MCP	&	97.1	&	99.9	&		&	97.1	&	99.8	&		&	97.4	&	99.5	\\
	&		&	2000 	&	SCAD	&	98.9	&	100	&		&	99.2	&	100	&		&	98.1	&	100	\\
	&		&		&	MCP	&	98.9	&	100	&		&	99.2	&	100	&		&	98.0	&	100	\\
\hline	
\end{tabular}}											
\end{table}

\begin{table}[h]\centering									
{\footnotesize				
\caption{Simulation results (multiplied by 100) for Case 1 with $C_n=\log\{\log (n)\}$ 
and $\epsilon_i\sim \mathcal{N}(0,1)$.}\label{Table:caseI:normal}												
\begin{tabular}{cccrccccrccccrcccc}												
\hline																	
& & &\multicolumn{4}{c}{Oracle}& &\multicolumn{4}{c}{SCAD}& &\multicolumn{4}{c}{MCP}\\				
\cline{4-7} \cline{9-12} \cline{14-17} 	
$n$ & $\nu$	&		&	Bias	&	SE	&	SD	&	CP	&		&	Bias	&	SE	&	SD	&	CP	&		&	Bias	&	SE	&	SD	&	CP	\\
\hline											
1000	&	0.6	&	$\alpha_0$	&	0.302 	&	4.74 	&	4.94 	&	94.4 	&		&	0.327 	&	4.74 	&	4.94 	&	94.3 	&		&	0.313 	&	4.74 	&	4.94 	&	94.1 	\\
	&		&	$\alpha_1$	&	-0.430 	&	7.82 	&	8.43 	&	92.6 	&		&	-0.487 	&	7.82 	&	8.42 	&	92.6 	&		&	-0.449 	&	7.82 	&	8.41 	&	92.5 	\\
	&		&	$\tau_1$	&	-0.244 	&	5.61 	&	5.90 	&	93.5 	&		&	-0.263 	&	5.61 	&	5.86 	&	93.5 	&		&	-0.242 	&	5.61 	&	5.85 	&	93.6 	\\
	&		&	$\beta_2$	&	-0.381 	&	4.13 	&	4.31 	&	93.4 	&		&	-0.403 	&	4.13 	&	4.29 	&	93.4 	&		&	-0.393 	&	4.13 	&	4.28 	&	93.5 	\\
	&		&	$\gamma_1$	&	0.026 	&	3.35 	&	3.39 	&	94.7 	&		&	0.028 	&	3.35 	&	3.37 	&	95.0 	&		&	0.026 	&	3.35 	&	3.37 	&	95.1 	\\
	&	0.8	&	$\alpha_0$	&	0.256 	&	4.74 	&	4.93 	&	94.1 	&		&	0.324 	&	4.74 	&	4.93 	&	94.1 	&		&	0.311 	&	4.74 	&	4.93 	&	94.0 	\\
	&		&	$\alpha_1$	&	-0.385 	&	7.81 	&	8.41 	&	92.3 	&		&	-0.504 	&	7.81 	&	8.41 	&	92.2 	&		&	-0.451 	&	7.81 	&	8.42 	&	92.4 	\\
	&		&	$\tau_1$	&	-0.279 	&	5.59 	&	5.88 	&	93.4 	&		&	-0.292 	&	5.59 	&	5.90 	&	93.1 	&		&	-0.249 	&	5.59 	&	5.88 	&	93.4 	\\
	&		&	$\beta_2$	&	-0.334 	&	4.12 	&	4.29 	&	93.7 	&		&	-0.399 	&	4.13 	&	4.28 	&	93.7 	&		&	-0.391 	&	4.12 	&	4.28 	&	93.6 	\\
	&		&	$\gamma_1$	&	0.033 	&	3.34 	&	3.39 	&	94.7 	&		&	0.027 	&	3.34 	&	3.37 	&	95.0 	&		&	0.027 	&	3.34 	&	3.38 	&	95.0 	\\		
\hline
2000	&	0.6	&	$\alpha_0$	&	0.195 	&	3.35 	&	3.31 	&	95.4 	&		&	0.204 	&	3.35 	&	3.30 	&	95.5 	&		&	0.203 	&	3.35 	&	3.31 	&	95.5 	\\
	&		&	$\alpha_1$	&	-0.341 	&	5.52 	&	5.52 	&	95.2 	&		&	-0.362 	&	5.52 	&	5.52 	&	95.1 	&		&	-0.357 	&	5.52 	&	5.54 	&	95.0 	\\
	&		&	$\tau_1$	&	-0.093 	&	3.95 	&	4.06 	&	94.4 	&		&	-0.101 	&	3.95 	&	4.06 	&	94.6 	&		&	-0.097 	&	3.95 	&	4.06 	&	94.5 	\\
	&		&	$\beta_2$	&	-0.311 	&	2.91 	&	2.90 	&	94.8 	&		&	-0.320 	&	2.91 	&	2.90 	&	94.7 	&		&	-0.321 	&	2.91 	&	2.90 	&	94.8 	\\
	&		&	$\gamma_1$	&	0.017 	&	2.37 	&	2.31 	&	95.6 	&		&	0.014 	&	2.37 	&	2.30 	&	95.7 	&		&	0.014 	&	2.37 	&	2.30 	&	95.6 	\\
	&	0.8	&	$\alpha_0$	&	0.168 	&	3.35 	&	3.31 	&	95.5 	&		&	0.195 	&	3.35 	&	3.31 	&	95.6 	&		&	0.189 	&	3.35 	&	3.32 	&	95.4 	\\
	&		&	$\alpha_1$	&	-0.311 	&	5.52 	&	5.51 	&	95.3 	&		&	-0.361 	&	5.52 	&	5.54 	&	95.1 	&		&	-0.344 	&	5.52 	&	5.54 	&	95.1 	\\
	&		&	$\tau_1$	&	-0.108 	&	3.94 	&	4.07 	&	94.3 	&		&	-0.118 	&	3.94 	&	4.09 	&	94.3 	&		&	-0.109 	&	3.94 	&	4.08 	&	94.3 	\\
	&		&	$\beta_2$	&	-0.282 	&	2.90 	&	2.90 	&	94.9 	&		&	-0.311 	&	2.90 	&	2.90 	&	94.9 	&		&	-0.306 	&	2.90 	&	2.91 	&	94.9 	\\
	&		&	$\gamma_1$	&	0.021 	&	2.37 	&	2.31 	&	95.7 	&		&	0.016 	&	2.37 	&	2.30 	&	95.8 	&		&	0.016 	&	2.37 	&	2.30 	&	95.7 	\\
\hline													
\end{tabular}}											
\end{table}

Under Cases 1-3, we vary $M_n^*\in \{1, 2, 4\}$, $n \in \{1000, 2000\}$ and present
 different curves of $\varphi(w)=\alpha_0^*w+\sum_{m=1}^{M_n^*}\alpha^*_{m}f(w,\tau_m^*)$ in Figure S1.
Considering both the theoretical rate and common practice, we set $M_n=5$ at the order of $n^{1/15}$. Let
$\delta_n=\{\log(M_n)/n\}^{\nu}$ with $\nu$ varying in $0.6$ and $0.8$,
and $\tau_{\infty}=\max_{1\le i\le n} |{X}_i\widehat{{\beta}}^{(k)}|+1$.
For the penalty function, we consider MCP with $t=3$ and SCAD with $t=3.7$, and set $C_n$ to be $1$ or $\log\{\log (n)\}$ in BIC.
All the reported results are based on 1000 replications.

We begin by evaluating the consistency of $\widehat M_n$.
Table \ref{Table:Simulations:CS} reports the percentage of correctly specified models (i.e. $\widehat M_n=M_n^*$).
It can be observed that all the percentages are close to 1 and increase as the sample size $n$ varies from 1000 to 2000. These results validate the selection consistency  in Theorem
\ref{Theorem:Consistenncy}.
Furthermore, the performance of the proposed procedure with $C_n=\log\{\log(n)\}$ outperforms that with $C_n=1$.
Additionally, the results remain comparable when the value of $\nu$ varies from $0.6$ to $0.8$.
This indicates that the selection consistency is stable with respect to a proper choice of $\nu$.

\begin{table}
{\footnotesize
\caption{Simulation results (multiplied by 100) for Case 2 with $\nu=0.6$, $C_n=\log\{\log (n)\}$
and $\epsilon_i\sim \mathcal{N}(0,1)$.}											
\begin{tabular}{crccccrccccrcccc}		
\hline
 & &\multicolumn{4}{c}{Oracle}& &\multicolumn{4}{c}{SCAD}& &\multicolumn{4}{c}{MCP}\\	
\cline{3-6} \cline{8-11} \cline{13-16} 	
$n$ &		&	Bias	&	SE	&	SD	&	CP	&		&	Bias	&	SE	&	SD	&	CP	&		&	Bias	&	SE	&	SD	&	CP	\\	
\hline
1000	&		$\alpha_0$	&	0.188 	&	4.13 	&	4.13 	&	94.7 	&		&	0.143 	&	4.13 	&	4.17 	&	94.2 	&		&	0.147 	&	4.13 	&	4.16 	&	94.4 	\\
	&			$\alpha_1$	&	-0.788 	&	12.2 	&	12.7 	&	94.3 	&		&	-0.604 	&	12.2 	&	12.9 	&	94.6 	&		&	-0.602 	&	12.2 	&	13.0 	&	94.5 	\\
	&			$\alpha_2$	&	1.275 	&	13.6 	&	14.2 	&	94.0 	&		&	1.074 	&	13.6 	&	14.2 	&	93.9 	&		&	1.134 	&	13.6 	&	14.2 	&	93.9 	\\
	&			$\tau_1$	&	-0.024 	&	6.62 	&	6.87 	&	94.0 	&		&	-0.136 	&	6.63 	&	6.96 	&	93.6 	&		&	-0.136 	&	6.63 	&	6.98 	&	93.9 	\\
	&			$\tau_2$	&	0.271 	&	8.57 	&	9.09 	&	93.8 	&		&	0.361 	&	8.59 	&	9.33 	&	93.0 	&		&	0.428 	&	8.58 	&	9.42 	&	92.5 	\\
	&			$\beta_2$	&	0.050 	&	3.26 	&	3.36 	&	93.6 	&		&	0.008 	&	3.27 	&	3.40 	&	93.3 	&		&	0.008 	&	3.27 	&	3.40 	&	93.2 	\\
	&			$\gamma_1$	&	-0.125 	&	3.32 	&	3.24 	&	95.3 	&		&	-0.132 	&	3.32 	&	3.25 	&	95.1 	&		&	-0.132 	&	3.32 	&	3.26 	&	95.2 	\\
2000	&		$\alpha_0$	&	-0.076 	&	2.91 	&	2.86 	&	95.6 	&		&	-0.076 	&	2.91 	&	2.84 	&	95.9 	&		&	-0.090 	&	2.91 	&	2.84 	&	95.7 	\\
	&			$\alpha_1$	&	-0.379 	&	8.56 	&	8.98 	&	93.8 	&		&	-0.380 	&	8.56 	&	8.96 	&	93.9 	&		&	-0.316 	&	8.56 	&	8.97 	&	94.2 	\\
	&			$\alpha_2$	&	0.644 	&	9.56 	&	9.97 	&	94.1 	&		&	0.627 	&	9.56 	&	9.94 	&	94.3 	&		&	0.580 	&	9.56 	&	9.95 	&	94.1 	\\
	&			$\tau_1$	&	0.036 	&	4.66 	&	4.84 	&	94.3 	&		&	0.036 	&	4.66 	&	4.79 	&	94.4 	&		&	0.003 	&	4.67 	&	4.82 	&	94.3 	\\
	&			$\tau_2$	&	0.223 	&	6.03 	&	6.36 	&	93.0 	&		&	0.221 	&	6.03 	&	6.40 	&	92.9 	&		&	0.261 	&	6.03 	&	6.41 	&	92.9 	\\
	&			$\beta_2$	&	-0.086 	&	2.30 	&	2.41 	&	93.5 	&		&	-0.086 	&	2.30 	&	2.40 	&	93.5 	&		&	-0.098 	&	2.30 	&	2.40 	&	93.3 	\\
	&			$\gamma_1$	&	-0.080 	&	2.34 	&	2.31 	&	95.8 	&		&	-0.079 	&	2.34 	&	2.31 	&	95.7 	&		&	-0.080 	&	2.34 	&	2.31 	&	95.7 	\\
\hline									
\end{tabular}}	
\label{Table:caseII:normal}											
\end{table}

\begin{table}
{\footnotesize
\caption{Simulation results (multiplied by 100) for Case 3 with $\nu=0.6$, $C_n=\log\{\log (n)\}$
and $\epsilon_i\sim \mathcal{N}(0,1)$.}\label{Table:caseIII:normal}												
\begin{tabular}{ccrccccrccccrcccc}																											
\hline	
& &\multicolumn{4}{c}{Oracle}& &\multicolumn{4}{c}{SCAD}& &\multicolumn{4}{c}{MCP}\\ 	
\cline{3-6} \cline{8-11} \cline{13-16} 	
$n$ & &	Bias	&	SE	&	SD	&	CP	&		&	Bias	&	SE	&	SD	&	CP	&		&	Bias	&	SE	&	SD	&	CP	\\		
\hline
1000		&	$\alpha_0$	&	0.076 	&	3.05 	&	3.21 	&	93.7 	&		&	0.021 	&	3.05 	&	3.20 	&	94.1 	&		&	0.022 	&	3.05 	&	3.20 	&	94.1 	\\
			&	$\alpha_1$	&	1.186 	&	16.3 	&	18.0 	&	92.5 	&		&	1.881 	&	16.4 	&	17.6 	&	93.5 	&		&	1.871 	&	16.4 	&	17.5 	&	93.7 	\\
			&	$\alpha_2$	&	-1.459 	&	15.8 	&	17.2 	&	93.4 	&		&	-1.783 	&	15.9 	&	17.0 	&	93.5 	&		&	-1.778 	&	15.9 	&	17.0 	&	93.6 	\\
			&	$\alpha_3$	&	-1.445 	&	15.9 	&	16.8 	&	93.5 	&		&	-0.628 	&	15.8 	&	17.0 	&	93.0 	&		&	-0.596 	&	15.7 	&	17.0 	&	93.0 	\\
			&	$\alpha_4$	&	1.818 	&	18.4 	&	19.2 	&	93.8 	&		&	0.845 	&	18.3 	&	19.5 	&	93.3 	&		&	0.833 	&	18.3 	&	19.5 	&	93.3 	\\
			&	$\tau_1$	&	-0.575 	&	12.5 	&	13.3 	&	93.0 	&		&	-0.212 	&	12.5 	&	13.2 	&	93.4 	&		&	-0.214 	&	12.5 	&	13.1 	&	93.3 	\\
			&	$\tau_2$	&	-0.465 	&	11.3 	&	12.1 	&	92.9 	&		&	-1.002 	&	11.3 	&	11.9 	&	93.4 	&		&	-1.006 	&	11.3 	&	11.9 	&	93.4 	\\
			&	$\tau_3$	&	1.091 	&	11.3 	&	12.2 	&	92.3 	&		&	-0.032 	&	11.3 	&	12.5 	&	92.4 	&		&	-0.047 	&	11.3 	&	12.5 	&	92.4 	\\
			&	$\tau_4$	&	0.553 	&	13.6 	&	14.3 	&	93.5 	&		&	0.848 	&	13.6 	&	14.5 	&	93.4 	&		&	0.865 	&	13.6 	&	14.4 	&	93.3 	\\
			&	$\beta_2$	&	-0.269 	&	5.08 	&	5.37 	&	93.9 	&		&	-0.198 	&	5.08 	&	5.35 	&	94.2 	&		&	-0.198 	&	5.08 	&	5.35 	&	94.2 	\\
			&	$\gamma_1$	&	-0.217 	&	3.46 	&	3.56 	&	95.0 	&		&	-0.212 	&	3.46 	&	3.55 	&	94.9 	&		&	-0.212 	&	3.46 	&	3.55 	&	94.9 	\\	
	
2000		&	$\alpha_0$	&	0.026 	&	2.16 	&	2.19 	&	94.1 	&		&	-0.020 	&	2.16 	&	2.17 	&	94.6 	&		&	-0.020 	&	2.16 	&	2.17 	&	94.6 	\\
			&	$\alpha_1$	&	-0.054 	&	11.4 	&	11.9 	&	93.2 	&		&	0.443 	&	11.5 	&	11.8 	&	93.9 	&		&	0.477 	&	11.5 	&	11.7 	&	93.8 	\\
			&	$\alpha_2$	&	-0.055 	&	11.1 	&	11.4 	&	94.2 	&		&	-0.321 	&	11.1 	&	11.3 	&	94.0 	&		&	-0.353 	&	11.1 	&	11.3 	&	94.0 	\\
			&	$\alpha_3$	&	-0.660 	&	11.1 	&	11.6 	&	94.0 	&		&	-0.189 	&	11.1 	&	11.6 	&	93.1 	&		&	-0.206 	&	11.1 	&	11.6 	&	93.3 	\\
			&	$\alpha_4$	&	0.881 	&	12.9 	&	13.6 	&	94.3 	&		&	0.319 	&	12.8 	&	13.7 	&	93.7 	&		&	0.337 	&	12.8 	&	13.7 	&	93.7 	\\
			&	$\tau_1$	&	-0.482 	&	8.85 	&	9.05 	&	93.2 	&		&	-0.209 	&	8.85 	&	8.94 	&	93.7 	&		&	-0.199 	&	8.84 	&	8.94 	&	93.7 	\\
			&	$\tau_2$	&	0.167 	&	7.96 	&	8.32 	&	93.4 	&		&	-0.165 	&	7.96 	&	8.19 	&	93.7 	&		&	-0.193 	&	7.96 	&	8.18 	&	94.0 	\\
			&	$\tau_3$	&	0.253 	&	7.96 	&	8.50 	&	93.3 	&		&	-0.408 	&	7.95 	&	8.57 	&	93.0 	&		&	-0.399 	&	7.95 	&	8.56 	&	93.0 	\\
			&	$\tau_4$	&	0.193 	&	9.58 	&	10.1 	&	93.5 	&		&	0.318 	&	9.59 	&	10.0 	&	93.3 	&		&	0.314 	&	9.59 	&	10.0 	&	93.4 	\\
			&	$\beta_2$	&	-0.167 	&	3.59 	&	3.65 	&	94.8 	&		&	-0.103 	&	3.59 	&	3.61 	&	95.4 	&		&	-0.102 	&	3.59 	&	3.61 	&	95.4 	\\
			&	$\gamma_1$	&	-0.033 	&	2.45 	&	2.47 	&	93.5 	&		&	-0.032 	&	2.45 	&	2.46 	&	93.5 	&		&	-0.032 	&	2.45 	&	2.46 	&	93.5 	\\
		\hline										
\end{tabular}}										
\end{table}

We proceed to evaluate the consistency and normality of ${\widehat \theta}$.
We fix $C_n=\log\{\log(n)\}$.
The simulation results are summarized in Tables \ref{Table:caseI:normal}-\ref{Table:caseIII:normal}
and S1-S10 in the Supplementary Material.
These tables provide information on various performance measures of ${\widehat \theta}$.
Specifically, the tables include the bias (Bias), which is given by the sample mean of the estimate minus the true value. The sample standard deviation (SD) are also provided to assess the variability of ${\widehat \theta}$.
Additionally, the sample mean of the standard error (SE) estimate is reported. Finally, the $95\%$ empirical coverage probability (CP) based on the normal approximation is calculated to evaluate the accuracy of the confidence intervals.
For comparison, we also carry out the oracle procedure, denoted as ${\widehat \theta}_{o}$. This procedure involves  minimizing \eqref{Oracle:est} with  $M_n^*$   known.

The results indicate that our proposed method performed well for the situations considered.
Specifically, the proposed estimators are practically unbiased.
The estimated standard errors are very close to  the sample standard deviations,
and decrease as $n$ varies from $1000$ to $2000$.
The $95\%$ empirical coverage probabilities are also reasonable, suggesting that the proposed confidence intervals have good coverage properties.
Importantly, the performance of the proposed estimators is comparable to that of the oracle estimators.
These results demonstrate the validity of Theorems \ref{Theorem:Consistenncy} and \ref{Theorem:Oracle}.
Moreover, the results are comparable with $\nu$ varying from $0.6$ to $0.8$,
indicating that the consistency and normality of ${\widehat\theta}$ are stable to a proper choice of $\nu$.

\subsection{Power analysis}
We conduct simulation studies to evaluate the power performance of the test for detecting the existence of knots.
We set ${\beta}^*=(1,-1)^\top$ and $\gamma_1^*=0.5$. 
The covariates ${X}_i$ and $Z_{i1}$ are generated as in Section \ref{section:6.1}.
We considered two cases for generating $Y_i$:

Case 4: $Y_i=Z_{i1}\gamma_1^*+\alpha_0^*{X}_i^\top{\beta}^*+\alpha_1^*f({X}_i^\top{\beta}^*,\tau_1^*)+\epsilon_i,$
where $\tau_1^*=0$, $\alpha_0^*=1$ and $\alpha_1^*=\widetilde\alpha$.

Case 5: $Y_i=Z_{i1}\gamma_1^*+\alpha_0^*{X}_i^\top{\beta}^*+\sum_{m=1}^2\alpha_m^*f({X}_i^\top{\beta}^*,\tau_m^*)+\epsilon_i,$
where ${\tau}^*=(-1,1)^\top$, $\alpha_0^*=1$ and $\alpha_1^*=\alpha_2^*=\widetilde\alpha$.

\begin{table}[h]\centering	
{\footnotesize												
\caption{Empirical sizes (\%) and powers (\%) of the proposed test at 0.05 level of significance.}\label{Table:Simulations:Test}					
\begin{tabular}{cccccccccccccccc}				
\hline
&&&\multicolumn{6}{c}{$n=500$}&&\multicolumn{6}{c}{$n=1000$}\\
\cline{4-9} \cline{11-16}	
Case	&	$\epsilon_i$	&	$\nu$	&	$\widetilde\alpha=0$	&	0.05	&	0.1	&	0.15	&	0.2	&	0.25	&		&	$\widetilde\alpha=0$	&	0.05	&	0.1	&	0.15	&	0.2	&	0.25	\\
\hline																															
4	&	$\mathcal{N}(0,1)$	&	0.6	&	4.50 	&	14.5	&	41.8	&	77.1	&	94.2	&	99.6	&		&	5.20 	&	23.5	&	74.0 	&	98.3	&	100	&	100	\\
	&		&	0.7	&	4.20 	&	14.4	&	41.8	&	77.1	&	94.1	&	99.5	&		&	4.70 	&	23.0 	&	74.5	&	98.3	&	100	&	100	\\
	&		&	0.8	&	4.60 	&	14.1	&	42.1	&	77.0 	&	94.3	&	99.6	&		&	4.80 	&	22.7	&	73.9	&	98.3	&	100	&	100	\\
	&	$\text{Schi}^2(2)$	&	0.6	&	4.70 	&	10.6	&	37.6	&	75.9	&	96.3	&	99.9	&		&	5.50 	&	17.0 	&	72.2	&	98.1	&	100	&	100	\\
	&		&	0.7	&	4.70 	&	10.6	&	36.6	&	75.1	&	96.3	&	99.9	&		&	5.60 	&	16.8	&	72.6	&	98.0 	&	100	&	100	\\
	&		&	0.8	&	4.60 	&	10.5	&	37.3	&	75.8	&	96.4	&	99.9	&		&	5.90 	&	16.5	&	72.7	&	97.9	&	100	&	100	\\
	&	$t(4)$	&	0.6	&	5.00 	&	10.6	&	24.0 	&	45.9	&	73.1	&	88.6	&		&	5.00 	&	13.7	&	40.9	&	79.8	&	95.5	&	99.3	\\
	&		&	0.7	&	5.20 	&	10.3	&	24.2	&	46.5	&	73.3	&	89.0 	&		&	5.00 	&	13.2	&	40.7	&	79.5	&	95.4	&	99.3	\\
	&		&	0.8	&	4.90 	&	10.4	&	23.7	&	45.9	&	73.4	&	88.7	&		&	4.80 	&	13.8	&	41.1	&	79.7	&	95.5	&	99.2	\\
5	&	$\mathcal{N}(0,1)$	&	0.6	&	4.40 	&	45.1	&	97.1	&	100	&	100	&	100	&		&	4.70 	&	74.3	&	100	&	100	&	100	&	100	\\
	&		&	0.7	&	4.60 	&	44.8	&	97.3	&	100	&	100	&	100	&		&	5.00 	&	74.6	&	100	&	100	&	100	&	100	\\
	&		&	0.8	&	4.60 	&	44.8	&	97.0 	&	100	&	100	&	100	&		&	4.60 	&	74.0 	&	100	&	100	&	100	&	100	\\
	&	$\text{Schi}^2(2)$	&	0.6	&	5.20 	&	42.4	&	99.2	&	100	&	100	&	100	&		&	4.90 	&	76.4	&	100	&	100	&	100	&	100	\\
	&		&	0.7	&	5.40 	&	41.8	&	99.3	&	100	&	100	&	100	&		&	4.70 	&	75.8	&	100	&	100	&	100	&	100	\\
	&		&	0.8	&	5.80 	&	42.0 	&	99.2	&	100	&	100	&	100	&		&	4.90 	&	75.8	&	100	&	100	&	100	&	100	\\
	&	$t(4)$	&	0.6	&	5.30 	&	27.5	&	75.6	&	97.8	&	99.9	&	100	&		&	4.40 	&	46.9	&	96.8	&	100	&	100	&	100	\\
	&		&	0.7	&	4.80 	&	27.9	&	75.2	&	98.0 	&	99.9	&	100	&		&	4.60 	&	46.9	&	96.8	&	100	&	100	&	100	\\
	&		&	0.8	&	5.30 	&	27.4	&	75.3	&	97.7	&	99.9	&	100	&		&	4.40 	&	47.2	&	97.2	&	100	&	100	&	100	\\

\hline		
\end{tabular}}											
\end{table}
Under these settings, the number of knots is $M_n^*=1$ for Case 4 and $M_n^*=2$ for Case 5.
We vary $\widetilde \alpha \in \{0,~ 0.05,~ 0.1,~ 0.15,~ 0.2,~ 0.25\}$, where $\widetilde\alpha=0$ corresponds to the null hypothesis.
The departure from $H_0$ increases as $\widetilde\alpha$ varies from 0.05 to 0.25.
When calculating the test statistic $\mathcal{T}$,
we take the grid of 100 evenly spaced points in $\Theta_{\tau}=[-2.5,2.5].$
We vary $\nu\in \{0.6,~0.7,~0.8\}$.
The significance level is set to be $\varsigma=0.05$, and the sample size is set to be $n=500$ and $n=1000$.
The critical value is calculated using the resampling method with 1000 simulated realizations. The empirical sizes and powers of the test, based on 1000 replications, are summarized in Table \ref{Table:Simulations:Test}.
Under the null hypothesis ($\widetilde\alpha=0$), the empirical sizes are close to the nominal significance level of $5\%$.
The empirical sizes are comparable for different values of $\nu$, demonstrating the validity and robustness of the proposed test statistics to detect knots.
The empirical powers increase as the sample size $n$ or the magnitude of the knot effect $\widetilde\alpha$ increases.
When $\widetilde\alpha$ increases to $0.25$, the powers are almost $100\%$ for all settings.
As expected, the powers are higher when $M_n^*=2$ compared to $M_n^*=1$.

\section{Real data analysis}\label{section:real:data}

In this section, we analyze two real datasets to validate the performance of the proposed method.

\subsection{Real estate valuation dataset}
\label{REV}
In Example 1, the dataset consists of 414 real estate transaction records during the period of June 2012 to May 2013,
collected from two districts in Taipei City and two districts in New Taipei City.
The response variable is the residential housing price per unit area.
The goal of the analysis is to investigate the knot effects of various variables on housing prices.
Based on related research in \cite{yeh2018building}, four covariates are considered:
the distance to the nearest MRT station (Meter), house age (Year), transaction date (Date),
and the number of convenience stores in the living circle on foot (Number).
The raw data is available at https://archive.ics.uci.edu.

The analysis begins by attempting several choices of $X_{ij}$ as potential knot variables, but none of them yield better results in terms of goodness-of-fit
$R^2$,
compared to using the Number variable as $Z_{i1}$, the negative Meter variable as $X_{i1}$,
the Year variable as $X_{i2}$, and the Data variable as $X_{i3}$.
Here, $R^2=1-\text{SSE/SST}$, $\text{SST}=\sum_{i=1}^n(Y_i-\overline Y)^2$,
$\text{SSE}=\sum_{i=1}^n(\widehat Y_i-Y_i)^2$ and $\widehat Y_i~(i=1,\dots,n)$ are the fitted values of $Y_i$ using model \eqref{SIR-model}.
The testing procedure described in Section \ref{section:testing} is performed to examine the existence of knots.
The resulting p-value is less than 0.001, which indicates the presence of at least one knot.
The parameters in model \eqref{SIR-model} are then estimated using the proposed method.
For this purpose, as in the simulation studies, we set $M_n=5$, $C_n=\log\{\log(n)\}$ and $\delta_n=\{\log(M_n)/n\}^{0.8}$.
The estimated results are provided in Table \ref{Table:real data:REV}, and the estimated number of knots is $\widehat M_n=1$.
Figure \ref{Figure:real data} displays a scatter plot between the estimated index ${X}_i^\top{\widehat\beta}$ and the residential housing price per unit area, along with the fitted LSIR curve.

The estimated location of knot is $-0.25$ with p-value less than 0.05.
In addition, $\widehat\alpha_0=2.99$ and $\widehat\alpha_1=16.7$ with p-value less than 0.05.
This indicates that the housing prices have a slow increase with the index in the interval $(-\infty, -0.25]$,
but  increase quickly with a rate of $\widehat\mu_1=\widehat\alpha_0+\widehat\alpha_1=19.69$ when  the index exceed $-0.25$.
Finally, it is observed that the house age has a significant negative effect on housing prices, while the transaction date and the number of convenience stores have positive effects.

\begin{table}[h]\centering								
{\footnotesize
\caption{Analysis results for real estate valuation detaset.}\label{Table:real data:REV}									
\begin{tabular}{crccccrccc}												
\hline																	
&\multicolumn{4}{c}{SCAD}& &\multicolumn{4}{c}{MCP}\\				
\cline{2-5} \cline{7-10}												
&	Est	&	SE	&	CI	&	p-value	&		&	Est	&	SE	&	CI	&	p-value	\\
\hline	
$\alpha_0$	&	2.99 	&	0.77 	&	(1.47,4.50)	&	0.000 	&		&	2.99 	&	0.77 	&	(1.47,4.50)	&	0.000 	\\
$\alpha_1$	&	16.7 	&	2.19 	&	(12.4,21.0)	&	0.000 	&		&	16.7 	&	2.19 	&	(12.4,21.0)	&	0.000 	\\
$\tau_1$	&	-0.25 	&	0.11 	&	(-0.47,-0.03)	&	0.025 	&		&	-0.25 	&	0.11 	&	(-0.47,-0.03)	&	0.025 	\\
$\beta_2$	&	-0.15 	&	0.03 	&	(-0.21,-0.10)	&	0.000 	&		&	-0.15 	&	0.03 	&	(-0.21,-0.10)	&	0.000 	\\
$\beta_3$	&	0.11 	&	0.03 	&	(0.06,0.16)	&	0.000 	&		&	0.11 	&	0.03 	&	(0.06,0.16)	&	0.000 	\\
$\gamma_0$	&	26.8 	&	1.38 	&	(24.1,29.5)	&	0.000 	&		&	26.8 	&	1.38 	&	(24.1,29.5)	&	0.000 	\\
$\gamma_1$	&	0.52 	&	0.20 	&	(0.13,0.90)	&	0.009 	&		&	0.52 	&	0.20 	&	(0.13,0.90)	&	0.009 	\\
\hline																						
\end{tabular}}										
\end{table}

\subsection{Fish toxicity dataset}

In Example 2, we have obtained a fish toxicity dataset sourced from a study on predicting acute toxicity of chemicals towards the fathead minnow (\citealp{cassotti2015similarity}).
The dataset consists of 908 samples, where each sample represents a different chemical.
The response variable is LC50, which is the concentration that causes death in $50\%$ of test fathead minnows over a test duration of 96 hours.
In line with the research conducted by \cite{cassotti2015similarity}, the focus is on six specific molecular descriptors:
CIC0 (information indices), SM1\_Dz(Z) (2D matrix-based descriptors), MLOGP (molecular properties), GATS1i (2D autocorrelations), NdsCH (atom-type counts), and NdssC (atom-type counts).
The raw data can be accessed at  https://michem.unimib.it.

\begin{figure}
\centering
\subfigure{\includegraphics[width=0.45\textwidth]{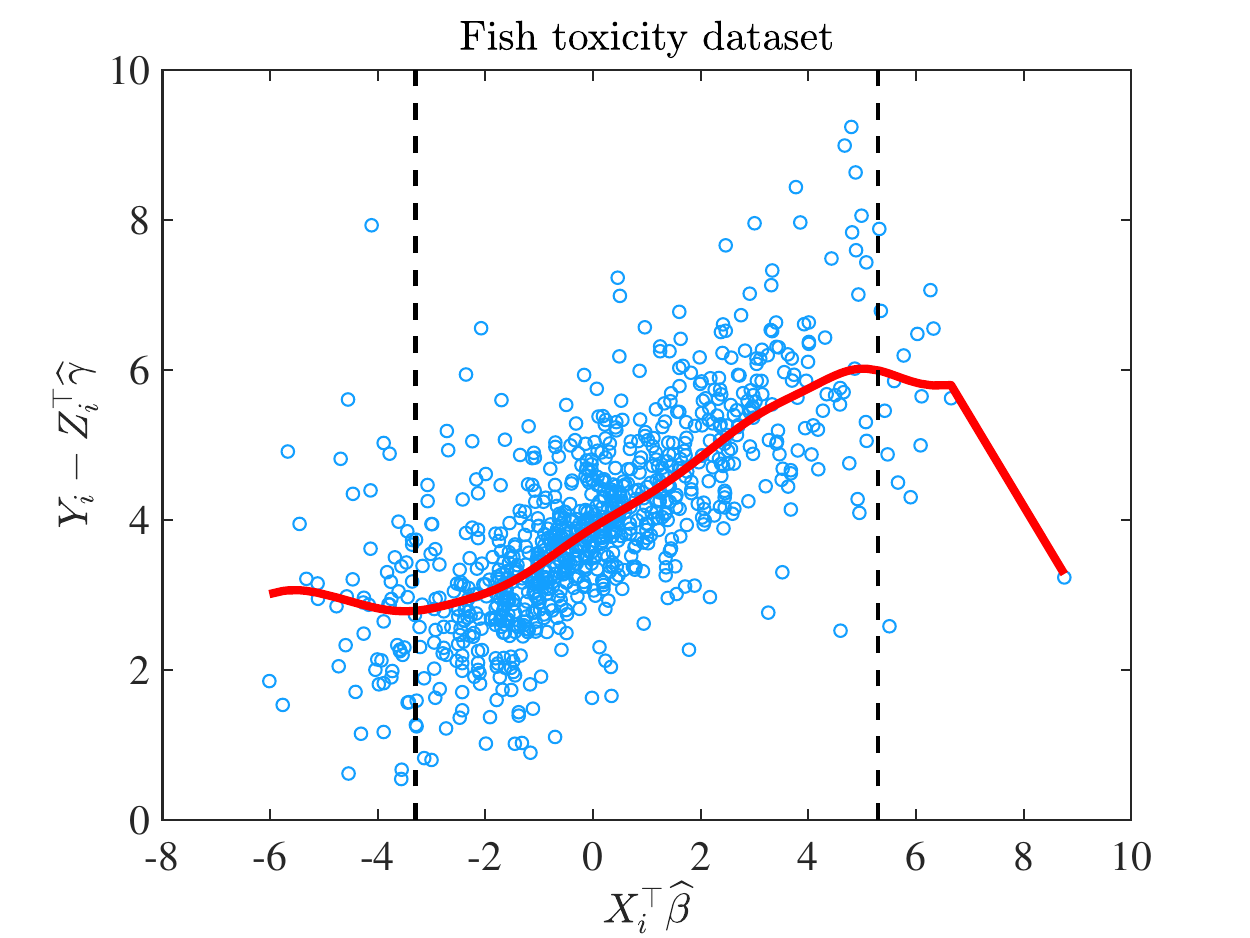}}
\subfigure{\includegraphics[width=0.45\textwidth]{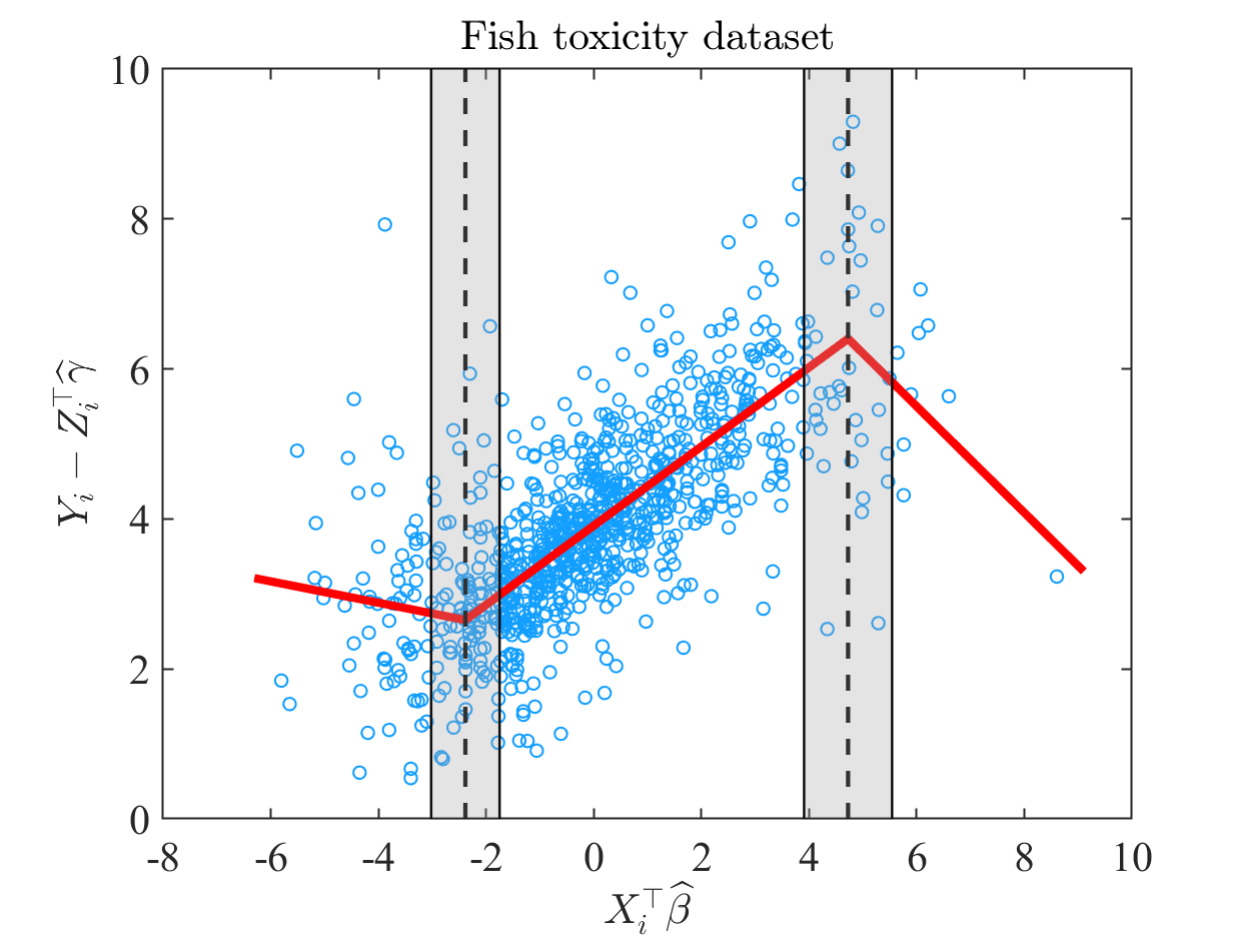}}
\caption{{\small The analysis of the fish toxicity dataset.
The red lines in the left and right panels are the estimated curves via the single-index regression model \citep{wang2010estimation}
and model \eqref{SIR-model}, respectively.
The grey region represents the $95\%$ confidence intervals of the estimates of the knot locations.}}\label{Figure:fish}
\end{figure}

We investigate whether there exist  effects of various variables on LC50.
Several choices of $X_i$ are tested, but none of them result in better $R^2$ values,
compared to using GATS1i as $Z_{i1}$, NdsCH as $Z_{i2}$, NdssC as $Z_{i3}$,
CIC0 as $X_{i1}$, SM1\_Dz(Z) as $X_{i2}$, and MLOGP as $X_{i3}$.
The $R^2$ is defined as in Section \ref{REV}.
To examine the presence of knots, the testing procedure described in Section \ref{section:testing} is performed, yielding a p-value of 0.0052,
which indicates the presence of at least one knot.
The parameters in model \eqref{SIR-model} are then estimated using the proposed method.
Here, we also set $M_n=5$, $C_n=\log\{\log(n)\}$ and $\delta_n=\{\log(M_n)/n\}^{0.8}$.
The estimated results can be found in Table \ref{Table:real data:Fish}, and the number of estimated knots is $\widehat M_n=2$.
The left panel of Figure \ref{Figure:fish} reports the estimated curve derived from the partial single-index regression model \citep{wang2010estimation},
while the right panel displays the fitted LSIR curve.

The right panel of Figure \ref{Figure:fish} shows the estimated knot locations at -2.38 and 4.73, both with p-values less than 0.05.
Additionally, $\widehat\alpha_0=-0.14$ with p-value larger than 0.05,
while both $\widehat\alpha_1=0.67$ and $\widehat\alpha_2=-1.24$ with p-values less than 0.05.
These findings indicate distinct effects of the index on LC50 in different intervals.
Specifically, LC50 remains stable when the index falls within the interval $(-\infty, -2.38]$.
LC50 increases at a rate of $\widehat{\mu}_1 = \widehat{\alpha}_0 + \widehat{\alpha}_1 = 0.53$ when the index is in the interval $(-2.28, 4.73]$.
Finally, LC50 decreases rapidly when the index is in the interval $(4.73, \infty)$.


\begin{table}[h]\centering								
{\footnotesize
\caption{Analysis results for fish toxicity dataset.}\label{Table:real data:Fish}									
\begin{tabular}{crccccrccc}												
\hline																	
&\multicolumn{4}{c}{SCAD}& &\multicolumn{4}{c}{MCP}\\				
\cline{2-5} \cline{7-10}												
&	Est	&	SE	&	CI	&	p-value	&		&	Est	&	SE	&	CI	&	p-value	\\
\hline	
$\alpha_0$	&	-0.14 	&	0.11 	&	(-0.37,0.08)	&	0.212 	&		&	-0.14 	&	0.11 	&	(-0.37,0.08)	&	0.212 	\\
$\alpha_1$	&	0.67 	&	0.13 	&	(0.42,0.92)	&	0.000 	&		&	0.67 	&	0.13 	&	(0.42,0.92)	&	0.000 	\\
$\alpha_2$	&	-1.24 	&	0.25 	&	(-1.73,-0.74)	&	0.000 	&		&	-1.24 	&	0.25 	&	(-1.73,-0.74)	&	0.000 	\\
$\tau_1$	&	-2.38 	&	0.32 	&	(-3.02,-1.75)	&	0.000 	&		&	-2.38 	&	0.32 	&	(-3.02,-1.75)	&	0.000 	\\
$\tau_2$	&	4.73 	&	0.42 	&	(3.91,5.54)	&	0.000 	&		&	4.73 	&	0.42 	&	(3.91,5.54)	&	0.000 	\\
$\beta_2$	&	1.23 	&	0.09 	&	(1.05,1.42)	&	0.000 	&		&	1.23 	&	0.09 	&	(1.05,1.42)	&	0.000 	\\
$\beta_3$	&	1.11 	&	0.16 	&	(0.80,1.43)	&	0.000 	&		&	1.11 	&	0.16 	&	(0.80,1.43)	&	0.000 	\\
$\gamma_0$	&	2.31 	&	0.40 	&	(1.52,3.10)	&	0.000 	&		&	2.31 	&	0.40 	&	(1.52,3.10)	&	0.000 	\\
$\gamma_1$	&	-0.38 	&	0.04 	&	(-0.45,-0.31)	&	0.000 	&		&	-0.38 	&	0.04 	&	(-0.45,-0.31)	&	0.000 	\\
$\gamma_2$	&	0.37 	&	0.05 	&	(0.27,0.47)	&	0.000 	&		&	0.37 	&	0.05 	&	(0.27,0.47)	&	0.000 	\\
$\gamma_3$	&	0.03 	&	0.04 	&	(-0.05,0.11)	&	0.444 	&		&	0.03 	&	0.04 	&	(-0.05,0.11)	&	0.444 	\\
\hline										
\end{tabular}}										
\end{table}

\section{Discussion}\label{section:discussion}
We proposed a linear spline index regression model for data analysis,
where the locations and number of knots were unknown a priori.
We developed a penalized smoothing least squares method to estimate the unknown parameters
and demonstrated the consistency and asymptotic normality of the proposed estimators.
Moreover, we allowed the number of knots to diverge with the sample size.
Simulation studies have indicated that the proposed method performs effectively.

To conclude this article, we will discuss several interesting topics for future study.
First, the proposed method may be applicable to models with multiple indices and unknown knots.
In addition, we can use a polynomial regression model to capture the nonlinear effects of an index on each segment, defined as follows:
\begin{align*}
Y_i=\sum_{l=0}^{L_0}\alpha_{0l}(X_i^\top\beta)^l+\sum_{m=1}^{M_n^*}\bigg[\sum_{l=0}^{L_{m}}\alpha_{ml}(X_i^\top\beta)^l\bigg](X_i^\top\beta-\tau_m)_{+}+Z_i^\top\gamma+\epsilon_i,
\end{align*}
where $L_{m}~(0\le m\le M_n^*)$ represents unspecified positive integers and $\alpha_{0l}~(1\le l\le L_m)$ are unknown parameters.
Let $L$ and $M_n$ be some prespecified integers.
Inspired by \eqref{Example}, we can estimate the unknown parameters by minimizing
\begin{align*}
\frac{1}{2}\sum_{i=1}^n\bigg\{Y_i-Z_i^\top\gamma-\sum_{l=0}^{L}\alpha_{0l}(X_i^\top\beta)^l
-\sum_{m=1}^{M_n}\bigg[\sum_{l=0}^{L}\alpha_{ml}(X_i^\top\beta)^l\bigg]q_n(X_i^\top\beta,\tau_m)\bigg\}^2
&+n\sum_{m=0}^{M_n} \|\widetilde \alpha_m\|_1^{1/2},
\end{align*}
where $\widetilde\alpha_m=(\alpha_{m0},\dots,\alpha_{mL})^\top$ and $\|\widetilde \alpha_m\|_1=\sum_{l=0}^{L}|\alpha_{ml}|$.
Finally, the proposed method can be extended to encompass
generalized linear models, quantile regression models and survival models. Investigating the applicability and performance of the method in these broader contexts would be an intriguing direction for future research.

\section*{Acknowledgments}
Lianqiang Qu's research was partially supported by the National Natural Science Foundation of China (Grant No. 12471256).
Meiling Hao's research was partially supported by the Fundamental Research Funds for the Central Universities in UIBE (CXTD14-05)
and the National Natural Science Foundation of China (Grant Nos. 12371264 and 12171329).
Liuquan Sun's research was partially supported by the National Natural Science Foundation of China (Grant No. 12171463).

\section*{Supplementary material}
\label{SM}
Supplementary material includes the proofs of Theorems
\ref{Theorem:Consistenncy}-\ref{Theorem:Test},
an iterative procedure to minimize \eqref{SP:est},
and additional simulation results.

\bibliographystyle{myapalike}
\bibliography{ref}

\clearpage
\newpage

\appendix

\setcounter{page}{1}
\setcounter{equation}{0}

\renewcommand\thesection{\Alph{section}}
\numberwithin{equation}{section}
\newcommand{\supplementalequation}{\Alph{section}.\arabic{equation}}
\renewcommand{\theequation}{\supplementalequation}
\setcounter{equation}{0}

\begin{center}
{\Large\bf \textsf{Supplementary Material to ``Linear spline index regression model: Interpretability, nonlinearity and dimension reduction"}}


\end{center}

This supplementary material is organized as follows.
Section \ref{secA} provides the proofs of Theorems 1-4.
Section \ref{secB} presents several lemmas.
Section \ref{secC} introduces an iterative algorithm to minimize (4).
Section \ref{secD} presents additional results from the simulation studies.

\section{Proofs of Theorems 1-4}\label{secA}
\subsection{Proof of Theorem 1}
Define $\rho_n=\sqrt{s_n}(1/\sqrt{n}+a_n)$ and
\begin{align*}
    L_n({\theta})&=\underbrace{\frac{1}{2}\sum_{i=1}^n\Big\{Y_i-\widetilde {Z}_i^\top\eta-\alpha_0{X}_i^\top{\beta}-\sum_{m=1}^{M_n}\alpha_mq_n({X}_i^\top{\beta},\tau_m)\Big\}^2}_{h_n({\theta})}+n\sum_{m=1}^{M_n}p_{\lambda_n,t}(|\alpha_m|).
\end{align*}
The proofs are divided in the following steps.

{\it Step 1.}
We constrain $L_n({\theta})$ to the $s_n$-dimensional subspace
$\mathcal{B}=\{{\theta}: \alpha_{m}=0, m=M_n^*+1,\dots,M_n\}.$
Let
\[
\bar L_n({\theta}_o)=\underbrace{\frac{1}{2}\sum_{i=1}^n\Big\{Y_i-\widetilde{Z}_i^\top{\eta}-\alpha_0{X}_i^\top{\beta}-\sum_{m=1}^{M_n^*}\alpha_mq_n({X}_i^\top{\beta},\tau_m)\Big\}^2}_{\bar h_n({\theta}_o)}+n\sum_{m=1}^{M_n^*}p_{\lambda_n,t}(|\alpha_m|).
\]
We show that there exists a strict local minimizer ${\widehat \theta}_1$ of $\bar L_n({\theta}_o)$
such that $\|{\widehat \theta}_1-{\theta}_{o}^*\|=O_p(\rho_n)$.
It suffices to show that for any given $\varepsilon>0$, there exists a sufficiently large constant $C_0$ such that
\begin{align}\label{Them1:proof:eq0:1}
    \mathbb{P}\bigg\{\inf_{\|{u}\|=C_0} \bar L_n({\theta}_{o}^*+\rho_n {u})>\bar L_n({\theta}_{o}^*) \bigg\} \geq 1-\varepsilon.
\end{align}
Define $\mathcal{O}_n({u})=\bar L_n({\theta}_{o}^*+\rho_n {u})-\bar L_n({\theta}_{o}^*).$
By the Taylor expansion of $\mathcal{O}_n({u})$ at $\theta_o^*$, we obtain
\begin{align}\label{Them1:proof:eq0:2}
    \mathcal{O}_n({u})&
    = \bar h_n({\theta}_{o}^*+\rho_n {u})-\bar h_n({\theta}_{o}^*)+n\sum_{m=1}^{M_n^*}\big\{p_{\lambda_n,t}(|\alpha_m^*+\rho_n u_m|)-p_{\lambda_n,t}(|\alpha_m^*|)\big\}
    \nonumber\\
    &=  \underbrace{\rho_n\bigg\{\frac{\partial \bar h_n({\theta}_{o}^*)}{\partial {\theta}_o}\bigg\}^\top {u}}_{I_1}
    +\underbrace{\frac{1}{2}{u}^\top\frac{\partial^2 \bar h_n({\theta}_{o}^*)}{\partial {\theta}_o \partial{\theta}_o^\top}{u} \rho_n^2}_{I_2}
    +\underbrace{\frac{1}{6}\bigg\{\frac{\partial}{\partial {\theta}_o}\bigg({u}^\top\frac{\partial^2 \bar h_n({\breve\theta}_o)}{\partial {\theta}_o\partial{\theta}_o^\top}{u}\bigg)\bigg\}^\top {u} \rho_n^3}_{I_3}\nonumber\\
    &\quad+\underbrace{n\sum_{m=1}^{M_n^*}\rho_np'_{\lambda_n,t}(|\alpha_m^*|)\text{sgn}(\alpha_m^*)u_m}_{I_4}
    +\underbrace{n\sum_{m=1}^{M_n^*}\frac{1}{2}\rho_n^2p''_{\lambda_n,t}(|\alpha_m^*|)u_m^2\{1+o(1)\}}_{I_5},
\end{align}
where ${\breve \theta}_o$ lies between ${\theta}_{o}^*$ and ${\theta}_{o}^*+\rho_n {u}$ and $u_m$ is the $m$th element of ${u}.$
To prove \eqref{Them1:proof:eq0:1}, we need to demonstrate that the right-hand side of \eqref{Them1:proof:eq0:2} is positive for a sufficiently large value of $C_0$.
For simplicity, let  $d=d_1+d_2+1$, and define
\begin{align*}
    &T_{ni}({\theta_o})=Y_i-\widetilde{Z}_i^\top{\eta}-\alpha_0{X}_i^\top{\beta}-\sum_{m=1}^{M_n^*}\alpha_mq_n({X}_i^\top{\beta},\tau_m),\\
\text{and}~~&~T_i({\theta_o})=Y_i-\widetilde{Z}_i^\top{\eta}-\alpha_0{X}_i^\top{\beta}-\sum_{m=1}^{M_n^*}\alpha_mf({X}_i^\top{\beta},\tau_m).
\end{align*}
Then, let ${g}_{ni}({\theta_o})=\partial \bar h_n({\theta}_{o})/\partial {\theta}_o$
and $g_{ni,j}({\theta_o})$ be the $j$th element of ${g}_{ni}({\theta_o})$, where
\begin{align*}
g_{ni,j}({\theta_o})=\begin{cases}
    -q_n({X}_i^\top{\beta},\tau_{j})T_{ni}({\theta_o}), &\text{if}~1\le j\le M_n^*,\\
    -\alpha_{j-M_n^*}\frac{\partial}{\partial \tau_{j-M_n^*}}q_n({X}_i^\top{\beta},\tau_{j-M_n^*})T_{ni}({\theta_o}),
     &\text{if}~ M_n^*+1\le j\le 2M_n^*,\\
    -{X}_i^\top{\beta} T_{ni}({\theta_o}), & \text{if}~j=2M_n^*+1,\\
    -\big(\alpha_0X_{i,j-2M_n^*}+\sum_{m=1}^{M_n^*}\alpha_m\frac{\partial q_n({X}_i^\top{\beta},\tau_m)}{\partial \beta_{j-2M_n^*}} \big)T_{ni}({\theta_o}), &\text{if}~ 2M_n^*+2\le j \le 2M_n^*+d_1,\\
    -\widetilde Z_{i,j-2M_n^*-d_1-1}T_{ni}({\theta_o}),
    &\text{if}~ 2M_n^*+d_1+1\le j\le 2M_n^*+d.
\end{cases}
\end{align*}
Additionally, we define ${g}_i({\theta_o})=\big(g_{i,1}({\theta_o}),\ldots,g_{i,s_n^*}({\theta_o})\big)^\top$, where
\begin{align*}
g_{i,j}({\theta_o})=\begin{cases}
    -f({X}_i^\top{\beta},\tau_{j})T_{i}({\theta_o}), &\text{if}~1\le j\le M_n^*,\\
    \alpha_{j-M_n^*}I({X}_i^\top{\beta}>\tau_{j-M_n^*})T_{i}({\theta_o}),
     &\text{if}~ M_n^*+1\le j\le 2M_n^*,\\
    -{X}_i^\top{\beta} T_{i}({\theta_o}), & \text{if}~j=2M_n^*+1,\\
    -X_{i,j-2M_n^*}\big(\alpha_0+\sum_{m=1}^{M_n^*}\alpha_mI({X}_i^\top{\beta}>\tau_m) \big)T_{i}({\theta_o}), &\text{if}~ 2M_n^*+2\le j \le 2M_n^*+d_1,\\
    -\widetilde Z_{i,j-2M_n^*-d_1-1}T_{i}({\theta_o}), &\text{if}~ 2M_n^*+d_1+1\le j\le 2M_n^*+d.
\end{cases}
\end{align*}

We begin by demonstrating
\begin{align}\label{Them1:proof:eq1}
    |I_1|=O_p(n\rho_n^2\|{u}\|).
\end{align}
For this purpose, we can express the summation as follows:
\begin{align*}
   \sum_{i=1}^n{g}_{ni}({\theta}_{o}^*)=\sum_{i=1}^n\big\{{g}_{ni}({\theta}_{o}^*)-{g}_{i}({\theta}_{o}^*)\big\}+\sum_{i=1}^n {g}_{i}({\theta}_{o}^*).
\end{align*}
For each $1\le j\le M_n^*$, we have
\begin{align}\label{Them1:proof:eq1:3}
    &\mathbb{E}\big\{|g_{ni,j}({\theta}_{o}^*)-g_{i,j}({\theta}_{o}^*)|\big\}\nonumber\\
    \leq &  \mathbb{E}\Big\{\big|\big(Y_i-\widetilde{Z}_i^\top{\eta}^*-\alpha_0^*{X}_i^\top{\beta}^*\big)\big(q_n({X}_i^\top{\beta}^*,\tau_{j}^*)-f({X}_i^\top{\beta}^*,\tau_{j}^*)\big)\big|\Big\}\nonumber\\
    &\quad+\mathbb{E}\bigg\{\Big|\sum_{m=1}^{M_n^*}\alpha_m^*\Big(q_n({X}_i^\top{\beta}^*,\tau_{j})q_n({X}_i^\top{\beta}^*,\tau_{m}^*)-f({X}_i^\top{\beta}^*,\tau_{j})f({X}_i^\top{\beta}^*,\tau_{m}^*)\Big)\Big|\bigg\}.
\end{align}
For simplicity of presence, we assume that the support of $\mathcal{K}(x)$ is $[-1,1]$.
According to condition 5, we obtain
\begin{align}\label{Them1:proof:eq1:1}
|q_n(w,\tau_m^*)-f(w,\tau_m^*)|\le C_1\delta_nI(\tau_m^*-\delta_n\leq w\leq \tau_m^*+\delta_n),
\end{align}
and
\begin{align}\label{Them1:proof:eq1:2}
\mathbb{E}\big\{I(\tau_m^*-\delta_n\leq W_i\leq \tau_m^*+\delta_n)\big\}&=\int_{\tau_m^*-\delta_n}^{\tau_m^*+\delta_n}dF(w)=O(\delta_n),
\end{align}
where $W_i={X}_i^\top{\beta}^*$, $F(w)$ represents the cumulative distribution function of $W_i$ and $C_1$ is a positive constant.
By  \eqref{Them1:proof:eq1:1}, \eqref{Them1:proof:eq1:2} and conditions 1-4,
we have that the first term on the right hand side of \eqref{Them1:proof:eq1:3} satisfies
\begin{align*}
    &\mathbb{E}\Big\{\big|\big(Y_i-\widetilde{Z}_i^\top{\eta}^*-\alpha_0^*{X}_i^\top{\beta}^*\big)\big(q_n({X}_i^\top{\beta}^*,\tau_{j}^*)-f({X}_i^\top{\beta}^*,\tau_{j}^*)\big)\big|\Big\}\nonumber\\
    \leq& C_1\sum_{m=1}^{j-1}|\alpha_m^*|\times|\tau_j^*-\tau_m^*|\delta_n\mathbb{E}\big\{I(\tau_j^*-\delta_n\leq {X}_i^\top{\beta}^*\leq \tau_j^*+\delta_n)\big\}\nonumber\\
    &+C_1|\alpha_j^*|\mathbb{E}\big\{f({X}_i^\top{\beta}^*,\tau_j^*)\delta_nI(\tau_j^*-\delta_n\leq {X}_i^\top{\beta}^*\leq \tau_j^*+\delta_n)\big\}\nonumber\\
    &+C_1\mathbb
    {E}\big\{\delta_nI(\tau_j^*-\delta_n\leq {X}_i^\top{\beta}^*\leq \tau_j^*+\delta_n)\mathbb{E}(|\epsilon_i|\big|{X}_i^\top{\beta}^*)\big\}\nonumber\\
    \leq & O(M_n\delta_n^2)+O(\delta_n^2),
\end{align*}
and the second term on the right hand side of \eqref{Them1:proof:eq1:3} satisfies
\begin{align*}
    &\mathbb{E}\bigg\{\Big|\sum_{m=1}^{M_n^*}\alpha_m^*\Big(q_n({X}_i^\top{\beta}^*,\tau_{j})q_n({X}_i^\top{\beta}^*,\tau_{m}^*)-f({X}_i^\top{\beta}^*,\tau_{j})f({X}_i^\top{\beta}^*,\tau_{m}^*)\Big)\Big|\Big\}\nonumber\\
    \leq & C_1\sum_{m=1}^{j-1}|\alpha_m^*|\times|\tau_j^*-\tau_m^*|\delta_n\mathbb{E}\big\{I(\tau_j^*-\delta_n\leq {X}_i^\top{\beta}^*\leq \tau_j^*+\delta_n)\big\}\nonumber\\
    &+C_1^2|\alpha_j^*|\delta_n^2\mathbb{E}\big\{I(\tau_j^*-\delta_n\leq {X}_i^\top{\beta}^*\leq \tau_j^*+\delta_n)\big\}\nonumber\\
    &+C_1\sum_{m=j+1}^{M_n^*}|\alpha_m^*|\times|\tau_m^*-\tau_j^*|\delta_n\mathbb{E}\big\{I(\tau_m^*-\delta_n\leq {X}_i^\top{\beta}^*\leq \tau_m^*+\delta_n)\big\}\nonumber\\
    \leq & O(M_n\delta_n^2)+O(\delta_n^3)+O(M_n\delta_n^2).
\end{align*}
These facts, together with \eqref{Them1:proof:eq1:3}, imply
\begin{align}\label{Them1:proof:eq1:6}
    \mathbb{E}\big\{|g_{ni,j}({\theta}_{o}^*)-g_{i,j}({\theta}_{o}^*)|\big\}
    \le  O(M_n\delta_n^2)+O(\delta_n^2)=O(M_n^2\delta_n),~~j=1,\dots,M_n^*.
\end{align}
Furthermore, for each $M_n^*+1\le j\le s_n^*,$ we have
\begin{align}\label{Them1:proof:eq1:7}
\mathbb{E}\big\{\big|g_{ni,j}({\theta}_{o}^*)-g_{i,j}({\theta}_{o}^*)\big|\big\}\le O(M_n^2\delta_n).
\end{align}
The proof of \eqref{Them1:proof:eq1:7} is provided in Section \ref{secB}.

By Markov's inequality, it follows from \eqref{Them1:proof:eq1:6} and \eqref{Them1:proof:eq1:7} that
\begin{align*}
     \frac{1}{\sqrt{n}}\sum_{i=1}^n\big\{g_{ni,j}({\theta}_{o}^*)-g_{i,j}({\theta}_{o}^*)\big\}=O_p(\sqrt{n}M_n^2\delta_n).
\end{align*}
Thus, if $\sqrt{n}M_n^2\delta_n\to0$ as $n\rightarrow \infty,$ we obtain
\begin{align}\label{Them1:proof:eq1:8}
     \bigg\|\frac{1}{\sqrt{n}}\sum_{i=1}^n\big\{{g}_{ni}({\theta}_{o}^*)-{g}_{i}({\theta}_{o}^*)\big\}\bigg\|\leq O_p(\sqrt{ns_n}M_n^2\delta_n)=o_p(\sqrt{s_n}).
\end{align}
Furthermore, for each $1\le j\le s_n^*,$ by utilizing the central limit theorem, we can demonstrate
\[\frac{1}{\sqrt{n}}\sum_{i=1}^ng_{i,j}({\theta}_{o}^*)=O_p(1).\]
This implies
\begin{align}\label{Them1:proof:eq1:9}
    \bigg\|\frac{1}{\sqrt{n}}\sum_{i=1}^n{g}_i({\theta}_{o}^*)\bigg\|\leq O_p(\sqrt{s_n}).
\end{align}
Combining $\eqref{Them1:proof:eq1:8}$ and $\eqref{Them1:proof:eq1:9}$, we have
\begin{align}\label{Them1:proof:eq1:10}
    \bigg\|\frac{1}{\sqrt{n}}\frac{\partial \bar h_n({\theta}_{o}^*)}{\partial {\theta}_o}\bigg\|
    &=\bigg\|\frac{1}{\sqrt{n}}\sum_{i=1}^n{g}_{ni}({\theta}_{o}^*)-\frac{1}{\sqrt{n}}\sum_{i=1}^n{g}_i({\theta}_{o}^*)+\frac{1}{\sqrt{n}}\sum_{i=1}^n{g}_i({\theta}_{o}^*)\bigg\|\nonumber\\
    &\leq\bigg\|\frac{1}{\sqrt{n}}\sum_{i=1}^n{g}_{ni}({\theta}_{o}^*)-\frac{1}{\sqrt{n}}\sum_{i=1}^n{g}_i({\theta}_{o}^*)\bigg\|+\bigg\|\frac{1}{\sqrt{n}}\sum_{i=1}^n{g}_i({\theta}_{o}^*)\bigg\|\nonumber\\
    &=o_p(\sqrt{s_n})+O_p(\sqrt{s_n})=O_p(\sqrt{s_n}).
\end{align}
Then, we can get
\begin{align*}
    |I_1|=\rho_n\bigg|\bigg\{\frac{\partial\bar h_n({\theta}_{o}^*)}{\partial {\theta}_o}\bigg\}^\top {u}\bigg| \leq \rho_n\bigg\|\frac{\partial\bar h_n({\theta}_{o}^*)}{\partial {\theta}_o}\bigg\|\|{u}\|= O_p(\sqrt{ns_n}\rho_n)\|{u}\| = O_p(n\rho_n^2)\|{u}\|,
\end{align*}
and hence \eqref{Them1:proof:eq1} holds.

Next, we show
\begin{align}\label{Them1:proof:eq2}
    I_2\geq O_p(n\rho_n^2\|{u}\|^2/2).
\end{align}
Since for any ${a}\in\mathbb{R}^{s_n^*}$,
${a}^\top {V}({\theta}_{o}^*) {a}={a}^\top \mathbb{E}\big\{{H}_i({\theta}_{o}^*){H}_i({\theta}_{o}^*)^\top\big\}{a}>0$,
there exists a positive constant $\tilde{c}$ such that $\text{eig}_{\min}\big({V}({\theta}_{o}^*)\big)>\tilde{c},$ where $\text{eig}_{\min}({A})$ denotes the minimum eigenvalues of matrix ${A}$.
By Lemma 1 in Section \ref{secB}, we have
\begin{align*}
    I_2&
    =\frac{n\rho_n^2}{2}{u}^\top\bigg[\frac{1}{n}\frac{\partial^2 \bar h_n({\theta}_{o}^*)}{\partial {\theta}_o\partial{\theta}_o^\top}-{V}({\theta}_{o}^*)\bigg]{u} +\frac{n\rho_n^2}{2}{u}^\top {V}({\theta}_{o}^*){u}\nonumber\\
    &=\frac{n\rho_n^2}{2}{u}^\top {V}({\theta}_{o}^*){u}+o_p(1)n\rho_n^2\|{u}\|^2\nonumber\\
    &\geq  \tilde{c} n\rho_n^2\|{u}\|^2/2+o_p(n\rho_n^2\|{u}\|^2),
\end{align*}
which implies that \eqref{Them1:proof:eq2} holds.

For $I_3,$
we begin with the following decomposition
\begin{align}
    I_3\nonumber
    &=\underbrace{\frac{1}{6}\sum_{j,k,l=1}^{s_n^*}\frac{1}{n}\bigg\{\frac{\partial^3\bar h_n({\breve\theta}_o)}{\partial \theta_{oj} \partial\theta_{ok}\partial\theta_{ol}}-\mathbb{E}\frac{\partial^3\bar h_n({\theta}_{o}^*)}{\partial \theta_{oj} \partial\theta_{ok}\partial\theta_{ol}}\bigg\}u_ju_ku_ln\rho_n^3}_{I_{31}}\\
    &\quad+\underbrace{\frac{1}{6n}\sum_{j,k,l=1}^{s_n^*}\mathbb{E}\frac{\partial^3\bar h_n({\theta}_{o}^*)}{\partial \theta_{oj} \partial\theta_{ok}\partial\theta_{ol}}u_ju_ku_ln\rho_n^3}_{I_{32}}.\nonumber
\end{align}
By \eqref{Them1:proof:eq1:1} and \eqref{Them1:proof:eq1:2}, along with conditions 1-4,
we can show that for any ${\theta}_o$ satisfying $\|{\theta}_{o}-{\theta}_{o}^*\|\le \rho_n\|{u}\|,$
\begin{align}
   & \mathbb{E}\bigg(\frac{1}{n}\frac{\partial^3 \bar h_n({\theta}_o)}{\partial\theta_{oj} \partial\theta_{ok}\partial\theta_{ol}}\bigg)=\mathbb{E}\bigg(\frac{1}{n}\frac{\partial^3 \bar h_n({\theta}_{o}^*)}{\partial\theta_{oj} \partial\theta_{ok}\partial\theta_{ol}}\bigg)+o(1),\nonumber\\
   & \mathbb{VAR}\bigg(\frac{1}{n}\frac{\partial^3 \bar h_n({\theta}_o)}{\partial\theta_{oj} \partial\theta_{ok}\partial\theta_{ol}}\bigg)\leq O\Big(\frac{M_n^3}{n\delta_n}\Big).\nonumber
\end{align}
Therefore, if $s_n^3/(n\delta_n)\to0$ as $n\to\infty$, then according to Chebyshev's inequality, we can obtain
\begin{align}\label{Them1:proof:eq3:5}
\frac{1}{n}\bigg\{\frac{\partial^3\bar h_n({\breve\theta}_o)}{\partial \theta_{oj} \partial\theta_{ok}\partial\theta_{ol}}-\mathbb{E}\frac{\partial^3\bar h_n({\theta}_{o}^*)}{\partial \theta_{oj} \partial\theta_{ok}\partial\theta_{ol}}\bigg\}=o_p(1),
\end{align}
which implies that $|I_{31}|\le o_p(s_n^{3/2}n\rho_n^3\|{u}\|^3) .$

For $I_{32},$ we have
\begin{align*}
    |I_{32}|
    &\le \frac{1}{6}\bigg[\sum_{j,k,l=1}^{s_n}\bigg\{\frac{1}{n}\mathbb{E}\frac{\partial^3\bar h_n({\theta}_{o}^*)}{\partial \theta_{oj} \partial\theta_{ok}\partial\theta_{ol}}\bigg\}^2\bigg]^{1/2}\|{u}\|^3n\rho_n^3\\
    & \le O(ns_n^{3/2}M_n^2\rho_n^3\|{u}\|^3),
\end{align*}
where the first inequality holds due to the Cauchy-Schwarz inequality,
and the second inequality holds because $n^{-1}\mathbb{E}\{\partial^3 \bar h_n({\theta}_{o}^*)/\partial\theta_{oj}\partial\theta_{ok}\partial\theta_{ol}\}< O(M_n^2).$
Thus, $|I_{32}|=o(n\rho_n^2)\|{u}\|^2$ if $s_n^4/\sqrt{n}\to0$ and $s_n^4a_n\to0$ as $n\rightarrow\infty.$
These facts implies
\begin{align}\label{Them1:proof:eq3}
        |I_3|=o_p(n\rho_n^2\|{u}\|^2).
\end{align}
By the definition of $a_n,$
we can show
\begin{align}\label{Them1:proof:eq4}
        |I_4|=\bigg|\sum_{m=1}^{M_n^*}n\rho_np'_{\lambda_n,t}(|\alpha_m^*|)\text{sgn}(\alpha_m^*)u_m\Big| \leq \sqrt{M_n^*}n\rho_na_n\|{u}\|\leq n\rho_n^2\|{u}\|.
\end{align}
For $I_5,$ since $b_n\to0$ as $n\to\infty$, we have
\begin{align}\label{Them1:proof:eq5}
    |I_5|=\bigg|\frac{1}{2}\sum_{m=1}^{M_n^*}n\rho_n^2p''_{\lambda_n,t}(|\alpha_m^*|)u_m^2\{1+o(1)\}\bigg| \leq o(n\rho_n^2)\|{u}\|^2.
\end{align}
Combining  \eqref{Them1:proof:eq1}, \eqref{Them1:proof:eq2}, and \eqref{Them1:proof:eq3}-\eqref{Them1:proof:eq5},
we obtain that $I_1,I_3,I_4$ and $I_5$ are dominated by $I_2$ for some sufficiently large $C_0$.
Thus, the sign of $\mathcal{O}_n({u})$ is determined by that of $I_2$,
which is positive.
This proves that \eqref{Them1:proof:eq0:1} holds, and  therefore $\|{\widehat \theta}_1-{\theta}_{o}^*\|=O_p(\rho_n).$

Note that when $\alpha_m=0~(m=M_{n}^*+1,\dots, M_n),$
$L_n({\theta})$ is free of ${\widetilde \tau}=(\tau_{M_n^*+1},\dots,\tau_{M_n})^\top$
and $\bar L_{n}({\theta}_o)=L_n({\theta}).$
Hence, for any ${\widetilde \tau}\in \mathbb{R}^{M_n-M_n^*},$
${\widehat\theta}({\widetilde\tau})=({\widehat\theta}_1^\top,{\widetilde \tau}^\top,{0}_{M_n-M_n^*}^\top)^\top$
is also a local minimizer of $L_{n}({\theta})$ in the subspace $\mathcal{B},$
and satisfies that $\|{\widehat\theta}({\widetilde\tau})-{\theta}^*({\widetilde\tau})\|=O_p(\rho_n),$
where ${\theta}^*({\widetilde\tau})=(({\theta}_{o}^*)^\top,{\widetilde \tau}^\top,{0}_{M_n-M_n^*}^\top)^\top.$

{\it Step 2.} We show that with probability tending to 1,
the local minimizer ${\widehat{\theta}}({\widetilde \tau})$ defined in Step 1
is also the local minimizer of $L_n({\theta})$ for any ${\widetilde\tau}$.
For this purpose, let $\mathbb{N}_C=\{\theta: \|{\theta}-{\widehat\theta}({\widetilde\tau})\|\le C_2\rho_n\}\cap\mathcal{B}$,
and $\mathbb{N}_1\subset\mathbb{R}^{s_n}$ be a sufficiently small ball around ${\widehat\theta}({\widetilde\tau})$,
satisfying that for any ${\theta}\in\mathbb{N}_1$, $\|{\theta}-{\widehat\theta}({\widetilde\tau})\|=O(\rho_n)$
and $\mathbb{N}_1\cap \mathcal{B}\subset \mathbb{N}_C$.
For simplicity, write ${\widehat\theta}={\widehat\theta}({\widetilde\tau})$ and ${\theta}^*={\theta}^*({\widetilde\tau})$.

It suffices to show that for any ${w}_1\in \mathbb{N}_1-\mathbb{N}_C,$ $L_n({w}_1)>L_n({\widehat\theta})$.
We begin with the following decomposition
\begin{align*}
    L_n({w}_1)-L_n({\widehat\theta})=L_n({w}_1)-L_n({w}_2)+L_n({w}_2)-L_n({\widehat\theta}),
\end{align*}
where ${w}_2$ is the projection of ${w}_1$ on subspace $\mathcal{B}$.
By the definition of ${\widehat \theta}$, we have that ${\widehat\theta}$ is a local minimizer of $L_n({\theta})$
constrained on subspace $\mathcal{B},$ and hence $L_n({w}_2) > L_n({\widehat\theta})$ if ${w}_2\neq \widehat\theta$.
Thus, we need to show that $L_n({w}_1)-L_n({w}_2)>0$.

Note that $w_{1j}=w_{2j}$ for $1\le j\le \tilde s_n,$ where $\tilde s_n=1+d_1+d_2+M_n+M_n^*,$
and $w_{kj}$ is the $j$th element of ${w}_k$ for $k=1, 2$.
Therefore, by the mean-value theorem, we have
\[
L_n({w}_1)-L_n({w}_2)=\sum_{j=\tilde s_n+1}^{s_n}\frac{\partial L_n({w}_0)}{\partial w_{j}}w_{1j},
\]
where ${w}_0$ lies on the line segment joining ${w}_1$ and ${w}_2.$
We show that the term on the right hand side of the above equation is positive.
It suffices to show that for each $j=\tilde s_n+1,\ldots,s_n$, with probability tending to 1, it holds
\begin{align}
    \frac{\partial L_n({w})}{\partial w_j} & < 0 \quad \text{for} -\varepsilon_n<w_j<0,\label{Them1:proof:2:eq0:1} \\
   \text{and}~~\frac{\partial L_n({w})}{\partial w_j} & >0 \quad \text{for}\quad 0<w_j<\varepsilon_n\label{Them1:proof:2:eq0:2},
\end{align}
where $\varepsilon_n=C_3\sqrt{s_n/n},$ and $C_3$ denotes a positive constant.
Here ${w}$ is a vector satisfying $\|{w}-{\theta}^*\|\le O(\sqrt{s_n/n})$, and $w_j$ is the $j$th element of ${w}$.

We first show \eqref{Them1:proof:2:eq0:2}.
By the Taylor expansion of $\partial L_n({w})/\partial w_{j}$ at ${\theta}^*$, we have
\begin{align}\label{Them1:proof:2:eq0:3}
    \frac{\partial L_n({w})}{\partial w_j}
    &=\underbrace{\frac{\partial h_n({\theta}^*)}{\partial w_j}}_{I_{6j}}
    +\underbrace{\sum_{k=1}^{s_n}\frac{\partial^2 h_n({\theta}^*)}{\partial w_{j}\partial w_{k}}(w_{k}-\theta_{k}^*)}_{I_{7j}}\nonumber\\
    &\quad+\underbrace{\sum_{k,l=1}^{s_n}\frac{\partial^3 h_n({\breve\theta})}{\partial w_{j}\partial w_{k}\partial w_{l}}(w_{k}-\theta_{k}^*)(w_{l}-\theta_{l}^*)}_{I_{8j}}
    +np'_{\lambda_n,t}(|w_{j}|)\text{sgn}(w_{j}),
\end{align}
where ${\breve \theta}$ lies between ${w}$ and ${\theta}^*$.

For each $j=\tilde s_n+1,\dots,s_n$, using the arguments similar to the proof of \eqref{Them1:proof:eq1:10}
and considering the fact that  $\sqrt{n}M_n^2\delta_n/s_n\to0$, we can conclude
\begin{align}\label{Them1:proof:2:eq1}
    I_{6j}=o_p(\sqrt{n}s_n).
\end{align}
For $I_{7j}$, using the arguments similar to the proof of Lemma 1 in Section \ref{secB}, we can show
\begin{align}\label{Them1:proof:2:eq2:1}
    \frac{1}{n}\frac{\partial^2 h_n({\theta}^*)}{\partial w_{j}\partial w_{k}}-\widetilde V_{j,k}({\theta}^*)=o_p(1), ~~k=1,\dots,s_n,
\end{align}
where $\widetilde V_{j,k}({\theta}^*)$ denotes the $(j,k)$th element of ${\widetilde V}({\theta}^*)=\mathbb{E}\big\{{\widetilde H}_i({\theta}^*){\widetilde H}_i({\theta}^*)^\top\big\}$ with
\begin{align*}
    {\widetilde H}_i({\theta}^*)=&\Big(f({X}_i^\top{\beta}^*,\tau_{1}^*),\ldots,f({X}_i^\top{\beta}^*,\tau_{M_n}^*),-\alpha_{1}^*I{({X}_i^\top{\beta}^*>\tau_{1}^*)},\dots,-\alpha_{M_n}^*I{({X}_i^\top{\beta}^*>\tau_{M_n}^*)},\\
    &\hspace{2in}{X}_i^\top{\beta}^*,\Big[\alpha_{0}^*+\sum_{m=1}^{M_n}\alpha_{m}^* I({X}_i^\top{\beta}^*>\tau_{m}^*)\Big]{\widetilde X}_i^\top,\widetilde{Z}_i^\top\Big)^\top.
\end{align*}
By \eqref{Them1:proof:2:eq2:1} and the fact that $\|{w}-{\theta}^*\|\
\leq O(\sqrt{s_n/n})$, we have
\begin{align}\label{Them1:proof:2:eq2}
    I_{7j}
    &=\sum_{k=1}^{s_n}\bigg\{\frac{\partial^2 h_n({\theta}^*)}{\partial w_{j}\partial w_{k}}-n\widetilde V_{j,k}({\theta}^*)\bigg\}(w_{k}-\theta_{k}^*)
    +\sum_{k=1}^{s_n}n\widetilde V_{j,k}({\theta}^*)(w_{k}-\theta_{k}^*)\nonumber\\
    &\leq\bigg[\sum_{k=1}^{s_n}\bigg\{\frac{\partial^2 h_n({\theta}^*)}{\partial w_{j}\partial w_{k}}-n\widetilde V_{j,k}({\theta}^*)\bigg\}^2\bigg]^{1/2}\|{w}-{\theta}^*\|+n\bigg\{\sum_{k=1}^{s_n}\widetilde V_{j,k}^2({\theta}^*)\bigg\}^{1/2}\|{w}-{\theta}^*\|\nonumber\\
    &=o_p(\sqrt{n}s_n)+O(\sqrt{n}s_n)=O_p(\sqrt{n}s_n),
\end{align}
where the inequality holds due to the Cauchy–Schwarz inequality.
Similarly, we can show
\begin{align}\label{Them1:proof:2:eq3}
    I_{8j}
    &=\sum_{k,l=1}^{s_n}\bigg\{\frac{\partial^3 h_n({\breve\theta})}{\partial w_{j}\partial w_{k}\partial w_{l}}
    -\mathbb{E}\bigg(\frac{\partial^3 h_n({\theta}^*)}{\partial w_{j}\partial w_{k}\partial w_{l}}\bigg)\bigg\}(w_{k}-\theta_{k}^*)(w_{l}-\theta_{l}^*)\nonumber\\
    &\quad+\sum_{k,l=1}^{s_n}\mathbb{E}\bigg(\frac{\partial^3 h_n({\theta}^*)}{\partial w_{j}\partial w_{k}\partial w_{l}}\bigg)(w_{k}-\theta_{k}^*)(w_{l}-\theta_{l}^*)\nonumber\\
    &\leq o_p(ns_n)\|{w}-{\theta}^*\|^2+O(ns_nM_n^2)\|{w}-{\theta}^*\|^2\nonumber\\
    &=o_p(s_n^2)+O(s_n^2M_n^2)=o_p(\sqrt{n}s_n).
\end{align}
Combining \eqref{Them1:proof:2:eq0:3}-\eqref{Them1:proof:2:eq3}, we obtain
\begin{align*}
     \frac{\partial L_n({w})}{\partial w_{j}}
     =n\sqrt{s_n}\lambda_n\Big\{\lambda_{n}^{-1}p'_{\lambda_n,t}(|w_{j}|)\text{sgn}(w_{j})
     +O_p\big(\sqrt{s_n/(n\lambda_n^2)}\big)\Big\}.
\end{align*}
Since $\sqrt{s_n/(n\lambda_n^2)}\to0$ and $\lambda_n^{-1}p'_{\lambda_n,t}(0^+)>0,$ we have that for all sufficiently large $n$,
\begin{align*}
     \frac{\partial L_n({w})}{\partial w_{j}}
     > n\sqrt{s_n}p'_{\lambda_n,t}(|w_{j}|)\text{sgn}(w_{j})/2,
\end{align*}
which implies that the sign of $\partial L_n({w})/\partial w_{j}$ is determined by that of $w_{j}$.
Hence, we have that \eqref{Them1:proof:2:eq0:2} holds.
Similarly, we can show that \eqref{Them1:proof:2:eq0:1} holds.
Note that ${w}_1\in\mathbb{N}_1-\mathbb{N}_C$ and $\mathbb{N}_1\cap \mathcal{B}\subset \mathbb{N}_C$.
Thus, there exists at least one $\tilde s_n+1\le j\le s_n$ such that $w_{1j}\neq 0,$
and the sign of $w_{0j}$ is the same as that of $w_{1j}$.
These facts imply
\[
L_n({w}_1)-L_n({w}_2)=\sum_{j=\tilde s_n+1}^{s_n}\frac{\partial L_n({w}_0)}{\partial w_{j}}w_{1j}\\
> n\sqrt{s_n}p'_{\lambda_n,t}(|w_{0j}|)|w_{1j}|/2\ge 0
\]
for all sufficient large $n$ and all ${\widetilde \tau}\in \mathbb{R}^{M_n-M_n^*}.$
This completes the proof of Theorem 1.

\subsection{Proof of Theorem 2}
By Theorem 1, we know that there exists a local minimizer ${\widehat\theta}$ of $L_n({\theta})$ such that ${\widehat\theta}$ is $(n/s_n)^{1/2}$-consistent.
To prove Theorem 2, we first show
\begin{align}\label{Them2:proof:eq1}
    \big\{ {V}({\theta}_{o}^*)+{\Sigma}_{\lambda_n}({\theta}_{o}^*)\big\}({\widehat \theta}_{1}-{\theta}_{o}^*)+{B}({\theta}_{o}^*)
    =-\frac{1}{n}\sum_{i=1}^n{g}_i({\theta}_{o}^*)+o_p(n^{-1/2}).
\end{align}
By the Taylor expansion of $\partial \bar L_n({\widehat\theta}_{1})/\partial {\theta}_o$ at ${\theta}_{o}^*$, we have
\begin{align}\label{Them2:proof:eq1:1}
    &\frac{1}{n}\bigg[\Big\{\frac{\partial^2 \bar h_n({\theta}_{o}^*)}{\partial {\theta}_{o}\partial{\theta}_{o}^\top}+n{\Sigma}_{\lambda_n}({\breve\theta}_{o})\Big\}({\widehat\theta}_{1}-{\theta}_{o}^*)+n{B}({\theta}_{o}^*)\bigg]\nonumber\\
    =&-\frac{1}{n}\frac{\partial \bar h_n({\theta}_{o}^*)}{\partial {\theta}_{o}}-\frac{1}{2n}({\widehat\theta}_{1}-{\theta}_{o}^*)^\top\frac{\partial^2}{\partial{\theta}_{o}\partial{\theta}_{o}^\top}\Big(\frac{\partial \bar h_n({\breve\theta}_{o})}{\partial {\theta}_{o}}\Big)({\widehat\theta}_{1}-{\theta}_{o}^*),
\end{align}
where ${\breve\theta}_{o}$ lies between ${\widehat\theta}_{1}$ and ${\theta}_{o}^*$.
For simplicity of presentation, we define
\begin{align*}
    {\Upsilon}_1&=\frac{\partial^2 \bar h_n({\theta}_{o}^*)}{\partial {\theta}_{o}\partial{\theta}_{o}^\top}+n{\Sigma}_{\lambda_n}({\breve\theta}_{o}),\\
\text{and}~~{\Upsilon}_2&=\frac{1}{2}({\widehat\theta}_{1}-{\theta}_{o}^*)^\top\frac{\partial^2}{\partial{\theta}_{o}\partial{\theta}_{o}^\top}\Big(\frac{\partial \bar h_n({\breve\theta}_{o})}{\partial {\theta}_{o}}\Big)({\widehat\theta}_{1}-{\theta}_{o}^*).
\end{align*}
By Lemma 1 and condition 6, we have
\begin{align*}
    \bigg\|\frac{1}{n}{\Upsilon}_1-\big\{{V}({\theta}_{o}^*)+{\Sigma}_{\lambda_n}({\theta}_{o}^*)\big\}\bigg\|_\mathrm{F}=o_p\Big(\frac{1}{\sqrt{s_n}}\Big).
\end{align*}
This, together with $\|{\widehat\theta}_{1}-{\theta}_{o}^*\|=O_p(\sqrt{s_n/n})$, implies
\begin{align}\label{Them2:proof:eq1:2}
    &\bigg\|\bigg\{\frac{1}{n}{\Upsilon}_1-{V}({\theta}_{o}^*)-{\Sigma}_{\lambda_n}({\theta}_{o}^*)\bigg\}({\widehat\theta}_{1}-{\theta}_{o}^*)\bigg\|\nonumber\\
    \leq & \bigg\|\frac{1}{n}{\Upsilon}_1-{V}({\theta}_{o}^*)-{\Sigma}_{\lambda_n}({\theta}_{o}^*)\bigg\|_\mathrm{F}\times\|{\widehat\theta}_{1}-{\theta}_{o}^*\|=o_p\Big(\frac{1}{\sqrt{n}}\Big).
\end{align}
By applying the Cauchy–Schwarz inequality to ${\Upsilon}_2$, we have
\begin{align}\label{Them2:proof:eq1:3}
    \bigg\|\frac{1}{n}{\Upsilon}_2\bigg\|^2
    &\leq\frac{1}{2}\|{\widehat\theta}_{1}-{\theta}_{o}^*\|^4\sum_{j,k,l=1}^{s_n^*}\bigg\{\frac{1}{n}\frac{\partial \bar h_n({\breve\theta}_{o})}{\partial \theta_{oj}\partial \theta_{ok}\partial \theta_{ol}}-\frac{1}{n}\mathbb{E}\frac{\partial \bar h_n({\breve\theta}_{o})}{\partial \theta_{oj}\partial \theta_{ok}\partial \theta_{ol}}\bigg\}^2\nonumber\\
    &\quad+\frac{1}{2}\|{\widehat\theta}_{1}-{\theta}_{o}^*\|^4\sum_{j,k,l=1}^{s_n^*}\bigg\{\frac{1}{n}\mathbb{E}\frac{\partial \bar h_n({\breve\theta}_{o})}{\partial \theta_{oj}\partial \theta_{ok}\partial \theta_{ol}}\bigg\}^2\nonumber\\
    &=O_p\Big(\frac{s_n^2}{n^2}\Big)o_p(s_n^3)+O_p\Big(\frac{s_n^2}{n^2}\Big)O(s_n^3M_n^4)=o_p\Big(\frac{1}{n}\Big).
\end{align}
In addition, it follows the proof of $I_1$ in Theorem 1 that
\begin{align}\label{Them2:proof:eq1:4}
        \frac{1}{n}\frac{\partial \bar h_n({\theta}_{o}^*)}{\partial {\theta}_{o}}
        =\frac{1}{n}\sum_{i=1}^n{g}_{i}({\theta}_{o}^*)+o_p\Big(\frac{1}{\sqrt{n}}\Big).
\end{align}
Combining \eqref{Them2:proof:eq1:1}-\eqref{Them2:proof:eq1:4}, we have that \eqref{Them2:proof:eq1} holds.

For simplicity, we write ${\Sigma}_{\lambda_n}^*={\Sigma}_{\lambda_n}({\theta}_{o}^*),$
${B}^*={B}({\theta}_{o}^*)$ and ${V}_*={V}({\theta}_{o}^*)$.
By \eqref{Them2:proof:eq1}, we consider the following equations:
\begin{align}\label{Them2:proof:eq2}
    &\sqrt{n}{A}_n{V}_*^{-1/2}\big\{ {V}_*+{\Sigma}_{\lambda_n}^*\big\}\Big[({\widehat \theta}_{1}-{\theta}_{o}^*)+\big\{ {V}_*+{\Sigma}_{\lambda_n}^*\big\}^{-1}{B}^*\Big]\nonumber\\
    =& -\frac{1}{\sqrt{n}}{A}_n{V}_*^{-1/2}\sum_{i=1}^n{g}_i({\theta}_{o}^*)+o_p\Big({A}_n{V}_*^{-1/2}\Big).
\end{align}
By the conditions in Theorem 2, it can be directly shown that the last term is $o_p(1)$.

Define
\begin{align*}
    \mathcal{Y}_{ni}=-\frac{1}{\sqrt{n}}{A}_n{V}_*^{-1/2}{g}_i({\theta}_{o}^*).
\end{align*}
Next, we show that $\mathcal{Y}_{ni}$ satisfies the Lindeberg–Feller conditions.
It follows that for any $\varepsilon>0,$
\begin{align}\label{Them2:proof:eq2:1}
    \sum_{i=1}^n\mathbb{E}\|\mathcal{Y}_{ni}\|^2I\big(\|\mathcal{Y}_{ni}\|>\varepsilon\big)\leq n\big\{\mathbb{E}\|\mathcal{Y}_{ni}\|^4\big\}^{1/2}\big\{\mathbb{P}\big(\|\mathcal{Y}_{ni}\|>\varepsilon\big)\big\}^{1/2}.
\end{align}
By the conditions in Theorem 2,  we have
\begin{align*}
    \mathbb{E}\|\mathcal{Y}_{n1}\|^4
    &=\frac{1}{n^2}\mathbb{E}\big\|{A}_n{V}_*^{-1/2}{g}_1({\theta}_{o}^*)\big\|^4\\
    &=\frac{1}{n^2}\mathbb{E}\big\|\big({g}_1({\theta}_{o}^*)\big)^\top \big({V}_*^{-1/2}\big)^\top {A}_n^\top {A}_n{V}_*^{-1/2}{g}_1({\theta}_{o}^*)\big\|^2\\
    &\leq \frac{1}{n^2}\text{eig}_{\max}({A}_n^\top {A}_n)\text{eig}_{\max}\big({V}_*^{-1}\big)\mathbb{E}\big\|\big({g}_1({\theta}_{o}^*)\big)^\top {g}_1({\theta}_{o}^*)\big\|^2\\
    &=O\Big(\frac{s_n^2}{n^2}\Big),
\end{align*}
and
\begin{align*}
    \mathbb{P}\big(\|\mathcal{Y}_{n1}\|>\varepsilon\big)\leq\frac{\mathbb{E}\big\|{A}_n{V}_*^{-1/2}{g}_1({\theta}_{o}^*)\big\|^2}{n\varepsilon^2}
    =O\Big(\frac{s_n}{n}\Big).
\end{align*}
Thus, by \eqref{Them2:proof:eq2:1} and the condition $s_n^3/n\to0,$ we have
\begin{align*}
    \sum_{i=1}^n\mathbb{E}\|\mathcal{Y}_{ni}\|^2I\big(\|\mathcal{Y}_{ni}\|>\varepsilon\big)=O\Big(n\frac{s_n}{n}\sqrt{\frac{s_n}{n}}\Big)=o(1).
\end{align*}
On the other hand, since ${A}_n{A}_n^\top \to {G}$ and $\mathbb{COV}\big\{{g}_1({\theta}_{o}^*)\big\}=\sigma^2{V}_*$,
we have
\begin{align*}
    \sum_{i=1}^n\mathbb{COV}(\mathcal{Y}_{ni})
    &=n\mathbb{COV}(\mathcal{Y}_{n1})=\mathbb{COV}\big\{{A}_n{V}_*^{-1/2}{g}_1({\theta}_{o}^*)\big\}\\
    &={A}_n{V}_*^{-1/2}\mathbb{COV}\big\{{g}_1({\theta}_{o}^*)\big\}\big({V}_*^{-1/2}\big)^\top {A}_n^\top\\
    &=\sigma^2{A}_n{V}_*^{-1/2}{V}_*\big({V}_*^{-1/2}\big)^\top {A}_n^\top\\
    &=\sigma^2{A}_n{A}_n^\top\to \sigma^2{G}.
\end{align*}
These facts imply that $\mathcal{Y}_{ni}$ satisfies the Lindeberg–Feller conditions.
Therefore, we obtain that $n^{-1/2}{A}_n{V}_*^{-1/2}\sum_{i=1}^n{g}_i({\theta}_{o}^*)$ converges in distribution to a multivariate normal random variable
with mean ${0}_q$ and covariance matrix $\sigma^2{G}$. This completes the proof of Theorem 2.

\subsection{Proof of Theorem 3}
Define $\mathcal{A}_n={V}_n({\widehat\theta}_1)+{\Sigma}_{\lambda_n}({\widehat\theta}_{1})$,
$\mathcal{A}={V}({\theta}_{o}^*)+{\Sigma}_{\lambda_n}({\theta}_{o}^*)$,
$\mathcal{B}_n={V}_n({\widehat\theta}_1)$ and $\mathcal{B}={V}({\theta}_{o}^*)$.
Then, we can write $n({\widehat\Xi}_n-{\Xi})$ as
\begin{align}\label{Them3:proof:eq01}
    n({\widehat\Xi}_n-{\Xi})=&\widehat\sigma^2\mathcal{A}_n^{-1}\mathcal{B}_n\mathcal{A}_n^{-1}-\sigma^2\mathcal{A}^{-1}\mathcal{B}\mathcal{A}^{-1}\nonumber\\
    =&\widehat\sigma^2\mathcal{A}_n^{-1}(\mathcal{B}_n-\mathcal{B})\mathcal{A}_n^{-1}+(\widehat\sigma^2-\sigma^2)\mathcal{A}_n^{-1}\mathcal{B}\mathcal{A}_n^{-1}\nonumber\\
    &+\sigma^2(\mathcal{A}_n^{-1}-\mathcal{A}^{-1})\mathcal{B}\mathcal{A}_n^{-1}+\sigma^2\mathcal{A}^{-1}\mathcal{B}(\mathcal{A}_n^{-1}-\mathcal{A}^{-1}),
\end{align}
and $\mathcal{A}_n^{-1}-\mathcal{A}^{-1}$ as
\begin{align}\label{Them3:proof:eq02}
    \mathcal{A}_n^{-1}-\mathcal{A}^{-1}=\mathcal{A}_n^{-1}(\mathcal{A}-\mathcal{A}_n)\mathcal{A}^{-1}.
\end{align}
Denote $\text{eig}_i({A})$ as the $i$th largest eigenvalue of a symmetric matrix ${A}$.
In what follows, we show
\begin{align}\label{Them3:proof:eq1}
    \text{eig}_i[n({\widehat\Xi}_n-{\Xi})]=o_p(1),
\end{align}
which indicates that ${\widehat\Xi}_n$ is a consistent estimator of ${\Xi}$.
Note that $|\text{eig}_i(\mathcal{A})|$ and $|\text{eig}_i(\mathcal{B})|$ are infinite and uniformly bounded away from 0.
Thus, in view of \eqref{Them3:proof:eq01} and \eqref{Them3:proof:eq02}, to prove \eqref{Them3:proof:eq1},
we need to demonstrate that $\text{eig}_i(\mathcal{A}_n-\mathcal{A})=o_p(1)$,
$\text{eig}_i(\mathcal{B}_n-\mathcal{B})=o_p(1)$ and $\widehat\sigma^2-\sigma^2=o_p(1)$.

We first show
\begin{align}\label{Them3:proof:eq6}
    \text{eig}_i(\mathcal{A}_n-\mathcal{A})=o_p(1).
\end{align}
We consider the following decomposition
\begin{align*}
    \mathcal{A}_n-\mathcal{A}&=\underbrace{{V}_n({\widehat\theta}_1)-{V}({\theta}_{o}^*)}_{I_9}+\underbrace{{\Sigma}_{\lambda_n}({\widehat\theta}_{1})-{\Sigma}_{\lambda_n}({\theta}_{o}^*)}_{I_{10}}.
\end{align*}
Since
\begin{align*}
    \text{eig}_{\min}(I_9)+\text{eig}_{\min}(I_{10})&\leq \text{eig}_{\min}(I_9+I_{10})\\
    &\leq \text{eig}_{\max}(I_9+I_{10})\leq \text{eig}_{\max}(I_{9})+ \text{eig}_{\max}(I_{10}),
\end{align*}
we can consider $\text{eig}_i(I_9)$ and $\text{eig}_i(I_{10})$ separately.
For $I_9$, since $\text{eig}_i^2(I_9)\leq\text{eig}_{\max}^2(I_9)\leq \|I_9\|_\mathrm{F}^2$,
it suffices to show
\begin{align}\label{Them3:proof:eq3}
    \|I_9\|_\mathrm{F}=o_p(1).
\end{align}
Note that
\begin{align*}
    I_9={V}_n({\widehat\theta}_1)-{V}({\theta}_{o}^*)=\underbrace{{V}_n({\widehat\theta}_1)-{V}_n({\theta}_{o}^*)}_{I_{91}}+\underbrace{{V}_n({\theta}_{o}^*)-{V}({\theta}_{o}^*)}_{I_{92}}.
\end{align*}
Using the arguments similar to the proof of \eqref{Them1:proof:eq3:5} in Theorem 1,
we can show
\begin{align*}
    \|I_{91}\|_\mathrm{F}^2&=\sum_{j=1}^{s_n^*}\sum_{k=1}^{s_n^*}\big\{V_{nj,k}({\widehat\theta}_1)-V_{nj,k}({\theta}_{o}^*)\big\}^2\\
    &=\sum_{j=1}^{s_n^*}\sum_{k=1}^{s_n^*}\bigg\{\sum_{l=1}^{s_n^*}\frac{\partial V_{nj,k}({\breve\theta}_o)}{\partial \theta_{ol} }(\widehat\theta_{1l}-\theta_{ol}^*)\bigg\}^2\\
    &\leq\sum_{j=1}^{s_n^*}\sum_{k=1}^{s_n^*}\sum_{l=1}^{s_n^*}\bigg\{\frac{\partial V_{nj,k}({\breve\theta}_o)}{\partial \theta_{ol} }\bigg\}^2
    \|{\widehat\theta}_1-{\theta}_{o}^*\|^2\\
    &=O_p\Big(\frac{s_n^4M_n^4}{n}\Big)+o_p\Big(\frac{s_n^4}{n}\Big),
\end{align*}
where ${\breve\theta}_{o}$ lies between ${\widehat\theta}_{1}$ and ${\theta}_{o}^*$.
Thus, if $s_n^8/n\to0$ as $n\to\infty$, we have
\begin{align}\label{Them3:proof:eq3:1}
    \|I_{91}\|_\mathrm{F}^2=O_p(s_n^4M_n^4/n)=o_p(1).
\end{align}
For $I_{92}$, by Lemma 1 in Section \ref{secB}, we obtain $\|I_{92}\|_\mathrm{F}=o_p(1).$
This, together with \eqref{Them3:proof:eq3:1}, implies \eqref{Them3:proof:eq3}.

In addition, by condition 6 and the fact that $\|{\widehat\theta}_1-{\theta}_{o}^*\|=O_p(\sqrt{s_n/n})$, we can show
\begin{align}\label{Them3:proof:eq4}
    \|I_{10}\|_\mathrm{F}&=\|{\Sigma}_{\lambda_n}({\widehat\theta}_{1})-{\Sigma}_{\lambda_n}({\theta}_{o}^*)\|_\mathrm{F}\nonumber\\
    &\leq \kappa_2\|{\widehat\theta}_1-{\theta}_{o}^*\|=O_p(\sqrt{s_n/n})=o_p(1).
\end{align}
Therefore, by \eqref{Them3:proof:eq3} and \eqref{Them3:proof:eq4}, we obtain \eqref{Them3:proof:eq6}.
Similarly, we can show that $\|\mathcal{B}_n-\mathcal{B}\|_\mathrm{F}=\|{V}_n({\widehat\theta}_1)-{V}({\theta}_{o}^*)\|_\mathrm{F}=o_p(1)$,
which implies
\begin{align}\label{Them3:proof:eq3:B}
    \text{eig}_i(\mathcal{B}_n-\mathcal{B})=o_p(1).
\end{align}

Next, we show
\begin{align}\label{Them3:proof:eq5}
    \widehat\sigma^2-\sigma^2=o_p(1).
\end{align}
By the definitions of $\widehat \sigma^2$ and $\sigma^2$, we can write
\begin{align*}
    \widehat\sigma^2-\sigma^2=\underbrace{\frac{1}{n}\sum_{i=1}^n\big\{T_{ni}^2({\widehat\theta}_1)-T_{ni}^2({\theta}_{o}^*)\big\}}_{I_{11}}+\underbrace{\frac{1}{n}\sum_{i=1}^nT_{ni}^2({\theta}_{o}^*)-\mathbb{E}T_{i}^2({\theta}_{o}^*)}_{I_{12}},
\end{align*}
where $T_{ni}({\widehat\theta}_1)$ and $T_{i}({\theta}_{o}^*)$ are defined in the proof of Theorem 1.
For $I_{11}$, by the Taylor expansion and the Cauchy-Schwarz inequality, we have
\begin{align*}
    I_{11}^2&=\bigg\{\frac{1}{n}\sum_{i=1}^n\sum_{j=1}^{s_n^*}\frac{\partial T_{ni}^2({\breve \theta}_o)}{\partial \theta_{oj}}(\widehat\theta_{1j}-\theta_{oj}^*)\bigg\}^2
    \leq\sum_{j=1}^{s_n^*}\bigg\{\frac{1}{n}\sum_{i=1}^n\frac{\partial T_{ni}^2({\breve \theta}_o)}{\partial \theta_{oj}}\bigg\}^2\|{\widehat\theta}_{1}-{\theta}_{o}^*\|^2,
\end{align*}
where ${\breve\theta}_{o}$ lies between ${\widehat\theta}_{1}$ and ${\theta}_{o}^*$.
Using the arguments similar to the proof of \eqref{Them1:proof:eq3:5}, we have
\begin{align}\label{Them3:proof:eq5:1}
    I_{11}^2\le&\sum_{j=1}^{s_n^*}\bigg\{\frac{1}{n}\sum_{i=1}^n\frac{\partial T_{ni}^2({\breve \theta}_o)}{\partial \theta_{oj}}\bigg\}^2\|{\widehat\theta}_{1}-{\theta}_{o}^*\|^2\nonumber\\
    \leq&2\sum_{j=1}^{s_n^*}\bigg(\mathbb{E}\frac{\partial T_1^2({\theta}_{o}^*)}{\partial \theta_{oj}}\bigg)^2\|{\widehat\theta}_{1}-{\theta}_{o}^*\|^2
    +2\sum_{j=1}^{s_n^*}\bigg\{\frac{1}{n}\sum_{i=1}^n\frac{\partial T_{ni}^2({\breve \theta}_o)}{\partial \theta_{oj}}-\mathbb{E}\bigg(\frac{\partial T_1^2({\theta}_{o}^*)}{\partial \theta_{oj}}\bigg)\bigg\}^2\|{\widehat\theta}_{1}-{\theta}_{o}^*\|^2\nonumber\\
    \leq & O_p\Big(\frac{s_n^2}{n}\Big)+o_p\Big(\frac{s_n^2}{n}\Big)=o_p(1)
\end{align}
with $s_n^2/n\to0$ as $n\to\infty$.
For $I_{12}$, by \eqref{Them1:proof:eq1:1}, \eqref{Them1:proof:eq1:2} and conditions 1-4, we have
\begin{align*}
    &\big|\mathbb{E}T_{ni}^2({\theta}_{o}^*)-\mathbb{E}T_i^2({\theta}_{o}^*)\big|\\
    =&\bigg|\mathbb{E}\bigg[\epsilon_i-\sum_{m=1}^{M_n^*}\alpha_m^*\Big\{q_n({X}_i^\top{\beta}^*,\tau_m^*)-f({X}_i^\top{\beta}^*,\tau_m^*)\Big\}\bigg]^2-\mathbb{E}\epsilon_i^2\bigg|\\
    =&\bigg|\mathbb{E}\bigg[\sum_{m=1}^{M_n^*}\alpha_m^*\Big\{q_n({X}_i^\top{\beta}^*,\tau_m^*)-f({X}_i^\top{\beta}^*,\tau_m^*)\Big\}\bigg]^2\\
    &-2\mathbb{E}\bigg[\epsilon_i\sum_{m=1}^{M_n^*}\alpha_m^*\Big\{q_n({X}_i^\top{\beta}^*,\tau_m^*)-f({X}_i^\top{\beta}^*,\tau_m^*)\Big\}\bigg]\bigg|\\
    \leq& C_1^2\delta_n^2\sum_{m=1}^{M_n^*}\sum_{k=1}^{M_n^*}|\alpha_m^*||\alpha_k^*|\times \mathbb{P}\big(|{X}_i^\top{\beta}^*-\tau_m^*|\leq\delta_n,|{X}_i^\top{\beta}^*-\tau_k^*|\leq\delta_n\big)\\
    \leq& C_1^2M_n\delta_n^3\max_{1\le m\le M_{n}^*} |\alpha_m^*|^2F'(w_m)=O(M_n\delta_n^3).
\end{align*}
Similarly, we can show
\begin{align*}
    \mathbb{VAR}\bigg\{\frac{1}{n}\sum_{i=1}^nT_{ni}^2({\theta}_{o}^*)\bigg\}
    &=\frac{1}{n}\mathbb{VAR}\bigg\{\epsilon_i-\sum_{m=1}^{M_n^*}\alpha_m^*\Big(q_n({X}_i^\top{\beta}^*,\tau_m^*)-f({X}_i^\top{\beta}^*,\tau_m^*)\Big)\bigg\}^2\\
    &\leq \frac{8}{n}\mathbb{E}(\epsilon_i^4)+\frac{8}{n}\mathbb{E}\bigg\{\sum_{m=1}^{M_n^*}\alpha_m^*\Big(q_n({X}_i^\top{\beta}^*,\tau_m^*)-f({X}_i^\top{\beta}^*,\tau_m^*)\Big)\bigg\}^4\\
    &\leq \frac{8}{n}\mathbb{E}(\epsilon_i^4)+\frac{8C_1^4M_n\delta_n^5}{n}\max_{1\le m\le M_{n}^*} |\alpha_m^*|^4F'(w_m)=O(1/n)+O(M_n\delta_n^5/n).
\end{align*}
Thus, if $s_n\delta_n^3\to0$ and $s_n\delta_n^5/n\to0$ as $n\to\infty$,
using the Chebyshev’s inequality yields $I_{12}=o_p(1),$
which, together with \eqref{Them3:proof:eq5:1}, implies that \eqref{Them3:proof:eq5} holds.
By \eqref{Them3:proof:eq6}, \eqref{Them3:proof:eq3:B} and \eqref{Them3:proof:eq5}, we obtain \eqref{Them3:proof:eq1}.
This completes the proof of Theorem 3.

\subsection{Proof of Theorem 4}
With a slight abuse of notation, let ${\widehat\theta}=(\widehat\alpha_0,\widehat\beta_2,\ldots,\widehat\beta_{d_1},{\widehat\eta}^\top)^\top$ and ${\theta}^*=(\alpha_0^*,\beta_2^*,\ldots,\beta_{d_1}^*,{\eta}^{*\top})^\top$,
where $\widehat\alpha_0, {\widehat\beta}$ and ${\widehat\eta}$ are obtained using (5).
By the Taylor expansion of $\psi(\tauo,\widehat\alpha_0,{\widehat\beta},{\widehat\eta})$ at ${\theta}^*$, we have
\begin{align*}
    \psi(\tauo,\widehat\alpha_0,{\widehat\beta},{\widehat\eta})
    =&\frac{1}{\sqrt{n}}\sum_{i=1}^nq_n({X}_i^\top{\widehat\beta},\tauo)\big(Y_i-\widetilde{Z}_i^\top{\widehat\eta}-\widehat\alpha_0{X}_i^\top{\widehat\beta}\big)\\
    =&\underbrace{\frac{1}{\sqrt{n}}\sum_{i=1}^nq_n({X}_i^\top{\beta}^*,\tauo)\big(Y_i-\widetilde{Z}_i^\top{\eta}^*-\alpha_0^*{X}_i^\top{\beta}^*\big)}_{I_{13}}+\underbrace{\frac{1}{n}\sum_{i=1}^n{D}_{ni}(\tauo)^\top\sqrt{n}({\widehat\theta}-{\theta}^*)}_{I_{14}}\\
    &+\underbrace{\frac{1}{\sqrt{n}}\sum_{i=1}^n({\widehat\theta}-{\theta}^*)^\top\frac{\partial^2\psi_i(\tauo,{\breve\theta})}{\partial {\theta} \partial {\theta}^\top}({\widehat\theta}-{\theta}^*)}_{I_{15}},
\end{align*}
where ${\breve\theta}$ lies between ${\widehat\theta}$ and ${\theta}^*$,  $\psi_i(\tauo,{\theta})=q_n({X}_i^\top{\beta},\tauo)\big(Y_i-\widetilde{Z}_i^\top{\eta}-\alpha_0{X}_i^\top{\beta}\big)$ and ${D}_{ni}(\tauo)=\big(D_{ni,1}(\tauo),\dots,D_{ni,d}(\tauo)\big)^\top$ with
\begin{align*}
    D_{ni,j}(\tauo)=\begin{cases}
    -{X}_i^\top{\beta}^*q_n({X}_i^\top{\beta}^*,\tauo), & \text{if}~j=1,\\
    -\alpha_0^*X_{i,j}q_n({X}_i^\top{\beta}^*,\tauo)+\frac{\partial q_n({X}_i^\top{\beta}^*,\tauo)}{\partial \beta_{j}}\Big(Y_i-\widetilde{Z}_i^\top{\eta}^*-\alpha_0^*{X}_i^\top{\beta}^*\Big), &\text{if}~ 2\le j \le d_1,\\
    -\widetilde Z_{i,j-d_1-1}q_n({X}_i^\top{\beta}^*,\tauo),
    &\text{if}~ d_1+1\le j\le d.\\
\end{cases}
\end{align*}
Note that
\begin{align*}
    I_{13}=&\frac{1}{\sqrt{n}}\sum_{i=1}^n\big\{q_n({X}_i^\top{\beta}^*,\tauo)-f({X}_i^\top{\beta}^*,\tauo)\big\}\big(Y_i-\widetilde{Z}_i^\top{\eta}^*-\alpha_0^*{X}_i^\top{\beta}^*\big)\\
    &+\frac{1}{\sqrt{n}}\sum_{i=1}^nf({X}_i^\top{\beta}^*,\tauo)\big(Y_i-\widetilde{Z}_i^\top{\eta}^*-\alpha_0^*{X}_i^\top{\beta}^*\big).
\end{align*}
Under $H_{1n}$, using some arguments similar to the proof of \eqref{Them1:proof:eq1:7}, we can show
\begin{align*}
    &\mathbb{E}\Big|\big\{q_n({X}_i^\top{\beta}^*,\tauo)-f({X}_i^\top{\beta}^*,\tauo)\big\}\big(Y_i-\widetilde{Z}_i^\top{\eta}^*-\alpha_0^*{X}_i^\top{\beta}^*\big)\Big|\nonumber\\
    \leq&C_1\delta_n\mathbb{E}\bigg|I(\tauo-\delta_n\leq {X}_i^\top{\beta}^* \leq \tauo+\delta_n) \Big\{\frac{1}{\sqrt{n}}\sum_{m=1}^{M_n^*}\varpi_mf({X}_i^\top{\beta}^*,\tau_m^*)+\epsilon_i\Big\}\bigg|\nonumber\\
    \leq&O(M_n^*\delta_n^2/\sqrt{n})+O(\delta_n^2).
\end{align*}
This, together with $s_n^2/n\to0$ and $\sqrt{n}\delta_n^2\to0$ as $n\to\infty$, implies  
\begin{align*}
    \frac{1}{\sqrt{n}}\sum_{i=1}^n\big\{q_n({X}_i^\top{\beta}^*,\tauo)-f({X}_i^\top{\beta}^*,\tauo)\big\}\big(Y_i-\widetilde{Z}_i^\top{\eta}^*-\alpha_0^*{X}_i^\top{\beta}^*\big)=O_p(\sqrt{n}\delta_n^2)=o_p(1).
\end{align*}
Therefore, we obtain
\begin{align}\label{Them4:proof:eq1}
     I_{13}=\frac{1}{\sqrt{n}}\sum_{i=1}^nf({X}_i^\top{\beta}^*,\tauo)\big(Y_i-\widetilde{Z}_i^\top{\eta}^*-\alpha_0^*{X}_i^\top{\beta}^*\big)+o_p(1).
\end{align}

For $I_{14},$ by the condition $s_n^2/n\to0$ as $n\to\infty$ and the definition of ${\widehat\theta},$
we have
\begin{align}\label{Them4:proof:eq2:2}
    I_{14}=-{D}(\tauo)^\top{\Omega}^{-1}\frac{1}{\sqrt{n}}\sum_{i=1}^n{\xi}_i\big(Y_i-\widetilde{Z}_i^\top{\eta}^*-\alpha_0^*{X}_i^\top{\beta}^*\big)+o_p(1).
\end{align}
Under conditions 1-4, using \eqref{Them1:proof:eq1:1} and \eqref{Them1:proof:eq1:2},
we can show that for any ${\theta}$ satisfying $\|{\theta}-{\theta}^*\|=O(1/\sqrt{n}),$
\begin{align}
   & \mathbb{E}\bigg(\frac{\partial^2 \psi_i(\tauo,{\theta})}{\partial\theta_{j}\partial\theta_{k}}\bigg)=\mathbb{E}\bigg(\frac{\partial^2 \psi_i(\tauo,{\theta}^*)}{\partial\theta_{j}\partial\theta_{k}}\bigg)+o(1),\nonumber\\
\text{and}~~ & \frac{1}{n}\mathbb{VAR}\bigg(\frac{\partial^2 \psi_i(\tauo,{\theta})}{\partial\theta_{j}\partial\theta_{k}}\bigg)\leq O\Big(\frac{1}{n\delta_n}\Big)+O\Big(\frac{s_n^2}{n^2\delta_n}\Big).\nonumber
\end{align}
Therefore, if $s_n^2/n\to0$ and $1/(n\delta_n)\to0$ as $n\to\infty$, using the Chebyshev's inequality, we can obtain
\begin{align}\label{Them4:proof:eq3}
    \|I_{15}\|^2\leq& 2n\|{\widehat\theta}-{\theta}^*\|^4\sum_{j,k}\bigg\{\frac{1}{n}\sum_{i=1}^n\frac{\partial^2 \psi_i(\tauo,{\breve\theta})}{\partial\theta_{j}\partial\theta_{k}}-\mathbb{E}\frac{\partial^2 \psi_1(\tauo,{\theta}^*)}{\partial\theta_{j}\partial\theta_{k}}\bigg\}^2\nonumber\\
    &+2n\|{\widehat\theta}-{\theta}^*\|^4\sum_{j,k}\bigg\{\mathbb{E}\frac{\partial^2 \psi_1(\tauo,{\theta}^*)}{\partial\theta_{j}\partial\theta_{k}}\bigg\}^2\nonumber\\
    =& O_p(1/n)=o_p(1).
\end{align}
Combining \eqref{Them4:proof:eq1}-\eqref{Them4:proof:eq3}, we have
\begin{align*}
    \psi(\tauo,\widehat\alpha_0,{\widehat\beta},{\widehat\eta})
    =&\frac{1}{\sqrt{n}}\sum_{i=1}^n\big\{f({X}_i^\top{\beta}^*,\tauo)-{D}(\tauo)^\top {\Omega}^{-1}{\xi}_i\big\}\big(Y_i-\widetilde{Z}_i^\top{\eta}^*-\alpha_0^*{X}_i^\top{\beta}^*\big)+o_p(1)\\
    =&\frac{1}{\sqrt{n}}\sum_{i=1}^n\psi_{*i}(\tauo)+o_p(1).
\end{align*}
In addition, it can be shown that the class $\{\psi_{*i}(\tauo):\tauo\in\Theta_{\tau}\}$ is P-Donsker.
Therefore, $n^{-1/2}\sum_{i=1}^n\psi_{*i}(\tauo)/\sqrt{\varrho(\tauo)}$ converges weakly to a Gaussian process with
mean function $\Delta(\tauo)=\mathbb{E}[\{f({X}_i^\top{\beta}^*,\tauo)-{D}(\tauo)^\top {\Omega}^{-1}{\xi}\}\sum\nolimits_{m=1}^{M_n^*}\varpi_mf({X}_i^\top{\beta},\tau_m^*)]/\sqrt{\varrho(\tauo)}$
and covariance function ${\Gamma}(\tau_1,\tau_2)=\mathbb{COV}\big(\psi_{*}(\tau_1),\psi_{*}(\tau_2)\big)/\sqrt{\varrho(\tau_1)\varrho(\tau_2)}$
where $\tau_1,\tau_2\in\Theta_{\tau}$.
Finally, we can show that the variance estimator $\widehat\varrho(\tauo)$ converges in probability to $\varrho(\tauo)$ uniformly in $\tauo\in\Theta_{\tau}$
under both the null and the local alternative hypotheses.
Therefore, Theorem 4 holds and the proof is complete.

\section{Proof of Lemmas and (A.8)}\label{secB}
\subsection{Proof of Lemma 1}
\begin{lemma}
 If $s_n^4/(n\delta_n)\to 0$ and $s_n^4\delta_n\to0$,
then we have
\begin{align}\label{Lemma2:1}
    \bigg\|\frac{1}{n}\frac{\partial^2 \bar h_n({\theta}_{o}^*)}{\partial {\theta}_{o} \partial{\theta}_{o}^\top}- {V}({\theta}_{o}^*)\bigg\|_\mathrm{F}=o_p\Big(\frac{1}{\sqrt{s_n}}\Big),
\end{align}
and
\begin{align}\label{Lemma2:2}
    \big\|{V}_n({\theta}_{o}^*)- {V}({\theta}_{o}^*)\big\|_\mathrm{F}=o_p\Big(\frac{1}{\sqrt{s_n}}\Big).
\end{align}
\end{lemma}

\begin{proof}
For any given $\varepsilon>0,$ we have
\begin{align*}
    & \mathbb{P}\bigg(\bigg\|\frac{1}{n}\frac{\partial^2 \bar h_n({\theta}_{o}^*)}{\partial {\theta}_o \partial{\theta}_o^\top}- {V}({\theta}_{o}^*)\bigg\|_\mathrm{F}\geq \frac{\varepsilon}{\sqrt{s_n}}\bigg)\\
    \leq & \frac{s_n}{\varepsilon^2}\sum_{j=1}^{s_n^*}\sum_{k=1}^{s_n^*}\mathbb{E}\bigg\{\frac{1}{n}\frac{\partial^2 \bar h_n({\theta}_{o}^*)}{\partial \theta_{oj} \partial\theta_{ok}}-\frac{1}{n}\mathbb{E}\frac{\partial^2 \bar h_n({\theta}_{o}^*)}{\partial \theta_{oj} \partial\theta_{ok}}+\frac{1}{n}\mathbb{E}\frac{\partial^2 \bar h_n({\theta}_{o}^*)}{\partial \theta_{oj} \partial\theta_{ok}}-V_{j,k}({\theta}_{o}^*)\bigg\}^2\\
    \leq & \underbrace{\frac{4s_n}{\varepsilon^2n}\sum_{j=1}^{s_n^*}\sum_{k=1}^{s_n^*}\mathbb{E}\bigg\{\frac{\partial g_{ni,j}({\theta}_{o}^*)}{\partial \theta_{ok}}-V_{j,k}({\theta}_{o}^*)\bigg\}^2}_{I_{16}}
    +\underbrace{\frac{4s_n}{\varepsilon^2n}\sum_{j=1}^{s_n^*}\sum_{k=1}^{s_n^*}\mathbb{E}\Big\{V_{j,k}({\theta}_{o}^*)\Big\}^2}_{I_{17}}\\
    &+\underbrace{\frac{2s_n}{\varepsilon^2}\sum_{j=1}^{s_n^*}\sum_{k=1}^{s_n^*}\bigg\{\frac{1}{n}\mathbb{E}\frac{\partial^2 \bar h_n({\theta}_{o}^*)}{\partial \theta_{oj} \partial\theta_{ok}}-V_{j,k}({\theta}_{o}^*)\bigg\}^2}_{I_{18}}.
\end{align*}
Next, we show that $I_{16}=O(s_n^3M_n/(n\delta_n)),~I_{17}=O(s_n^3/n)$ and $I_{18}=O(s_n^3M_n^4\delta_n^2)$.
It suffices to show that the terms in the summations of $I_{16},~I_{17}$ and $I_{18}$ have the order of $O(M_n/\delta_n),~O(1)$ and $O(M_n^4\delta_n^2)$, respectively.
For $I_{16}$ and $I_{18}$, we show that the statements hold for the terms with $j, k=2(M_n^*+1)$,
and the statements for other terms can be obtained similarly.
Note that for $j, k=2(M_n^*+1)$, we have
\begin{align}\label{Lemma2:proof:eq1:2}
    &\bigg|\mathbb{E}\bigg\{\frac{\partial g_{ni,j}({\theta}_{o}^*)}{\partial \theta_{ok}}-V_{j,k}({\theta}_{o}^*)\bigg\}\bigg|\nonumber\\
    =&\bigg|\mathbb{E}\bigg\{\Big(\alpha_0^*X_{i,2}+\sum_{m=1}^{M_n^*}\alpha_m^*\frac{\partial q_n({X}_i^\top{\beta}^*,\tau_m^*)}{\partial \beta_{2}} \Big)^2
    -\Big(\alpha_0^*X_{i,2}+\sum_{m=1}^{M_n^*}\alpha_m^*X_{i,2}I(X_i^\top\beta^*>\tau_m^*) \Big)^2\nonumber\\
    &-\sum_{l=1}^{M_n^*}\alpha_l^*\frac{\partial^2 q_n({X}_i^\top{\beta}^*,\tau_l^*)}{\partial \beta_{2}^2} \Big(Y_i-\widetilde{Z}_i^\top{\eta}^*-\alpha_0^*{X}_i^\top{\beta}^*-\sum_{m=1}^{M_n^*}\alpha_m^*q_n({X}_i^\top{\beta}^*,\tau_m^*)\Big)\bigg\}\bigg|\nonumber\\
    \le &\bigg|2\mathbb{E}\bigg\{\sum_{l=1}^{M_n^*}\alpha_0^*\alpha_l^*X_{i,2}\bigg(\frac{\partial  q_n({X}_i^\top{\beta}^*,\tau_l^*)}{\partial \beta_{2}}-X_{i,2}I(X_i^\top\beta^*>\tau_l^*) \bigg)\bigg\}\bigg|\nonumber\\
   &+\bigg|\mathbb{E}\bigg\{\sum_{m=1}^{M_n^*}\sum_{l=1}^{M_n^*}\alpha_l^*\alpha_m^*\Big(\frac{\partial q_n({X}_i^\top{\beta}^*,\tau_m^*)}{\partial \beta_{2}}
    +X_{i,2}I(X_i^\top\beta^*>\tau_m^*) \Big)\Big(\frac{\partial  q_n({X}_i^\top{\beta}^*,\tau_l^*)}{\partial \beta_{2}}-X_{i,2}I(X_i^\top\beta^*>\tau_l^*) \Big)\bigg\}\bigg|\nonumber\\
    &+\bigg|\mathbb{E}\bigg\{\sum_{l=1}^{M_n^*}\alpha_l^*\frac{\partial^2 q_n({X}_i^\top{\beta}^*,\tau_l^*)}{\partial \beta_{2}^2} \Big(\sum_{m=1}^{M_n^*}\alpha_m^*q_n({X}_i^\top{\beta}^*,\tau_m^*)-\sum_{m=1}^{M_n^*}\alpha_m^*f({X}_i^\top{\beta}^*,\tau_m^*)-\epsilon_i\Big)\bigg\}\bigg|.
\end{align}
By \eqref{Them1:proof:eq1:1} and \eqref{Them1:proof:eq1:2},  the first term on the right hand side of \eqref{Lemma2:proof:eq1:2} satisfies
 \begin{align*}
    &\bigg|\mathbb{E}\bigg\{2\alpha_0^*X_{i,2}\sum_{l=1}^{M_n^*}\alpha_l^*\bigg(\frac{\partial q_n({X}_i^\top{\beta}^*,\tau_l^*)}{\partial \beta_{2}} -X_{i,2}I(X_i^\top\beta^*>\tau_l^*)\bigg)\bigg\}\bigg|\\
    \leq&2|\alpha_0^*|\sum_{l=1}^{M_n^*}|\alpha_l^*|\mathbb{E}\big\{X_{i,2}^2I(\tau_l^*-\delta_n\leq X_{i}^\top\beta^*\leq\tau_l^*+\delta_n)/2\big\}\\
    \leq&|\alpha_0^*|\sum_{l=1}^{M_n^*}|\alpha_l^*|\int_{\tau_l^*-\delta_n}^{\tau_l^*+\delta_n}\int_{-\infty}^{\infty}x_{2}^2dF(x_{2}|w)dF(w)\\
    \leq&2|\alpha_0^*|\sum_{l=1}^{M_n^*}|\alpha_l^*|\mathbb{E}(X_{i,2}^2|{X}_i^\top{\beta}^*=w_l)F'(w_l)\delta_n=O(M_n\delta_n),
\end{align*}
where $w_l$ is between $\tau_l^*-\delta_n$ and $\tau_l^*+\delta_n$.
Similarly, we can show that the second and the third terms of \eqref{Lemma2:proof:eq1:2} are $O(M_n^2\delta_n)$.
This implies
\[
\bigg|\mathbb{E}\bigg\{\frac{\partial g_{ni,j}({\theta}_{o}^*)}{\partial \theta_{ok}}-V_{j,k}({\theta}_{o}^*)\bigg\}\bigg|=O(M_n^2\delta_n).
\]
In addition, by some the arguments similar to the proof of \eqref{Lemma2:proof:eq1:2} , we can show
\begin{align}\label{Lemma2:proof:eq1:1}
    \mathbb{E}\bigg\{\frac{\partial g_{ni,j}({\theta}_{o}^*)}{\partial \theta_{ok}}-V_{j,k}({\theta}_{o}^*)\bigg\}^2\leq O(M_n/\delta_n).
\end{align}
These facts imply $I_{16}=O(s_n^3M_n/(n\delta_n))$ and $I_{18}=O(s_n^3M_n^4\delta_n^2)$.

For $I_{17}$, we show that $|V_{j,k}({\theta}_{o}^*)|$ ($j=2M_n^*+2$, $k=2M_n^*+d_1+2$) is bounded away from infinite.
Note that
\begin{align}\label{Lemma2:proof:eq2:1}
    &\alpha_0^*+\sum_{m=1}^{M_n^*}\alpha_m^*I({X}_i^\top{\beta}^*>\tau_m^*)
    =\sum_{m=0}^{M_n^*}\mu_m^*I(\tau_m^*< {X}_i^\top{\beta}^*\leq \tau_{m+1}^*),\nonumber\\
    \text{and}~~&\sum_{m=0}^{M_n^*}I(\tau_m^*< {X}_i^\top{\beta}^*\leq \tau_{m+1}^*)=1,
\end{align}
where $\mu_m^*=\sum_{k=0}^m\alpha_k^*$ and $\mu_m^*~(0\le m\le M_n^*)$ are bounded under condition 2.
By \eqref{Lemma2:proof:eq2:1}, condition 1 and the Cauchy–Schwartz inequality, we can show that for $j=2M_n^*+2$ and $k=2M_n^*+d_1+2$,
\begin{align*}
    \big|V_{j,k}({\theta}_{o}^*)\big|=&\bigg|\mathbb{E}\Big[Z_{i,1}X_{i,2}\Big\{\alpha_{0}^*+\sum_{m=1}^{M_n^*}\alpha_{m}^*I({X}_i^\top{\beta}^*>\tau_{m}^*)\Big\}\Big]\bigg|\\
    \leq&\bigg|\mathbb{E}\Big\{Z_{i,1}X_{i,2}\sum_{m=0}^{M_n^*}\mu_m^*I(\tau_m^*< {X}_i^\top{\beta}^*\leq \tau_{m+1}^*)\Big\}\bigg|\\
    \leq&\max_{0\leq m \leq M_n^*}|\mu_m^*|\mathbb{E}\Big\{|Z_{i,1}X_{i,2}|\sum_{m=0}^{M_n^*}I(\tau_m^*< {X}_i^\top{\beta}^*\leq \tau_{m+1}^*)\Big\}\\
    \leq&\max_{0\leq m \leq M_n^*}|\mu_m^*|\big(\mathbb{E}Z_{i,1}^2\big)^{1/2}\big(\mathbb{E}X_{i,2}^2\big)^{1/2},
\end{align*}
which is bounded above by conditions 2 and 3.
Similarly, we can show that $|V_{j,k}({\theta}_{o}^*)|$ ($j\neq 2M_n^*+2$, $k\neq 2M_n^*+d_1+2$) is also bounded away from infinite.
These facts imply  $I_{17}=O(s_n^3/n)$. If $s_n^4/(n\delta_n)\to0$ and $s_n^4\delta_n\to0$ as $n\to\infty$,  then
\begin{align*}
     \mathbb{P}\bigg(\bigg\|\frac{1}{n}\frac{\partial^2 \bar h_n({\theta}_{o}^*)}{\partial {\theta}_o\partial{\theta}_o^\top}-{V}({\theta}_{o}^*)\bigg\|_\mathrm{F}\geq \frac{\varepsilon}{\sqrt{s_n}}\bigg)\leq O\Big(\frac{s_n^3M_n}{n\delta_n}\Big)+O\Big(\frac{s_n^3}{n}\Big)+O(s_n^3M_n^4\delta_n^2)\to0.
\end{align*}
Hence, we obtain that \eqref{Lemma2:1} holds. Similarly, we can prove \eqref{Lemma2:2}.
This completes the proof of Lemma 1.
\end{proof}

\subsection{Proof of (A.8)}
\begin{proof}
The proof of \eqref{Them1:proof:eq1:7} is based on \eqref{Them1:proof:eq1:1} and \eqref{Them1:proof:eq1:2}.
For $M_n^*+1\leq j \leq 2M_n^*$, we have
\begin{align}\label{eq8:term1}
    &\mathbb{E}\big\{\big|g_{ni,j}({\theta}_{o}^*)-g_{i,j}({\theta}_{o}^*)\big|\big\}\nonumber\\
    \leq&|\alpha_{j-M_n^*}^*|\mathbb{E}\bigg\{\bigg|\Big(Y_i-\widetilde{Z}_i^\top{\eta}^*-\alpha_0^*{X}_i^\top{\beta}^*\Big)\bigg(I({X}_i^\top{\beta}^*>\tau_{j-M_n^*}^*)+\frac{\partial q_n({X}_i^\top{\beta}^*,\tau_{j-M_n^*}^*)}{\partial \tau_{j-M_n^*}}\bigg)\bigg|\bigg\}\nonumber\\
    &+|\alpha_{j-M_n^*}^*|\max_{1\le m\le M_n^*}|\alpha_m^*|\sum_{m=1}^{M_n^*}\mathbb{E}\bigg\{\Big|I({X}_i^\top{\beta}^*>\tau_{j-M_n^*})f({X}_i^\top{\beta}^*,\tau_{m})\nonumber\\
    &\hspace{2.5in}+\frac{\partial q_n({X}_i^\top{\beta}^*,\tau_{j-M_n^*})}{\partial \tau_{j-M_n^*}}q_n({X}_i^\top{\beta}^*,\tau_{m})\Big|\bigg\}.
\end{align}
By \eqref{Them1:proof:eq1:1}, \eqref{Them1:proof:eq1:2} and conditions 1-5, the first term on the right hand of \eqref{eq8:term1} satisfies
\begin{align*}
    &|\alpha_{j-M_n^*}^*|\mathbb{E}\bigg\{\bigg|\Big(Y_i-\widetilde{Z}_i^\top{\eta}^*-\alpha_0^*{X}_i^\top{\beta}^*\Big)\bigg(I({X}_i^\top{\beta}^*>\tau_{j-M_n^*}^*)+\frac{\partial q_n({X}_i^\top{\beta}^*,\tau_{j-M_n^*}^*)}{\partial \tau_{j-M_n^*}}\bigg)\bigg|\bigg\}\\
    \leq&\frac{1}{2}|\alpha_{j-M_n^*}^*|\max_{1\le m\le M_n^*} |\alpha_m^*| \sum_{m=1}^{j-M_n^*-1}|\tau_{j-M_n^*}^*-\tau_m^*|\int_{\tau_{j-M_n^*}^*-\delta_n}^{\tau_{j-M_n^*}^*+\delta_n}dF(w)\\
    &+\frac{1}{2}|\alpha_{j-M_n^*}^*||\alpha_j^*|\int_{\tau_{j-M_n^*}^*-\delta_n}^{\tau_{j-M_n^*}^*+\delta_n}|w-\tau_{j-M_n^*}^*|dF(w)\\
    &+\frac{1}{2}|\alpha_{j-M_n^*}^*|\times\int_{\tau_{j-M_n^*}^*-\delta_n}^{\tau_{j-M_n^*}^*+\delta_n}\mathbb{E}(|\epsilon_i|\big|W_i=w)dF(w)\\
    =&O(M_n\delta_n)+O(\delta_n).
\end{align*}
Similarly, the second term on the right hand side of \eqref{eq8:term1} satisfies
\begin{align*}
    &|\alpha_{j-M_n^*}^*|\max_{1\le m\le M_n^*} |\alpha_m^*|\sum_{m=1}^{M_n^*}\mathbb{E}\bigg\{\Big|I({X}_i^\top{\beta}^*>\tau_{j-M_n^*})f({X}_i^\top{\beta}^*,\tau_{m})+\frac{\partial q_n({X}_i^\top{\beta}^*,\tau_{j-M_n^*})}{\partial \tau_{j-M_n^*}}q_n({X}_i^\top{\beta}^*,\tau_{m})\Big|\bigg\}\\
    \leq&\frac{1}{2}|\alpha_{j-M_n^*}^*|\max_{1\le m\le M_n^*} |\alpha_m^*|\sum_{m=1}^{j-M_n^*-1}|\tau_{j-M_n^*}^*-\tau_m^*|\int_{\tau_{j-M_n^*}^*-\delta_n}^{\tau_{j-M_n^*}^*+\delta_n}dF(w)\\
    &+\frac{C_1\delta_n}{2}|\alpha_{j-M_n^*}^*|\max_{1\le m\le M_n^*} |\alpha_m^*|\int_{\tau_{j-M_n^*}^*-\delta_n}^{\tau_{j-M_n^*}^*+\delta_n}dF(w)\\
    &+\frac{C_1\delta_n}{2}|\alpha_{j-M_n^*}^*|\max_{1\le m\le M_n^*} |\alpha_m^*| \sum_{m=j-M_n^*+1}^{M_n^*}\int_{\tau_m^*-\delta_n}^{\tau_m^*+\delta_n}dF(w)\\
    =&O(M_n\delta_n)+O(\delta_n^2)+O(M_n\delta_n^2).
\end{align*}
For $j=2M_n^*+1$,  by \eqref{Them1:proof:eq1:1}, \eqref{Them1:proof:eq1:2} and conditions 2 and 3, we have
\begin{align*}
    \mathbb{E}\big\{\big|g_{ni,j}({\theta}_{o}^*)-g_{i,j}({\theta}_{o}^*)\big|\big\}
    &=\mathbb{E}\bigg\{\Big|{X}_i^\top{\beta}^*\sum_{m=1}^{M_n^*}\alpha_m^*\Big(q_n({X}_i^\top{\beta}^*,\tau_m^*)-f({X}_i^\top{\beta}^*,\tau_m^*)\Big)\Big|\bigg\}\\
    &\leq C_1\delta_n\max_{1\le m\le M_n^*} |\alpha_m^*|\sum_{m=1}^{M_n^*}\int_{\tau_m^*-\delta_n}^{\tau_m^*+\delta_n}|w|dF(w)\\
    &=O(M_n\delta_n^{2}).
\end{align*}
For $2M_n^*+2\leq j \leq 2M_n^*+d_1$, a straightforward calculation yields
\begin{align*}
    &\mathbb{E}\big\{\big|g_{ni,j}({\theta}_{o}^*)-g_{i,j}({\theta}_{o}^*)\big|\big\}\\
    \leq & \underbrace{\sum_{m=1}^{M_n^*}\mathbb{E}\Big\{\big|\alpha_0^*X_{i,j-2M_n^*}\big(q_n({X}_i^\top{\beta}^*,\tau_m^*)-f({X}_i^\top{\beta}^*,\tau_m^*)\big)\big|\Big\}}_{I_{19}}\\
&+\underbrace{\sum_{m=1}^{M_n^*}\sum_{k=1}^{M_n^*}|\alpha_m^*\alpha_k^*|\mathbb{E}\bigg\{\Big|q_n({X}_i^\top{\beta}^*,\tau_k^*)\frac{\partial q_n({X}_i^\top{\beta}^*,\tau_m^*)}{\partial \beta_{j-2M_n^*}}-f({X}_i^\top{\beta}^*,\tau_k^*)X_{i,j-2M_n^*}I(X_i^\top\beta^*>\tau_m^*)\Big|\bigg\}}_{I_{20}}\\
&+\underbrace{\sum_{m=1}^{M_n^*}|\alpha_m^*|\mathbb{E}\bigg\{\bigg|\Big(Y_i-\widetilde{Z}_i^\top{\eta}^*-\alpha_0^*{X}_i^\top{\beta}^*\Big)\bigg(\frac{\partial q_n({X}_i^\top{\beta}^*,\tau_m^*)}{\partial \beta_{j-2M_n^*}}-X_{i,j-2M_n^*}I(X_i^\top\beta^*>\tau_m^*)\bigg)\bigg|\bigg\}}_{I_{21}}.
\end{align*}
Next, we show that $I_{19}=O(M_n\delta_n^2)$,
$I_{20}=O(M_n^2\delta_n)$, and $I_{21}=O(M_n^2\delta_n)$.
For $I_{19}$, by \eqref{Them1:proof:eq1:1}, \eqref{Them1:proof:eq1:2} and conditions 1-3, we have
\begin{align*}
I_{19}\leq C_1\delta_n|\alpha_0^*|\sum_{m=1}^{M_n^*}\int_{\tau_m^*-\delta_n}^{\tau_m^*+\delta_n}\mathbb{E}(|X_{i,j-2M_n^*}|\big|W_i=w)dF(w)=O(M_n\delta_n^2).
\end{align*}
Similarly, we can show
\begin{align*}
I_{20}\leq&\frac{1}{2}\sum_{m=1}^{M_n^*}\sum_{k=1}^{m-1}|\alpha_m^*\alpha_k^*|\times|\tau_m^*-\tau_k^*|\int_{\tau_m^*-\delta_n}^{\tau_m^*+\delta_n}\mathbb{E}(|X_{i,j-2M_n^*}|\big|W_i=w)dF(w)\\
    &+\frac{C_1\delta_n}{2}\sum_{m=1}^{M_n^*}|\alpha_m^*|^2\int_{\tau_m^*-\delta_n}^{\tau_m^*+\delta_n}\mathbb{E}(|X_{i,j-2M_n^*}|\big|W_i=w)dF(w)\\
    &+C_1\delta_n\sum_{m=1}^{M_n^*}\sum_{k=m+1}^{M_n^*}|\alpha_m^*\alpha_k^*|\int_{\tau_k^*-\delta_n}^{\tau_k^*+\delta_n}\mathbb{E}(|X_{i,j-2M_n^*}|\big|W_i=w)dF(w)\\
    =&O(M_n^2\delta_n)+O(M_n\delta_n^2)+O(M_n^2\delta_n),
\end{align*}
and
\begin{align*}
   I_{21}\leq&\frac{1}{2}\sum_{m=1}^{M_n^*}\sum_{k=1}^{m-1}|\alpha_m^*\alpha_k^*|\times|\tau_m^*-\tau_k^*|\int_{\tau_m^*-\delta_n}^{\tau_m^*+\delta_n}\mathbb{E}(|X_{i,j-2M_n^*}|\big|W_i=w)dF(w)\\
    &+\frac{1}{2}\sum_{m=1}^{M_n^*}|\alpha_m^*|^2\int_{\tau_m^*-\delta_n}^{\tau_m^*+\delta_n}\mathbb{E}(|X_{i,j-2M_n^*}|\big|W_i=w)|w-\tau_m^*|dF(w)\\
    &+\frac{1}{2}\sum_{m=1}^{M_n^*}|\alpha_m^*|\int_{\tau_m^*-\delta_n}^{\tau_m^*+\delta_n}\int_{-\infty}^{\infty}|x_{j-2M_n^*}|\mathbb{E}(|\epsilon_i|\big|X_{i,j-2M_n^*}=x_{j-2M_n^*},W_i=w)dF(x_{j-2M_n^*},w)\\
    =&O(M_n^2\delta_n)+O(M_n\delta_n)+O(\delta_n).
\end{align*}
Combining the results of $I_{19}, I_{20}$ and $I_{21}$, we have that for $2M_n^*+2\leq j \leq 2M_n^*+d_1$,
\[
\mathbb{E}\big\{\big|g_{ni,j}({\theta}_{o}^*)-g_{i,j}({\theta}_{o}^*)\big|\big\}\le O(M_n^2\delta_n).
\]
For $2M_n^*+d_1+1\leq j\leq 2M_n^*+d$, by \eqref{Them1:proof:eq1:1}, \eqref{Them1:proof:eq1:2} and conditions 2 and 3, we obtain
\begin{align*}
    &\mathbb{E}\big\{\big|g_{ni,j}({\theta}_{o}^*)-g_{i,j}({\theta}_{o}^*)\big|\big\}\\
    =&\mathbb{E}\bigg\{\Big|\widetilde Z_{i,j-2M_n^*-d_1-1}\sum_{m=1}^{M_n^*}\alpha_m^*\big(q_n({X}_i^\top{\beta}^*,\tau_m^*)-f({X}_i^\top{\beta}^*,\tau_m^*)\big)\Big|\bigg\}\\
    \leq & \max_{1\le m\le M_n^*} |\alpha_m^*|\big(\mathbb{E}\widetilde Z_{i,j-2M_n^*-d_1-1}^2\big)^{1/2}\sum_{m=1}^{M_n^*} \bigg\{\int_{\tau_m^*-\delta_n}^{\tau_m^*+\delta_n}\Big(q_n(w,\tau_m^*)-f(w,\tau_m^*)\Big)^2dF(w)\bigg\}^{1/2}\\
    = & O(M_n\delta_n^{3/2}).
\end{align*}
These facts imply that \eqref{Them1:proof:eq1:7} also holds for $2M_n^*+1 \leq j \leq 2M_n^*+d$.
\end{proof}

\section{An iterative algorithm for solving problem (4)}\label{secC}
In this section, we develop an iterative procedure to obtain ${\widehat\theta}.$ For simplicity, let ${\alpha}={\alpha}(M_n)$,
${\tau}={\tau}(M_n)$ and ${\theta}={\theta}(M_n).$
Define
\begin{align*}
    h_n({\eta},{\beta},{\alpha},{\tau})=\frac{1}{2}\sum_{i=1}^n\Big\{Y_i-\widetilde{Z}_i^\top{\eta}-\alpha_0{X}_i^\top{\beta}-
\sum_{m=1}^{M_n}\alpha_mq_n({X}_i^\top{\beta},\tau_m)\Big\}^2.
\end{align*}
Let ${\theta}^{[k]}=\big(({\alpha}_{(-0)}^{[k]})^\top,({\tau}^{[k]})^\top,\alpha_0^{[k]},({\beta}_{(-1)}^{[k]})^\top,({\eta}^{[k]})^\top\big)^\top$ be the estimate of ${\theta}$ at the $k$th iteration.
Then, the proposed iterative procedure is summarized in Algorithm \ref{ALG1}.

\begin{algorithm}
\caption{: An iterative procedure to minimize (4)} \label{ALG1}
\begin{algorithmic}
\STATE{Input: ${\beta}^{[0]}$, ${\tau}^{[0]}$, $k=0$, an integer $K$ and a tolerance parameter $\varepsilon_1.$}

\STATE{\textit{Step 1}. Update ${\eta}^{[k+1]}$ and ${\alpha}^{[k+1]}$ by
\begin{align}\label{upalpha}
({\alpha}^{[k+1]},{\eta}^{[k+1]})=\text{arg}\min_{({\alpha}^\top,{\eta}^\top)^\top}~h_n({\eta},{\beta}^{[k]},{\alpha},{\tau}^{[k]})
+n\sum_{m=1}^{M_n} p_{\lambda_n,t}(|\alpha_m|).
\end{align}
}
\STATE{\textit{Step 2}. Update $\tau^{[k+1]}$ by
\begin{align*}
    \tau^{[k+1]}_{m}&=\widetilde\tau^{[k+1]}_{m}I(\alpha^{[k+1]}_{m}\neq0)+\tau_\infty I(\alpha^{[k+1]}_m=0)~~(m=1,\dots,M_n),
\end{align*}
where
\begin{align}\label{uptau}
{\widetilde\tau}^{[k+1]}&=\text{arg}\min_{{\tau}}~h_n({\eta}^{[k+1]},{\beta}^{[k]},{\alpha}^{[k+1]},{\tau}).
\end{align}
\textit{Step 3}. Update ${\beta}^{[k+1]}$ by
\begin{align}\label{upbeta}
{\beta}^{[k+1]}=\text{arg}\min_{{\beta}}~h_n({\eta}^{[k+1]},{\beta},{\alpha}^{[k+1]},{\tau}^{[k+1]}).
\end{align}
}
\STATE{\textit{Step 4}. Let ${\theta}^{[k+1]}=\big(({\alpha}_{(-0)}^{[k+1]})^\top,({\tau}^{[k+1]})^\top,\alpha_0^{[k+1]},({\beta}_{(-1)}^{[k+1]})^\top,({\eta}^{[k+1]})^\top\big)^\top$ and $k=k+1.$
Repeat \textit{Steps 1-3} until $\|{\theta}^{[k]}-{\theta}^{[k-1]}\|<\varepsilon_1$ or $k>K$.}

\STATE{Output: ${\widehat\theta}={\theta}^{[k]}$.}
\end{algorithmic}
\end{algorithm}

The problems  \eqref{uptau} and \eqref{upbeta} in Algorithm \ref{ALG1} can be efficiently solved using the Matlab function {\it fmincon}.
Next, we consider to minimize \eqref{upalpha}. 
Note that minimizing problem \eqref{upalpha} is equivalent to minimize the following the constraint optimization problem:
\begin{align*}
\nonumber
&h_n({\eta},{\beta}^{[k]},{\alpha},{\tau}^{[k]})
+n\sum_{m=1}^{M_n} p_{\lambda_n,t}(|\zeta_m|)\\
\text{subject to}~~~& \alpha_m-\zeta_m=0,\quad m=1,\dots, M_n,
\end{align*}
where ${\zeta}=(\zeta_1,\ldots,\zeta_{M_n})^\top$.
By the augmented Lagrangian method, we can estimate ${\eta}$, ${\alpha}$ and ${\zeta}$ by minimizing
\begin{align}\label{L}
\mathcal{L}({\eta},{\alpha},{\zeta},{v})=&h_n({\eta},{\beta}^{[k]},{\alpha},{\tau}^{[k]})+n\sum_{m=1}^{M_n}p_{\lambda_n,t}(|\zeta_m|)\nonumber\\
&+n\sum_{m=1}^{M_n}v_m(\alpha_m-\zeta_m)+\frac{n\vartheta}{2}\sum_{m=1}^{M_n}(\alpha_m-\zeta_m)^2,
\end{align}
where ${v}=(v_1,\dots,v_{M_n})^\top$ with $v_m~(m=1,\dots,M_n)$ being the Lagrange multipliers and $\vartheta$ is a penalty parameter.

Let $\mathbb{Y}=(Y_1,\ldots,Y_n)^\top,$ $\mathbb{Z}=(\widetilde{Z}_1,\dots,\widetilde{Z}_n)^\top$
and $\mathbb{X}=({X}_1,\ldots,{X}_n)^\top.$
Define $\mathbb{Q}({\beta},{\tau})=\big(\mathbb{X}{\beta},q_n(\mathbb{X}{\beta},\tau_1),\ldots,q_n(\mathbb{X}{\beta},\tau_{M_n})\big).$
To solve problem \eqref{L}, we utilize the alternating direction method of multipliers (ADMM), which iteratively updates the estimates of ${\alpha}$, ${\eta}$, ${\zeta}$, and ${v}$.  Define ${\alpha}^{[k+1,l]}, {\eta}^{[k+1,l]}, {\zeta}^{[k+1,l]}$ and ${v}^{[k+1,l]}$ as the estimates of ${\alpha}, {\eta}, {\zeta}$ and ${v}$
at the $l$th iteration.
The algorithm proceeds as follows.

First, for each given ${\eta}={\eta}^{[k+1,l]}, {\zeta}={\zeta}^{[k+1,l]}$ and ${v}={v}^{[k+1,l]},$ we can write
\begin{align*}
&\mathcal{L}({\eta}^{[k+1,l]},{\alpha},{\zeta}^{[k+1,l]},{v}^{[k+1,l]})\\
=&\frac{1}{2}\Big\|\mathbb{Y}-\mathbb{Z}{\eta}^{[k+1,l]}-\mathbb{Q}^{[k]}{\alpha}\Big\|^2+\frac{n\vartheta}{2}\Big\|{\alpha}_{(-0)}-{\zeta}^{[k+1,l]}+\vartheta^{-1}{v}^{[k+1,l]}\Big\|^2+C,
\end{align*}
where $\mathbb{Q}^{[k]}=\mathbb{Q}({\beta}^{[k]},{\tau}^{[k]})$ and $C$ is a constant independent of ${\alpha}.$
Thus, by solving the equation
\[
\partial\mathcal{L}({\eta}^{[k+1,l]},{\alpha},{\zeta}^{[k+1,l]},{v}^{[k+1,l]})/\partial {\alpha}=0,
\]
we obtain an updating form of ${\alpha}$ as follows:
\begin{align}\label{Step2:upalpha}
    {\alpha}^{[k+1,l+1]}=&\Big\{(\mathbb{Q}^{[k]})^\top \mathbb{Q}^{[k]}+n\vartheta\mathbb{I}_{M_n+1}\Big\}^{-1}\times\nonumber\\
    &\quad\Big\{(\mathbb{Q}^{[k]})^\top(\mathbb{Y}-\mathbb{Z}{\eta}^{[k+1,l]})+n(\vartheta{\bar\zeta}^{[k+1,l]}-{\bar v}^{[k+1,l]})\Big\},
\end{align}
where $\mathbb{I}_{M_n+1}=\text{diag}\{0,1,\dots,1\}$ is an $(M_n+1)\times (M_n+1)$ diagonal matrix,
${\bar\zeta}^{[k+1,l]}=(0,({\zeta}^{[k+1,l]})^\top)^\top$ and ${\bar v}^{[k+1,l]}=(0,({v}^{[k+1,l]})^\top)^\top$.

In the second step, by some arguments similar to \eqref{Step2:upalpha},
for each given ${\alpha}={\alpha}^{[k+1,l+1]},$ ${\zeta}={\zeta}^{[k+1,l]}$ and ${v}={v}^{[k+1,l]}$ in $\mathcal{L}({\eta},{\alpha},{\zeta},{v}),$
we can update ${\eta}$ by
\begin{align}\label{upgamma}
{\eta}^{[k+1,l+1]}=(\mathbb{Z}^\top \mathbb{Z})^{-1}\mathbb{Z}^\top(\mathbb{Y}-\mathbb{Q}^{[k]}{\alpha}^{[k+1,l+1]}).
\end{align}
Third, we update ${\zeta}.$
By \eqref{L}, we observe that for given ${\eta}={\eta}^{[k+1,l+1]}, {\alpha}={\alpha}^{[k+1,l+1]}$ and ${v}={v}^{[k+1,l]}$,
it suffices to minimize
\begin{align}\label{eq:upeta}
&n\sum_{m=1}^{M_n} v_m^{[k+1,l]}(\alpha_m^{[k+1,l+1]}-\zeta_m)+\frac{n\vartheta}{2}\sum_{m=1}^{M_n}(\alpha_m^{[k+1,l+1]}-\zeta_m)^2+n\sum_{m=1}^{M_n}p_{\lambda_n,t}(|\zeta_m|)\nonumber\\
=& \sum_{m=1}^{M_n}\frac{n\vartheta}{2}(u_m^{[k+1,l+1]}-\zeta_m)^2+n\sum_{m=1}^{M_n}p_{\lambda_n,t}(|\zeta_m|),
\end{align}
where $u_m^{[k+1,l+1]}=\alpha_m^{[k+1,l+1]}+\vartheta^{-1}v_m^{[k+1,l]}$.
For the MCP penalty with  $t>1/\vartheta$, by minimizing \eqref{eq:upeta}, we have
\begin{align}\label{upeta:MCP}
\zeta_m^{[k+1,l+1]}=
\begin{cases}
\frac{\text{ST}(u_m^{[k+1,l+1]},\lambda_n/\vartheta)}{1-1/(t\vartheta)} & \text{if}~~|u_m^{[k+1,l+1]}|\leq t\lambda_n, \\
u_m^{[k+1,l+1]} & \text{if}~~|u_m^{[k+1,l+1]}|>t\lambda_n,
\end{cases}
\end{align}
where $\text{ST}(x,\lambda_n)=\text{sign}(x)(|x|-\lambda_n)_+$ is the soft thresholding rule,
and $(x)_+ = \max\{x,0\}.$
For the SCAD penalty with $t>1/\vartheta+1$, it is
\begin{align}\label{upeta:SCAD}
\zeta_m^{[k+1,l+1]}=
\begin{cases}
\text{ST}(u_m^{[k+1,l+1]},\lambda_n/\vartheta) & \text{if}~~|u_m^{[k+1,l+1]}|\leq \lambda_n(1+1/\vartheta),\\
\frac{\text{ST}(u_m^{[k+1,l+1]},t\lambda_n/((t-1)\vartheta))}{1-1/((t-1)\vartheta)} & \text{if}~~\lambda_n(1+1/\vartheta)<|u_m^{[k+1,l+1]}|\leq t\lambda_n,\\
u_m^{[k+1,l+1]} & \text{if}~~|u_m^{[k+1,l+1]}|>t\lambda_n.
\end{cases}
\end{align}
Finally, we update $v$ by
\begin{align}\label{upv}
    {v}^{[k+1,l+1]}= {v}^{[k+1,l]}+\vartheta\big({\alpha}_{(-0)}^{[k+1,l+1]}-{\zeta}^{[k+1,l+1]}\big).
\end{align}

The ADMM algorithm to solve \eqref{upalpha} are summarized in Algorithm \ref{ADMM}.

\begin{algorithm}
\caption{: Alternating Mirection Method of Dultipliers Algorithm (ADMM algorithm)}\label{ADMM}
\begin{algorithmic}
\STATE {{Input}: ${\alpha}^{[k+1,0]}={\alpha}^{[k]},$ ${\eta}^{[k+1,0]}={\eta}^{[k]},$  $l=0$, a positive integer $L$ and a tolerance parameter $\varepsilon_2$.}
\STATE{{Step 2.1}. Calculate ${\alpha}^{[k+1,l+1]}$ by \eqref{Step2:upalpha};}
\STATE{{Step 2.2}. Calculate ${\eta}^{[k+1,l+1]}$ by \eqref{upgamma};}
\STATE{{Step 2.3}. Calculate ${\zeta}^{[k+1,l+1]}$ by \eqref{upeta:MCP} and \eqref{upeta:SCAD} for MCP and SCAD, respectively;}
\STATE{{Step 2.4}. Calculate ${v}^{[k+1,l+1]}$ by \eqref{upv};}
\STATE{{Step 2.5}. Let $l\leftarrow l+1.$ Repeat Steps 2.1-2.4 until $\|(({\alpha}^{[k+1,l]})^\top,({\eta}^{[k+1,l]})^\top)^\top-(({\alpha}^{[k+1,l-1]})^\top,({\eta}^{[k+1,l-1]})^\top)^{\top}\|<\varepsilon_2$ or $l>L$.}
\STATE{{Output}: ${\alpha}^{[k+1]}={\alpha}^{[k+1,l]}$ and ${\eta}^{[k+1]}={\eta}^{[k+1,l]}$}.
\end{algorithmic}
\end{algorithm}



\section{Additional simulation results}\label{secD}


\begin{figure}
\centering
\includegraphics[width=1\textwidth]{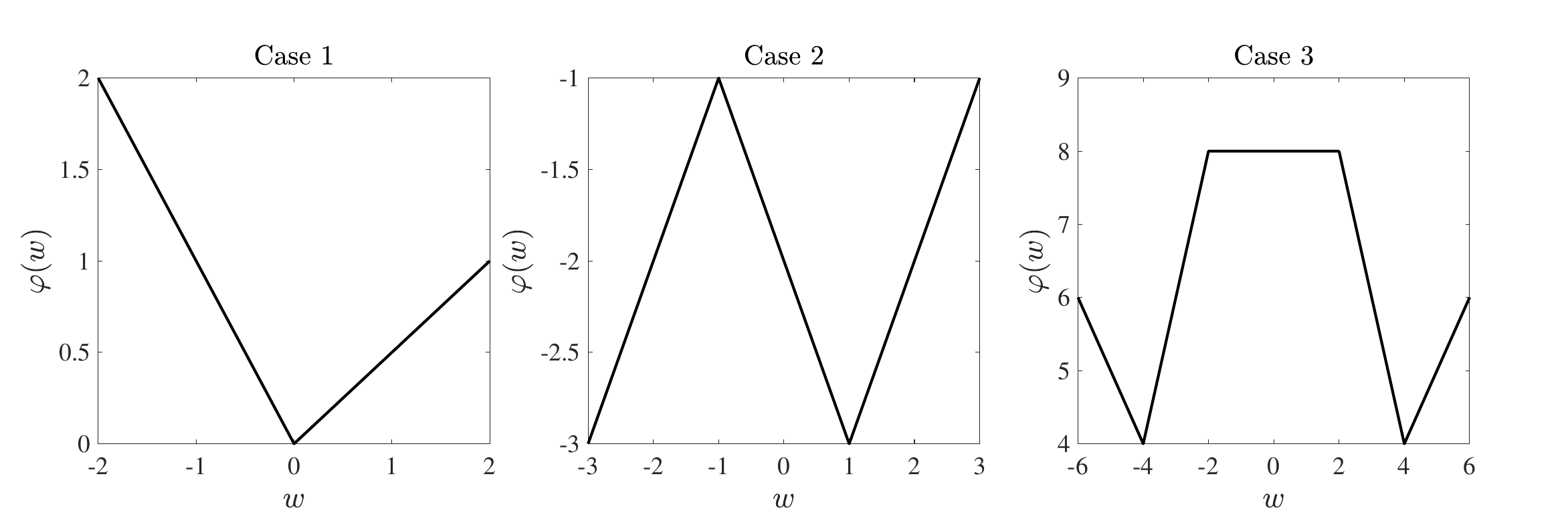}
\caption*{Figure S1: The curves of $\varphi(w)=\alpha_0w+\sum_{m=1}^{M_n^*}\alpha_m^*f(w,\tau_m^*)$ under different settings}\label{Simulation:Example:Figure}
\end{figure}


\begin{table}		
{\footnotesize
\caption*{Table S1: Simulation results (multiplied by 100) for Case 1 with $\alpha^*=(1,-1.5)^\top$, $C_n=\log\{\log (n)\}$
and $\epsilon_i\sim\text{Schi}^2(2)$.}												
\begin{tabular}{cccccccccccccccccc}		
\hline
& & &\multicolumn{4}{c}{Oracle}& &\multicolumn{4}{c}{SCAD}& &\multicolumn{4}{c}{MCP}\\				
\cline{4-7} \cline{9-12} \cline{14-17}
$n$ & $\nu$	&		&	Bias	&	SE	&	SD	&	CP	&		&	Bias	&	SE	&	SD	&	CP	&		&	Bias	&	SE	&	SD	&	CP	\\		
\hline
1000	&	0.6	&	$\alpha_0$	&	-0.027 	&	4.73 	&	4.68 	&	95.3 	&		&	-0.010 	&	4.73 	&	4.64 	&	95.6 	&		&	-0.027 	&	4.73 	&	4.62 	&	95.8 	\\
	&		&	$\alpha_1$	&	0.466 	&	7.82 	&	8.01 	&	93.8 	&		&	0.424 	&	7.81 	&	8.02 	&	93.8 	&		&	0.462 	&	7.82 	&	7.94 	&	94.0 	\\
	&		&	$\tau_1$	&	0.364 	&	5.57 	&	5.78 	&	95.2 	&		&	0.336 	&	5.57 	&	5.72 	&	95.4 	&		&	0.357 	&	5.57 	&	5.70 	&	95.3 	\\
	&		&	$\beta_2$	&	-0.244 	&	4.09 	&	4.10 	&	94.3 	&		&	-0.249 	&	4.09 	&	4.08 	&	94.1 	&		&	-0.237 	&	4.09 	&	4.08 	&	94.2 	\\
	&		&	$\gamma_1$	&	0.065 	&	3.34 	&	3.38 	&	94.5 	&		&	0.071 	&	3.34 	&	3.40 	&	94.3 	&		&	0.080 	&	3.34 	&	3.40 	&	94.4 	\\
	&	0.8	&	$\alpha_0$	&	-0.078 	&	4.73 	&	4.68 	&	95.4 	&		&	-0.032 	&	4.73 	&	4.64 	&	95.5 	&		&	-0.039 	&	4.73 	&	4.62 	&	95.8 	\\
	&		&	$\alpha_1$	&	0.489 	&	7.81 	&	7.96 	&	94.1 	&		&	0.407 	&	7.80 	&	8.01 	&	93.8 	&		&	0.462 	&	7.81 	&	7.95 	&	94.1 	\\
	&		&	$\tau_1$	&	0.292 	&	5.55 	&	5.76 	&	95.5 	&		&	0.275 	&	5.55 	&	5.71 	&	95.6 	&		&	0.329 	&	5.55 	&	5.72 	&	95.5 	\\
	&		&	$\beta_2$	&	-0.190 	&	4.08 	&	4.08 	&	94.7 	&		&	-0.229 	&	4.09 	&	4.07 	&	94.3 	&		&	-0.222 	&	4.08 	&	4.08 	&	94.3 	\\
	&		&	$\gamma_1$	&	0.075 	&	3.34 	&	3.38 	&	94.6 	&		&	0.076 	&	3.34 	&	3.39 	&	94.3 	&		&	0.083 	&	3.34 	&	3.40 	&	94.4 	\\
\hline																												
2000	&	0.6	&	$\alpha_0$	&	-0.034 	&	3.35 	&	3.39 	&	95.1 	&		&	-0.029 	&	3.35 	&	3.37 	&	95.1 	&		&	-0.026 	&	3.35 	&	3.38 	&	95.0 	\\
	&		&	$\alpha_1$	&	-0.033 	&	5.53 	&	5.54 	&	94.8 	&		&	-0.053 	&	5.53 	&	5.56 	&	94.5 	&		&	-0.048 	&	5.53 	&	5.54 	&	94.4 	\\
	&		&	$\tau_1$	&	-0.087 	&	3.94 	&	4.11 	&	94.0 	&		&	-0.098 	&	3.95 	&	4.13 	&	93.8 	&		&	-0.088 	&	3.95 	&	4.11 	&	94.0 	\\
	&		&	$\beta_2$	&	-0.078 	&	2.89 	&	2.94 	&	95.0 	&		&	-0.083 	&	2.89 	&	2.92 	&	94.8 	&		&	-0.084 	&	2.89 	&	2.93 	&	94.8 	\\
	&		&	$\gamma_1$	&	-0.080 	&	2.37 	&	2.45 	&	93.8 	&		&	-0.081 	&	2.37 	&	2.44 	&	94.0 	&		&	-0.079 	&	2.37 	&	2.45 	&	93.9 	\\
	&	0.8	&	$\alpha_0$	&	-0.060 	&	3.35 	&	3.38 	&	95.2 	&		&	-0.042 	&	3.35 	&	3.37 	&	95.1 	&		&	-0.045 	&	3.35 	&	3.38 	&	95.0 	\\
	&		&	$\alpha_1$	&	-0.023 	&	5.52 	&	5.52 	&	94.8 	&		&	-0.056 	&	5.52 	&	5.56 	&	94.6 	&		&	-0.028 	&	5.52 	&	5.54 	&	94.5 	\\
	&		&	$\tau_1$	&	-0.127 	&	3.94 	&	4.11 	&	93.9 	&		&	-0.132 	&	3.94 	&	4.14 	&	93.9 	&		&	-0.097 	&	3.94 	&	4.12 	&	93.9 	\\
	&		&	$\beta_2$	&	-0.050 	&	2.89 	&	2.93 	&	95.1 	&		&	-0.067 	&	2.89 	&	2.92 	&	94.9 	&		&	-0.068 	&	2.89 	&	2.93 	&	94.8 	\\
	&		&	$\gamma_1$	&	-0.075 	&	2.37 	&	2.45 	&	93.8 	&		&	-0.075 	&	2.37 	&	2.44 	&	94.0 	&		&	-0.079 	&	2.37 	&	2.45 	&	94.0 	\\			
 \hline
\end{tabular}}
\label{Tab:case1:chi2}										
\end{table}

\begin{table}
{\footnotesize
\caption*{Table S2: Simulation results (multiplied by 100) for Case 1 with $\alpha^*=(1,-1.5)^\top$, $C_n=\log\{\log (n)\}$
and $\epsilon_i\sim t(4)$.}										
\begin{tabular}{cccccccccccccccccc}		
 \hline
& & &\multicolumn{4}{c}{Oracle}& &\multicolumn{4}{c}{SCAD}& &\multicolumn{4}{c}{MCP}\\				
\cline{4-7} \cline{9-12} \cline{14-17} 	
$n$ & $\nu$	&		&	Bias	&	SE	&	SD	&	CP	&		&	Bias	&	SE	&	SD	&	CP	&		&	Bias	&	SE	&	SD	&	CP	\\	
\hline
1000	&	0.6	&	$\alpha_0$	&	0.373 	&	6.69 	&	6.99 	&	94.4 	&		&	0.478 	&	6.69 	&	6.96 	&	94.2 	&		&	0.415 	&	6.68 	&	7.00 	&	94.3 	\\
	&		&	$\alpha_1$	&	-0.165 	&	11.0 	&	11.2 	&	95.3 	&		&	-0.326 	&	11.0 	&	11.2 	&	94.9 	&		&	-0.214 	&	11.0 	&	11.2 	&	95.0 	\\
	&		&	$\tau_1$	&	0.275 	&	7.92 	&	8.29 	&	93.7 	&		&	0.265 	&	7.93 	&	8.30 	&	93.9 	&		&	0.295 	&	7.92 	&	8.29 	&	93.9 	\\
	&		&	$\beta_2$	&	-0.653 	&	5.86 	&	6.16 	&	94.5 	&		&	-0.767 	&	5.88 	&	6.13 	&	94.4 	&		&	-0.711 	&	5.87 	&	6.17 	&	94.1 	\\
	&		&	$\gamma_1$	&	-0.098 	&	4.71 	&	4.72 	&	95.2 	&		&	-0.096 	&	4.71 	&	4.74 	&	95.0 	&		&	-0.122 	&	4.71 	&	4.72 	&	95.2 	\\
	&	0.8	&	$\alpha_0$	&	0.323 	&	6.68 	&	6.94 	&	94.7 	&		&	0.451 	&	6.68 	&	6.92 	&	94.5 	&		&	0.412 	&	6.68 	&	6.98 	&	94.6 	\\
	&		&	$\alpha_1$	&	-0.144 	&	11.0 	&	11.1 	&	95.3 	&		&	-0.329 	&	11.0 	&	11.1 	&	95.1 	&		&	-0.204 	&	11.0 	&	11.1 	&	95.1 	\\
	&		&	$\tau_1$	&	0.197 	&	7.90 	&	8.23 	&	94.0 	&		&	0.196 	&	7.91 	&	8.35 	&	93.5 	&		&	0.299 	&	7.91 	&	8.28 	&	93.9 	\\
	&		&	$\beta_2$	&	-0.595 	&	5.85 	&	6.10 	&	94.7 	&		&	-0.732 	&	5.87 	&	6.10 	&	94.5 	&		&	-0.699 	&	5.86 	&	6.15 	&	94.2 	\\
	&		&	$\gamma_1$	&	-0.089 	&	4.71 	&	4.71 	&	95.2 	&		&	-0.089 	&	4.71 	&	4.74 	&	95.2 	&		&	-0.105 	&	4.71 	&	4.73 	&	95.1 	\\			 \hline																											
2000	&	0.6	&	$\alpha_0$	&	-0.091 	&	4.75 	&	4.93 	&	94.3 	&		&	-0.087 	&	4.75 	&	4.91 	&	94.3 	&		&	-0.089 	&	4.75 	&	4.94 	&	94.3 	\\
	&		&	$\alpha_1$	&	0.078 	&	7.83 	&	8.19 	&	94.5 	&		&	0.052 	&	7.83 	&	8.10 	&	94.8 	&		&	0.065 	&	7.83 	&	8.18 	&	94.7 	\\
	&		&	$\tau_1$	&	0.148 	&	5.59 	&	5.53 	&	95.2 	&		&	0.125 	&	5.59 	&	5.49 	&	95.4 	&		&	0.137 	&	5.59 	&	5.52 	&	95.1 	\\
	&		&	$\beta_2$	&	-0.124 	&	4.11 	&	4.31 	&	94.3 	&		&	-0.126 	&	4.11 	&	4.28 	&	94.4 	&		&	-0.122 	&	4.11 	&	4.31 	&	94.1 	\\
	&		&	$\gamma_1$	&	-0.024 	&	3.35 	&	3.33 	&	95.0 	&		&	-0.030 	&	3.35 	&	3.33 	&	94.8 	&		&	-0.030 	&	3.35 	&	3.33 	&	94.9 	\\
	&	0.8	&	$\alpha_0$	&	-0.117 	&	4.75 	&	4.92 	&	94.2 	&		&	-0.092 	&	4.75 	&	4.91 	&	94.6 	&		&	-0.102 	&	4.75 	&	4.92 	&	94.4 	\\
	&		&	$\alpha_1$	&	0.083 	&	7.82 	&	8.16 	&	94.0 	&		&	0.034 	&	7.82 	&	8.12 	&	94.5 	&		&	0.088 	&	7.82 	&	8.14 	&	94.4 	\\
	&		&	$\tau_1$	&	0.100 	&	5.58 	&	5.52 	&	95.1 	&		&	0.088 	&	5.58 	&	5.53 	&	94.8 	&		&	0.140 	&	5.58 	&	5.51 	&	94.8 	\\
	&		&	$\beta_2$	&	-0.096 	&	4.10 	&	4.30 	&	94.1 	&		&	-0.117 	&	4.10 	&	4.28 	&	94.1 	&		&	-0.108 	&	4.10 	&	4.29 	&	94.1 	\\
	&		&	$\gamma_1$	&	-0.019 	&	3.35 	&	3.33 	&	95.0 	&		&	-0.027 	&	3.35 	&	3.33 	&	95.1 	&		&	-0.026 	&	3.35 	&	3.33 	&	95.0 	\\			 \hline						
\end{tabular}}											
\end{table}

\begin{table}
{\footnotesize
\caption*{Table S3: Simulation results (multiplied by 100) for Case 2 with $\nu=0.8$,  $C_n=\log\{\log (n)\}$
and $\epsilon_i\sim \mathcal{N}(0,1)$.}											
\begin{tabular}{ccrccccrccccrcccc}		
\hline
& &\multicolumn{4}{c}{Oracle}& &\multicolumn{4}{c}{SCAD}& &\multicolumn{4}{c}{MCP}\\	
\cline{3-6} \cline{8-11} \cline{13-16} 	
$n$ & 	&	Bias	&	SE	&	SD	&	CP	&		&	Bias	&	SE	&	SD	&	CP	&		&	Bias	&	SE	&	SD	&	CP	\\	
\hline
1000 &	$\alpha_0$	&	0.237 	&	4.12 	&	4.11 	&	94.7 	&		&	0.145 	&	4.13 	&	4.16 	&	94.1 	&		&	0.115 	&	4.13 	&	4.18 	&	94.0 	\\
	&	$\alpha_1$	&	-0.786 	&	12.1 	&	12.2 	&	95.3 	&		&	-0.538 	&	12.1 	&	12.8 	&	94.8 	&		&	-0.424 	&	12.1 	&	12.9 	&	94.3 	\\
	&	$\alpha_2$	&	1.256 	&	13.6 	&	13.9 	&	94.4 	&		&	0.974 	&	13.6 	&	14.2 	&	94.0 	&		&	0.971 	&	13.6 	&	14.2 	&	93.9 	\\
	&	$\tau_1$	&	0.011 	&	6.59 	&	6.62 	&	94.9 	&		&	-0.180 	&	6.60 	&	6.94 	&	93.8 	&		&	-0.240 	&	6.60 	&	7.01 	&	93.5 	\\
	&	$\tau_2$	&	0.226 	&	8.51 	&	8.70 	&	94.7 	&		&	0.343 	&	8.53 	&	9.32 	&	93.4 	&		&	0.496 	&	8.54 	&	9.45 	&	92.5 	\\
	&	$\beta_2$	&	0.103 	&	3.25 	&	3.33 	&	93.8 	&		&	0.016 	&	3.26 	&	3.40 	&	93.3 	&		&	-0.019 	&	3.26 	&	3.42 	&	93.2 	\\
	&	$\gamma_1$	&	-0.130 	&	3.32 	&	3.24 	&	95.3 	&		&	-0.130 	&	3.32 	&	3.25 	&	95.2 	&		&	-0.129 	&	3.32 	&	3.25 	&	95.3 	\\
2000&	$\alpha_0$	&	-0.062 	&	2.90 	&	2.86 	&	95.9 	&		&	-0.082 	&	2.90 	&	2.85 	&	95.7 	&		&	-0.101 	&	2.90 	&	2.85 	&	95.8 	\\
	&	$\alpha_1$	&	-0.370 	&	8.53 	&	8.82 	&	93.8 	&		&	-0.345 	&	8.53 	&	8.99 	&	93.6 	&		&	-0.216 	&	8.52 	&	9.02 	&	93.7 	\\
	&	$\alpha_2$	&	0.645 	&	9.54 	&	9.90 	&	94.0 	&		&	0.562 	&	9.54 	&	9.95 	&	94.0 	&		&	0.504 	&	9.54 	&	9.97 	&	93.9 	\\
	&	$\tau_1$	&	0.043 	&	4.65 	&	4.77 	&	94.5 	&		&	0.012 	&	4.65 	&	4.81 	&	94.4 	&		&	-0.050 	&	4.66 	&	4.84 	&	94.3 	\\
	&	$\tau_2$	&	0.231 	&	6.00 	&	6.21 	&	93.0 	&		&	0.211 	&	6.01 	&	6.41 	&	92.9 	&		&	0.325 	&	6.01 	&	6.44 	&	92.8 	\\
	&	$\beta_2$	&	-0.071 	&	2.29 	&	2.40 	&	93.3 	&		&	-0.090 	&	2.30 	&	2.41 	&	93.3 	&		&	-0.107 	&	2.30 	&	2.41 	&	93.0 	\\
	&	$\gamma_1$	&	-0.082 	&	2.34 	&	2.31 	&	95.8 	&		&	-0.080 	&	2.34 	&	2.31 	&	95.7 	&		&	-0.080 	&	2.34 	&	2.31 	&	95.7 	\\			
\hline									
\end{tabular}}	
\label{Table:caseII:normal}											
\end{table}

\begin{table}							
{\footnotesize
\caption*{Table S4: Simulation results (multiplied by 100) for Case 2 with $\alpha^*=(1,-2,2)^\top$,
$C_n=\log\{\log (n)\}$ and $\epsilon_i\sim \text{Schi}^2(2)$.}							
\begin{tabular}{cccccccccccccccccc}		
\hline	
& & &\multicolumn{4}{c}{Oracle}& &\multicolumn{4}{c}{SCAD}& &\multicolumn{4}{c}{MCP}\\				
\cline{4-7} \cline{9-12} \cline{14-17} 	
$n$ & $\nu$	&		&	Bias	&	SE	&	SD	&	CP	&		&	Bias	&	SE	&	SD	&	CP	&		&	Bias	&	SE	&	SD	&	CP	\\		
\hline
1000	&	0.6	&	$\alpha_0$	&	0.034 	&	4.10 	&	4.32 	&	92.6 	&		&	-0.016 	&	4.10 	&	4.33 	&	92.6 	&		&	0.001 	&	4.10 	&	4.35 	&	92.4 	\\
	&		&	$\alpha_1$	&	-1.590 	&	12.2 	&	13.0 	&	93.0 	&		&	-1.566 	&	12.2 	&	13.2 	&	92.6 	&		&	-1.543 	&	12.2 	&	13.1 	&	93.2 	\\
	&		&	$\alpha_2$	&	2.118 	&	13.6 	&	13.9 	&	94.1 	&		&	1.997 	&	13.6 	&	14.0 	&	93.6 	&		&	1.991 	&	13.6 	&	13.9 	&	94.1 	\\
	&		&	$\tau_1$	&	0.349 	&	6.58 	&	6.99 	&	93.0 	&		&	0.290 	&	6.59 	&	7.07 	&	93.0 	&		&	0.298 	&	6.59 	&	6.99 	&	93.6 	\\
	&		&	$\tau_2$	&	-0.278 	&	8.51 	&	8.97 	&	93.2 	&		&	-0.362 	&	8.51 	&	9.15 	&	92.5 	&		&	-0.303 	&	8.51 	&	9.20 	&	92.4 	\\
	&		&	$\beta_2$	&	-0.109 	&	3.24 	&	3.40 	&	93.5 	&		&	-0.142 	&	3.25 	&	3.42 	&	93.2 	&		&	-0.140 	&	3.25 	&	3.43 	&	93.1 	\\
	&		&	$\gamma_1$	&	0.105 	&	3.31 	&	3.34 	&	95.2 	&		&	0.110 	&	3.31 	&	3.30 	&	95.4 	&		&	0.117 	&	3.31 	&	3.32 	&	95.3 	\\
	&	0.8	&	$\alpha_0$	&	0.047 	&	4.10 	&	4.30 	&	92.6 	&		&	-0.028 	&	4.09 	&	4.32 	&	92.8 	&		&	-0.016 	&	4.09 	&	4.34 	&	92.3 	\\
	&		&	$\alpha_1$	&	-1.417 	&	12.1 	&	12.6 	&	93.2 	&		&	-1.504 	&	12.1 	&	13.1 	&	93.0 	&		&	-1.423 	&	12.1 	&	13.0 	&	93.2 	\\
	&		&	$\alpha_2$	&	2.023 	&	13.6 	&	13.7 	&	94.0 	&		&	1.912 	&	13.6 	&	14.0 	&	93.8 	&		&	1.902 	&	13.6 	&	13.8 	&	94.0 	\\
	&		&	$\tau_1$	&	0.308 	&	6.56 	&	6.83 	&	93.5 	&		&	0.257 	&	6.56 	&	7.03 	&	93.2 	&		&	0.244 	&	6.56 	&	6.99 	&	93.2 	\\
	&		&	$\tau_2$	&	-0.141 	&	8.46 	&	8.61 	&	94.3 	&		&	-0.364 	&	8.46 	&	9.13 	&	92.6 	&		&	-0.227 	&	8.46 	&	9.22 	&	92.2 	\\
	&		&	$\beta_2$	&	-0.092 	&	3.23 	&	3.35 	&	94.0 	&		&	-0.153 	&	3.24 	&	3.41 	&	93.7 	&		&	-0.155 	&	3.24 	&	3.42 	&	93.5 	\\
	&		&	$\gamma_1$	&	0.101 	&	3.31 	&	3.34 	&	95.2 	&		&	0.112 	&	3.31 	&	3.30 	&	95.4 	&		&	0.111 	&	3.31 	&	3.31 	&	95.4 	\\			\hline																												
2000	&	0.6	&	$\alpha_0$	&	0.060 	&	2.91 	&	3.06 	&	93.6 	&		&	0.055 	&	2.91 	&	3.06 	&	93.5 	&		&	0.031 	&	2.91 	&	3.06 	&	93.7 	\\
	&		&	$\alpha_1$	&	-0.547 	&	8.58 	&	9.11 	&	93.2 	&		&	-0.595 	&	8.59 	&	9.25 	&	93.0 	&		&	-0.490 	&	8.58 	&	9.32 	&	92.6 	\\
	&		&	$\alpha_2$	&	0.662 	&	9.57 	&	9.79 	&	94.3 	&		&	0.641 	&	9.57 	&	9.88 	&	93.9 	&		&	0.587 	&	9.57 	&	9.90 	&	94.0 	\\
	&		&	$\tau_1$	&	0.062 	&	4.67 	&	4.80 	&	94.1 	&		&	0.068 	&	4.67 	&	4.86 	&	93.6 	&		&	0.015 	&	4.67 	&	4.89 	&	93.6 	\\
	&		&	$\tau_2$	&	-0.102 	&	6.03 	&	6.57 	&	92.8 	&		&	-0.179 	&	6.03 	&	6.68 	&	92.3 	&		&	-0.087 	&	6.04 	&	6.72 	&	92.3 	\\
	&		&	$\beta_2$	&	-0.045 	&	2.30 	&	2.45 	&	93.2 	&		&	-0.050 	&	2.30 	&	2.44 	&	93.4 	&		&	-0.069 	&	2.30 	&	2.44 	&	93.4 	\\
	&		&	$\gamma_1$	&	0.142 	&	2.35 	&	2.21 	&	96.3 	&		&	0.145 	&	2.35 	&	2.21 	&	96.4 	&		&	0.146 	&	2.35 	&	2.21 	&	96.4 	\\
	&	0.8	&	$\alpha_0$	&	0.070 	&	2.91 	&	3.06 	&	93.2 	&		&	0.048 	&	2.91 	&	3.06 	&	93.3 	&		&	0.032 	&	2.91 	&	3.06 	&	93.5 	\\
	&		&	$\alpha_1$	&	-0.467 	&	8.55 	&	8.92 	&	93.4 	&		&	-0.577 	&	8.56 	&	9.22 	&	93.0 	&		&	-0.430 	&	8.55 	&	9.20 	&	92.5 	\\
	&		&	$\alpha_2$	&	0.628 	&	9.55 	&	9.73 	&	94.4 	&		&	0.587 	&	9.56 	&	9.84 	&	94.3 	&		&	0.541 	&	9.55 	&	9.85 	&	93.8 	\\
	&		&	$\tau_1$	&	0.041 	&	4.65 	&	4.74 	&	94.3 	&		&	0.052 	&	4.65 	&	4.84 	&	93.9 	&		&	-0.009 	&	4.65 	&	4.85 	&	93.5 	\\
	&		&	$\tau_2$	&	-0.025 	&	6.01 	&	6.40 	&	93.7 	&		&	-0.209 	&	6.01 	&	6.69 	&	92.1 	&		&	-0.048 	&	6.01 	&	6.63 	&	92.7 	\\
	&		&	$\beta_2$	&	-0.034 	&	2.30 	&	2.45 	&	93.0 	&		&	-0.055 	&	2.30 	&	2.43 	&	93.6 	&		&	-0.070 	&	2.30 	&	2.43 	&	93.3 	\\
	&		&	$\gamma_1$	&	0.140 	&	2.34 	&	2.21 	&	96.3 	&		&	0.146 	&	2.34 	&	2.21 	&	96.4 	&		&	0.146 	&	2.35 	&	2.21 	&	96.4 	\\			\hline										
\end{tabular}}											
\end{table}

\begin{table}	
{\footnotesize
\caption*{Table S5: Simulation results (multiplied by 100) for Case 2 with $\alpha^*=(1,-2,2)^\top$, $C_n=\log\{\log (n)\}$
and $\epsilon_i\sim t(4)$.}								
\begin{tabular}{cccccccccccccccccc}	
\hline
& & &\multicolumn{4}{c}{Oracle}& &\multicolumn{4}{c}{SCAD}& &\multicolumn{4}{c}{MCP}\\				
\cline{4-7} \cline{9-12} \cline{14-17} 	
$n$ & $\nu$	&		&	Bias	&	SE	&	SD	&	CP	&		&	Bias	&	SE	&	SD	&	CP	&		&	Bias	&	SE	&	SD	&	CP	\\		
\hline
1000	&	0.6	&	$\alpha_0$	&	-0.097 	&	5.82 	&	6.15 	&	93.1 	&		&	-0.152 	&	5.82 	&	6.20 	&	92.7 	&		&	-0.153 	&	5.78 	&	6.24 	&	92.1 	\\
	&		&	$\alpha_1$	&	-2.133 	&	17.4 	&	19.2 	&	92.3 	&		&	-2.252 	&	17.4 	&	20.1 	&	91.7 	&		&	-2.286 	&	17.3 	&	20.2 	&	91.2 	\\
	&		&	$\alpha_2$	&	3.438 	&	19.4 	&	21.2 	&	93.1 	&		&	3.394 	&	19.5 	&	21.4 	&	92.7 	&		&	3.581 	&	19.3 	&	21.3 	&	92.5 	\\
	&		&	$\tau_1$	&	0.083 	&	9.38 	&	10.4 	&	93.3 	&		&	-0.029 	&	9.40 	&	10.7 	&	92.6 	&		&	0.003 	&	9.33 	&	10.7 	&	92.2 	\\
	&		&	$\tau_2$	&	0.185 	&	12.1 	&	13.1 	&	93.2 	&		&	-0.006 	&	12.1 	&	13.8 	&	91.9 	&		&	0.216 	&	12.1 	&	13.7 	&	91.3 	\\
	&		&	$\beta_2$	&	-0.153 	&	4.63 	&	4.89 	&	94.2 	&		&	-0.203 	&	4.63 	&	4.97 	&	93.9 	&		&	-0.218 	&	4.60 	&	4.97 	&	93.8 	\\
	&		&	$\gamma_1$	&	0.199 	&	4.69 	&	4.76 	&	94.3 	&		&	0.153 	&	4.69 	&	4.69 	&	94.3 	&		&	0.187 	&	4.67 	&	4.69 	&	94.2 	\\
	&	0.8	&	$\alpha_0$	&	-0.070 	&	5.81 	&	6.09 	&	92.9 	&		&	-0.173 	&	5.81 	&	6.20 	&	92.6 	&		&	-0.170 	&	5.77 	&	6.24 	&	92.4 	\\
	&		&	$\alpha_1$	&	-1.929 	&	17.3 	&	18.5 	&	92.5 	&		&	-2.015 	&	17.3 	&	20.0 	&	91.8 	&		&	-2.100 	&	17.2 	&	20.2 	&	90.9 	\\
	&		&	$\alpha_2$	&	3.239 	&	19.4 	&	20.9 	&	93.0 	&		&	3.121 	&	19.4 	&	21.3 	&	92.7 	&		&	3.400 	&	19.3 	&	21.3 	&	92.5 	\\
	&		&	$\tau_1$	&	0.051 	&	9.34 	&	10.1 	&	93.8 	&		&	-0.168 	&	9.37 	&	10.6 	&	92.9 	&		&	-0.118 	&	9.30 	&	10.6 	&	92.4 	\\
	&		&	$\tau_2$	&	0.225 	&	12.1 	&	12.6 	&	93.8 	&		&	0.021 	&	12.1 	&	13.8 	&	91.7 	&		&	0.284 	&	12.0 	&	13.9 	&	90.9 	\\
	&		&	$\beta_2$	&	-0.115 	&	4.61 	&	4.83 	&	94.3 	&		&	-0.200 	&	4.62 	&	4.96 	&	93.8 	&		&	-0.225 	&	4.59 	&	4.97 	&	93.8 	\\
	&		&	$\gamma_1$	&	0.197 	&	4.69 	&	4.75 	&	94.3 	&		&	0.164 	&	4.69 	&	4.69 	&	94.5 	&		&	0.189 	&	4.67 	&	4.66 	&	94.3 	\\				\hline																												
2000	&	0.6	&	$\alpha_0$	&	0.016 	&	4.12 	&	4.13 	&	95.1 	&		&	0.023 	&	4.12 	&	4.19 	&	94.5 	&		&	-0.002 	&	4.12 	&	4.18 	&	94.5 	\\
	&		&	$\alpha_1$	&	-0.851 	&	12.2 	&	12.4 	&	94.1 	&		&	-1.009 	&	12.2 	&	12.9 	&	92.9 	&		&	-0.825 	&	12.2 	&	12.9 	&	92.9 	\\
	&		&	$\alpha_2$	&	1.085 	&	13.6 	&	13.8 	&	94.5 	&		&	1.148 	&	13.6 	&	13.8 	&	94.3 	&		&	1.048 	&	13.6 	&	14.0 	&	94.2 	\\
	&		&	$\tau_1$	&	0.088 	&	6.60 	&	6.69 	&	94.9 	&		&	0.114 	&	6.60 	&	6.86 	&	94.1 	&		&	0.038 	&	6.61 	&	6.88 	&	94.2 	\\
	&		&	$\tau_2$	&	0.005 	&	8.55 	&	9.23 	&	93.2 	&		&	-0.122 	&	8.55 	&	9.76 	&	91.7 	&		&	0.045 	&	8.56 	&	9.72 	&	91.6 	\\
	&		&	$\beta_2$	&	-0.019 	&	3.26 	&	3.38 	&	93.5 	&		&	-0.021 	&	3.26 	&	3.44 	&	93.3 	&		&	-0.039 	&	3.26 	&	3.42 	&	92.6 	\\
	&		&	$\gamma_1$	&	0.073 	&	3.32 	&	3.41 	&	94.5 	&		&	0.078 	&	3.32 	&	3.40 	&	94.6 	&		&	0.077 	&	3.32 	&	3.41 	&	94.5 	\\
	&	0.8	&	$\alpha_0$	&	0.031 	&	4.12 	&	4.11 	&	95.3 	&		&	0.014 	&	4.12 	&	4.18 	&	95.0 	&		&	-0.019 	&	4.12 	&	4.18 	&	94.4 	\\
	&		&	$\alpha_1$	&	-0.807 	&	12.1 	&	12.1 	&	94.0 	&		&	-0.941 	&	12.1 	&	12.8 	&	93.0 	&		&	-0.707 	&	12.1 	&	12.8 	&	93.0 	\\
	&		&	$\alpha_2$	&	1.034 	&	13.6 	&	13.6 	&	94.5 	&		&	1.033 	&	13.6 	&	13.8 	&	94.1 	&		&	0.951 	&	13.6 	&	13.9 	&	93.6 	\\
	&		&	$\tau_1$	&	0.085 	&	6.58 	&	6.55 	&	95.1 	&		&	0.074 	&	6.58 	&	6.81 	&	94.0 	&		&	-0.011 	&	6.59 	&	6.82 	&	94.2 	\\
	&		&	$\tau_2$	&	-0.005 	&	8.51 	&	8.94 	&	93.5 	&		&	-0.143 	&	8.51 	&	9.61 	&	92.1 	&		&	0.113 	&	8.52 	&	9.68 	&	91.9 	\\
	&		&	$\beta_2$	&	-0.001 	&	3.25 	&	3.36 	&	93.3 	&		&	-0.027 	&	3.25 	&	3.43 	&	92.9 	&		&	-0.055 	&	3.25 	&	3.43 	&	93.2 	\\
	&		&	$\gamma_1$	&	0.072 	&	3.32 	&	3.40 	&	94.6 	&		&	0.081 	&	3.32 	&	3.41 	&	94.7 	&		&	0.079 	&	3.32 	&	3.41 	&	94.7 	\\			\hline									
\end{tabular}}											
\end{table}

\begin{table}
{\footnotesize
\caption*{Table S6: Simulation results (multiplied by 100) for Case 3 with $\nu=0.8$, $C_n=\log\{\log (n)\}$ 
and $\epsilon_i\sim \mathcal{N}(0,1)$.}										
\begin{tabular}{ccrccccrccccrcccc}																											
\hline	
& &\multicolumn{4}{c}{Oracle}& &\multicolumn{4}{c}{SCAD}& &\multicolumn{4}{c}{MCP}\\ 	
\cline{3-6} \cline{8-11} \cline{13-16}
$n$ & 	&	Bias	&	SE	&	SD	&	CP	&		&	Bias	&	SE	&	SD	&	CP	&		&	Bias	&	SE	&	SD	&	CP	\\		
\hline
1000  &	$\alpha_0$	&	0.072 	&	3.04 	&	3.18 	&	94.0 	&		&	0.035 	&	3.04 	&	3.19 	&	94.0 	&		&	0.028 	&	3.04 	&	3.19 	&	93.9 	\\
	 &	$\alpha_1$	&	0.873 	&	16.2 	&	17.3 	&	93.3 	&		&	1.808 	&	16.4 	&	17.4 	&	93.4 	&		&	1.894 	&	16.4 	&	17.3 	&	93.4 	\\
	 &	$\alpha_2$	&	-1.156 	&	15.7 	&	16.7 	&	93.8 	&		&	-1.632 	&	15.8 	&	16.9 	&	93.6 	&		&	-1.708 	&	15.8 	&	16.8 	&	93.6 	\\
	 &	$\alpha_3$	&	-1.192 	&	15.8 	&	16.6 	&	93.6 	&		&	-0.118 	&	15.6 	&	17.0 	&	93.3 	&		&	-0.152 	&	15.6 	&	17.0 	&	93.3 	\\
	 &	$\alpha_4$	&	1.600 	&	18.3 	&	19.0 	&	93.4 	&		&	0.324 	&	18.2 	&	19.5 	&	93.1 	&		&	0.374 	&	18.2 	&	19.4 	&	93.1 	\\
	 &	$\tau_1$	&	-0.618 	&	12.5 	&	13.0 	&	93.1 	&		&	-0.279 	&	12.5 	&	13.1 	&	93.0 	&		&	-0.224 	&	12.5 	&	13.0 	&	93.2 	\\
	 &	$\tau_2$	&	-0.297 	&	11.2 	&	11.9 	&	93.1 	&		&	-1.122 	&	11.2 	&	11.8 	&	92.9 	&		&	-1.156 	&	11.2 	&	11.8 	&	92.8 	\\
	 &	$\tau_3$	&	0.956 	&	11.3 	&	11.8 	&	93.3 	&		&	-0.412 	&	11.3 	&	12.6 	&	92.1 	&		&	-0.427 	&	11.3 	&	12.6 	&	92.2 	\\
	 &	$\tau_4$	&	0.581 	&	13.5 	&	14.2 	&	93.2 	&		&	1.155 	&	13.6 	&	14.4 	&	93.2 	&		&	1.119 	&	13.6 	&	14.4 	&	93.1 	\\
	 &	$\beta_2$	&	-0.251 	&	5.07 	&	5.32 	&	93.9 	&		&	-0.229 	&	5.07 	&	5.34 	&	94.1 	&		&	-0.213 	&	5.07 	&	5.34 	&	94.0 	\\
	 &	$\gamma_1$	&	-0.216 	&	3.46 	&	3.56 	&	95.0 	&		&	-0.215 	&	3.46 	&	3.54 	&	95.1 	&		&	-0.216 	&	3.46 	&	3.54 	&	95.0 	\\	
2000  &	$\alpha_0$	&	0.030 	&	2.15 	&	2.17 	&	94.2 	&		&	-0.009 	&	2.15 	&	2.16 	&	94.3 	&		&	-0.014 	&	2.15 	&	2.16 	&	94.2 	\\
	 &	$\alpha_1$	&	-0.252 	&	11.4 	&	11.6 	&	93.6 	&		&	0.403 	&	11.5 	&	11.7 	&	94.1 	&		&	0.487 	&	11.5 	&	11.6 	&	94.1 	\\
	 &	$\alpha_2$	&	0.106 	&	11.0 	&	11.2 	&	94.2 	&		&	-0.241 	&	11.1 	&	11.2 	&	94.1 	&		&	-0.317 	&	11.1 	&	11.2 	&	94.1 	\\
	 &	$\alpha_3$	&	-0.604 	&	11.1 	&	11.4 	&	94.0 	&		&	0.093 	&	11.0 	&	11.6 	&	93.5 	&		&	0.029 	&	11.0 	&	11.6 	&	93.5 	\\
	 &	$\alpha_4$	&	0.860 	&	12.9 	&	13.4 	&	94.5 	&		&	0.035 	&	12.8 	&	13.7 	&	94.1 	&		&	0.103 	&	12.8 	&	13.7 	&	94.0 	\\
	 &	$\tau_1$	&	-0.538 	&	8.84 	&	8.88 	&	93.6 	&		&	-0.256 	&	8.83 	&	8.89 	&	93.8 	&		&	-0.214 	&	8.83 	&	8.87 	&	93.8 	\\
	 &	$\tau_2$	&	0.285 	&	7.93 	&	8.15 	&	93.9 	&		&	-0.238 	&	7.94 	&	8.11 	&	94.0 	&		&	-0.279 	&	7.94 	&	8.10 	&	94.2 	\\
	 &	$\tau_3$	&	0.296 	&	7.93 	&	8.25 	&	94.2 	&		&	-0.602 	&	7.93 	&	8.61 	&	92.6 	&		&	-0.589 	&	7.93 	&	8.60 	&	92.5 	\\
	 &	$\tau_4$	&	0.201 	&	9.56 	&	9.94 	&	93.3 	&		&	0.494 	&	9.57 	&	9.97 	&	93.3 	&		&	0.457 	&	9.57 	&	9.97 	&	93.2 	\\
	 &	$\beta_2$	&	-0.168 	&	3.58 	&	3.61 	&	94.7 	&		&	-0.123 	&	3.58 	&	3.59 	&	95.0 	&		&	-0.112 	&	3.58 	&	3.59 	&	95.1 	\\
	 &	$\gamma_1$	&	-0.032 	&	2.45 	&	2.47 	&	93.7 	&		&	-0.033 	&	2.45 	&	2.46 	&	93.6 	&		&	-0.032 	&	2.45 	&	2.46 	&	93.6 	\\		
\hline										
\end{tabular}}	
\label{Table:caseIII1:normal}											
\end{table}

\begin{table}
{\footnotesize
\caption*{Table S7: Simulation results (multiplied by 100) for Case 3 with $\nu=0.6$, $\alpha^*=(-1,3,-2,-2,3)^\top$, $C_n=\log\{\log (n)\}$
and $\epsilon_i\sim \text{Schi}^2(2)$.}								
\begin{tabular}{ccrccccrccccrcccc}		
\hline
& &\multicolumn{4}{c}{Oracle}& &\multicolumn{4}{c}{SCAD}& &\multicolumn{4}{c}{MCP}\\				
\cline{3-6} \cline{8-11} \cline{13-16}
$n$& &	Bias	&	SE	&	SD	&	CP	&		&	Bias	&	SE	&	SD	&	CP	&		&	Bias	&	SE	&	SD	&	CP	\\		
\hline
1000&	$\alpha_0$	&	0.286 	&	3.03 	&	3.16 	&	93.9 	&		&	0.285 	&	3.03 	&	3.21 	&	93.1 	&		&	0.278 	&	3.03 	&	3.20 	&	93.1 	\\
	&	$\alpha_1$	&	0.697 	&	16.3 	&	17.3 	&	93.3 	&		&	1.053 	&	16.3 	&	17.4 	&	93.2 	&		&	1.020 	&	16.3 	&	17.4 	&	93.2 	\\
	&	$\alpha_2$	&	-1.292 	&	15.7 	&	16.5 	&	94.0 	&		&	-1.475 	&	15.8 	&	16.6 	&	93.7 	&		&	-1.440 	&	15.8 	&	16.6 	&	93.7 	\\
	&	$\alpha_3$	&	-1.020 	&	15.9 	&	17.5 	&	92.2 	&		&	-0.340 	&	15.8 	&	17.8 	&	91.9 	&		&	-0.324 	&	15.8 	&	17.8 	&	91.8 	\\
	&	$\alpha_4$	&	1.273 	&	18.4 	&	19.5 	&	94.0 	&		&	0.646 	&	18.3 	&	19.9 	&	93.6 	&		&	0.618 	&	18.3 	&	19.9 	&	93.8 	\\
	&	$\tau_1$	&	-0.940 	&	12.5 	&	13.0 	&	94.3 	&		&	-0.862 	&	12.5 	&	13.3 	&	93.7 	&		&	-0.853 	&	12.5 	&	13.3 	&	93.7 	\\
	&	$\tau_2$	&	-0.193 	&	11.3 	&	12.9 	&	91.5 	&		&	-0.585 	&	11.3 	&	12.9 	&	91.3 	&		&	-0.542 	&	11.3 	&	12.9 	&	91.1 	\\
	&	$\tau_3$	&	1.019 	&	11.3 	&	11.9 	&	92.3 	&		&	0.435 	&	11.3 	&	12.2 	&	91.7 	&		&	0.428 	&	11.3 	&	12.1 	&	91.8 	\\
	&	$\tau_4$	&	1.058 	&	13.6 	&	14.9 	&	92.0 	&		&	1.563 	&	13.6 	&	15.3 	&	91.3 	&		&	1.540 	&	13.6 	&	15.3 	&	91.6 	\\
	&	$\beta_2$	&	-0.481 	&	5.07 	&	5.34 	&	93.9 	&		&	-0.508 	&	5.08 	&	5.44 	&	93.2 	&		&	-0.498 	&	5.08 	&	5.43 	&	93.3 	\\
	&	$\gamma_1$	&	0.064 	&	3.45 	&	3.49 	&	95.3 	&		&	0.061 	&	3.45 	&	3.47 	&	95.4 	&		&	0.062 	&	3.45 	&	3.47 	&	95.4 	\\
2000&	$\alpha_0$	&	0.074 	&	2.16 	&	2.21 	&	94.6 	&		&	0.048 	&	2.16 	&	2.22 	&	94.6 	&		&	0.047 	&	2.16 	&	2.22 	&	94.5 	\\
	&	$\alpha_1$	&	0.522 	&	11.5 	&	11.7 	&	94.2 	&		&	0.884 	&	11.5 	&	11.9 	&	93.9 	&		&	0.902 	&	11.5 	&	11.9 	&	93.8 	\\
	&	$\alpha_2$	&	-0.779 	&	11.1 	&	11.4 	&	94.3 	&		&	-1.008 	&	11.1 	&	11.5 	&	94.2 	&		&	-1.018 	&	11.1 	&	11.5 	&	94.1 	\\
	&	$\alpha_3$	&	-0.403 	&	11.1 	&	12.2 	&	92.2 	&		&	-0.129 	&	11.1 	&	12.3 	&	92.0 	&		&	-0.134 	&	11.1 	&	12.3 	&	91.9 	\\
	&	$\alpha_4$	&	0.408 	&	12.9 	&	13.8 	&	92.5 	&		&	0.121 	&	12.9 	&	13.9 	&	92.4 	&		&	0.116 	&	12.9 	&	13.9 	&	92.5 	\\
	&	$\tau_1$	&	-0.427 	&	8.85 	&	8.93 	&	94.9 	&		&	-0.253 	&	8.84 	&	9.05 	&	94.8 	&		&	-0.248 	&	8.84 	&	9.05 	&	94.7 	\\
	&	$\tau_2$	&	-0.324 	&	7.96 	&	8.84 	&	92.4 	&		&	-0.537 	&	7.96 	&	8.87 	&	92.5 	&		&	-0.560 	&	7.96 	&	8.85 	&	92.5 	\\
	&	$\tau_3$	&	0.547 	&	7.97 	&	8.52 	&	92.9 	&		&	0.180 	&	7.97 	&	8.63 	&	92.8 	&		&	0.176 	&	7.97 	&	8.62 	&	92.8 	\\
	&	$\tau_4$	&	0.340 	&	9.59 	&	10.1 	&	93.9 	&		&	0.444 	&	9.60 	&	10.2 	&	93.7 	&		&	0.447 	&	9.60 	&	10.2 	&	93.7 	\\
	&	$\beta_2$	&	-0.217 	&	3.59 	&	3.74 	&	94.3 	&		&	-0.184 	&	3.59 	&	3.76 	&	94.6 	&		&	-0.185 	&	3.59 	&	3.76 	&	94.5 	\\
	&	$\gamma_1$	&	-0.105 	&	2.45 	&	2.47 	&	94.9 	&		&	-0.108 	&	2.45 	&	2.47 	&	94.9 	&		&	-0.108 	&	2.45 	&	2.47 	&	94.9 	\\
\hline										
\end{tabular}}											
\end{table}

\begin{table}
{\footnotesize
\caption*{Table S8: Simulation results (multiplied by 100) for Case 3 with $\nu=0.8$, $\alpha^*=(-1,3,-2,-2,3)^\top$, $C_n=\log\{\log (n)\}$
and $\epsilon_i\sim \text{Schi}^2(2)$.}								
\begin{tabular}{ccrccccrccccrcccc}		
\hline
& &\multicolumn{4}{c}{Oracle}& &\multicolumn{4}{c}{SCAD}& &\multicolumn{4}{c}{MCP}\\				
\cline{3-6} \cline{8-11} \cline{13-16}
$n$&	&	Bias	&	SE	&	SD	&	CP	&		&	Bias	&	SE	&	SD	&	CP	&		&	Bias	&	SE	&	SD	&	CP	\\		
\hline
1000	&	$\alpha_0$	&	0.277 	&	3.03 	&	3.14 	&	93.9 	&		&	0.297 	&	3.03 	&	3.20 	&	93.1 	&		&	0.289 	&	3.03 	&	3.20 	&	93.4 	\\
	&	$\alpha_1$	&	0.466 	&	16.2 	&	16.7 	&	93.6 	&		&	0.834 	&	16.2 	&	17.2 	&	93.7 	&		&	0.885 	&	16.2 	&	17.2 	&	93.5 	\\
	&	$\alpha_2$	&	-1.068 	&	15.6 	&	16.2 	&	94.0 	&		&	-1.220 	&	15.7 	&	16.5 	&	93.6 	&		&	-1.254 	&	15.7 	&	16.5 	&	93.7 	\\
	&	$\alpha_3$	&	-0.842 	&	15.8 	&	17.0 	&	93.0 	&		&	0.206 	&	15.6 	&	17.7 	&	91.8 	&		&	0.214 	&	15.6 	&	17.7 	&	91.6 	\\
	&	$\alpha_4$	&	1.156 	&	18.3 	&	19.2 	&	94.2 	&		&	0.151 	&	18.2 	&	19.9 	&	93.5 	&		&	0.138 	&	18.2 	&	19.9 	&	93.4 	\\
	&	$\tau_1$	&	-0.957 	&	12.5 	&	12.9 	&	94.3 	&		&	-0.966 	&	12.5 	&	13.3 	&	93.7 	&		&	-0.926 	&	12.5 	&	13.2 	&	94.1 	\\
	&	$\tau_2$	&	-0.060 	&	11.2 	&	12.6 	&	91.9 	&		&	-0.558 	&	11.2 	&	12.8 	&	91.1 	&		&	-0.571 	&	11.2 	&	12.8 	&	91.2 	\\
	&	$\tau_3$	&	0.975 	&	11.2 	&	11.6 	&	93.2 	&		&	0.082 	&	11.2 	&	12.1 	&	92.0 	&		&	0.065 	&	11.2 	&	12.1 	&	92.2 	\\
	&	$\tau_4$	&	1.057 	&	13.5 	&	14.6 	&	92.6 	&		&	1.844 	&	13.6 	&	15.3 	&	91.2 	&		&	1.841 	&	13.6 	&	15.3 	&	91.3 	\\
	&	$\beta_2$	&	-0.452 	&	5.06 	&	5.30 	&	93.9 	&		&	-0.525 	&	5.07 	&	5.44 	&	93.1 	&		&	-0.515 	&	5.07 	&	5.43 	&	93.3 	\\
	&	$\gamma_1$	&	0.063 	&	3.45 	&	3.49 	&	95.3 	&		&	0.056 	&	3.45 	&	3.49 	&	95.2 	&		&	0.056 	&	3.45 	&	3.49 	&	95.2 	\\	
2000	&	$\alpha_0$	&	0.079 	&	2.15 	&	2.19 	&	94.6 	&		&	0.055 	&	2.15 	&	2.22 	&	94.5 	&		&	0.051 	&	2.15 	&	2.22 	&	94.6 	\\
	&	$\alpha_1$	&	0.371 	&	11.4 	&	11.5 	&	95.0 	&		&	0.830 	&	11.5 	&	11.8 	&	94.2 	&		&	0.861 	&	11.5 	&	11.8 	&	94.4 	\\
	&	$\alpha_2$	&	-0.645 	&	11.1 	&	11.2 	&	94.4 	&		&	-0.920 	&	11.1 	&	11.4 	&	94.3 	&		&	-0.937 	&	11.1 	&	11.4 	&	94.5 	\\
	&	$\alpha_3$	&	-0.342 	&	11.1 	&	12.0 	&	92.4 	&		&	0.171 	&	11.0 	&	12.3 	&	91.8 	&		&	0.147 	&	11.0 	&	12.2 	&	91.8 	\\
	&	$\alpha_4$	&	0.363 	&	12.9 	&	13.7 	&	92.7 	&		&	-0.179 	&	12.8 	&	13.9 	&	92.4 	&		&	-0.172 	&	12.8 	&	13.9 	&	92.4 	\\
	&	$\tau_1$	&	-0.479 	&	8.83 	&	8.82 	&	95.1 	&		&	-0.297 	&	8.83 	&	9.05 	&	94.4 	&		&	-0.276 	&	8.83 	&	9.04 	&	94.6 	\\
	&	$\tau_2$	&	-0.261 	&	7.94 	&	8.76 	&	92.5 	&		&	-0.584 	&	7.94 	&	8.80 	&	92.2 	&		&	-0.604 	&	7.94 	&	8.81 	&	92.4 	\\
	&	$\tau_3$	&	0.556 	&	7.95 	&	8.38 	&	92.9 	&		&	-0.023 	&	7.95 	&	8.61 	&	92.6 	&		&	-0.033 	&	7.95 	&	8.61 	&	92.6 	\\
	&	$\tau_4$	&	0.359 	&	9.57 	&	10.0 	&	94.4 	&		&	0.612 	&	9.58 	&	10.2 	&	93.9 	&		&	0.589 	&	9.58 	&	10.2 	&	93.6 	\\
	&	$\beta_2$	&	-0.220 	&	3.58 	&	3.71 	&	94.5 	&		&	-0.200 	&	3.58 	&	3.77 	&	94.1 	&		&	-0.192 	&	3.58 	&	3.77 	&	94.1 	\\
	&	$\gamma_1$	&	-0.106 	&	2.45 	&	2.47 	&	94.9 	&		&	-0.107 	&	2.45 	&	2.47 	&	94.9 	&		&	-0.107 	&	2.45 	&	2.47 	&	94.9 	\\	
 \hline										
\end{tabular}}											
\end{table}

\begin{table}
{\footnotesize
\caption*{Table S9: Simulation results (multiplied by 100) for Case 3 with $\nu=0.6$, $\alpha^*=(-1,3,-2,-2,3)^\top$, $C_n=\log\{\log (n)\}$
and $\epsilon_i\sim t(4)$.}							
\begin{tabular}{ccrccccrccccrcccc}
\hline
& &\multicolumn{4}{c}{Oracle}& &\multicolumn{4}{c}{SCAD}& &\multicolumn{4}{c}{MCP}\\				
\cline{3-6} \cline{8-11} \cline{13-16} 	
$n$& $\nu$ 	& Bias	&	SE	&	SD	&	CP	&		&	Bias	&	SE	&	SD	&	CP	&		&	Bias	&	SE	&	SD	&	CP	\\	
\hline
1000&	$\alpha_0$	&	0.095 	&	4.29 	&	4.38 	&	95.4 	&		&	0.104 	&	4.29 	&	4.43 	&	95.1 	&		&	0.085 	&	4.29 	&	4.42 	&	94.9 	\\
	&	$\alpha_1$	&	2.647 	&	23.3 	&	25.3 	&	91.8 	&		&	3.535 	&	23.5 	&	25.6 	&	91.8 	&		&	3.707 	&	23.5 	&	25.5 	&	91.8 	\\
	&	$\alpha_2$	&	-2.741 	&	22.6 	&	25.0 	&	91.5 	&		&	-3.040 	&	22.8 	&	25.2 	&	92.0 	&		&	-3.191 	&	22.8 	&	25.1 	&	92.0 	\\
	&	$\alpha_3$	&	-3.172 	&	22.7 	&	24.4 	&	92.7 	&		&	-1.635 	&	22.4 	&	24.4 	&	92.4 	&		&	-1.774 	&	22.4 	&	24.4 	&	92.4 	\\
	&	$\alpha_4$	&	4.188 	&	26.3 	&	28.1 	&	93.3 	&		&	2.227 	&	26.0 	&	28.1 	&	92.8 	&		&	2.401 	&	26.0 	&	28.1 	&	92.7 	\\
	&	$\tau_1$	&	-0.648 	&	17.7 	&	18.6 	&	93.1 	&		&	-0.535 	&	17.7 	&	19.0 	&	93.0 	&		&	-0.411 	&	17.7 	&	18.8 	&	93.1 	\\
	&	$\tau_2$	&	-0.894 	&	16.0 	&	17.2 	&	92.5 	&		&	-2.033 	&	16.0 	&	17.7 	&	92.3 	&		&	-2.087 	&	16.0 	&	17.6 	&	92.2 	\\
	&	$\tau_3$	&	1.071 	&	16.0 	&	17.0 	&	91.7 	&		&	-0.905 	&	16.0 	&	18.2 	&	90.0 	&		&	-0.868 	&	16.0 	&	18.2 	&	90.0	\\
	&	$\tau_4$	&	0.965 	&	19.2 	&	20.0 	&	94.4 	&		&	1.845 	&	19.3 	&	20.4 	&	94.0 	&		&	1.755 	&	19.3 	&	20.3 	&	94.0 	\\
	&	$\beta_2$	&	-0.553 	&	7.17 	&	7.32 	&	95.0 	&		&	-0.598 	&	7.19 	&	7.46 	&	94.7 	&		&	-0.564 	&	7.18 	&	7.43 	&	94.7 	\\
	&	$\gamma_1$	&	-0.023 	&	4.88 	&	4.99 	&	94.6 	&		&	-0.017 	&	4.88 	&	4.99 	&	94.6 	&		&	-0.017 	&	4.88 	&	4.99 	&	94.6 	\\
2000&	$\alpha_0$	&	-0.044 	&	3.04 	&	3.17 	&	93.4 	&		&	-0.047 	&	3.04 	&	3.18 	&	93.5 	&		&	-0.060 	&	3.04 	&	3.17 	&	93.8 	\\
	&	$\alpha_1$	&	1.270 	&	16.3 	&	16.6 	&	94.1 	&		&	1.809 	&	16.4 	&	16.7 	&	94.0 	&		&	1.985 	&	16.4 	&	16.7 	&	93.9 	\\
	&	$\alpha_2$	&	-1.136 	&	15.8 	&	16.3 	&	93.7 	&		&	-1.403 	&	15.9 	&	16.4 	&	93.9 	&		&	-1.548 	&	15.9 	&	16.4 	&	93.9 	\\
	&	$\alpha_3$	&	-1.065 	&	15.7 	&	16.8 	&	92.6 	&		&	-0.496 	&	15.6 	&	17.0 	&	92.5 	&		&	-0.583 	&	15.6 	&	16.9 	&	92.4 	\\
	&	$\alpha_4$	&	1.499 	&	18.2 	&	18.7 	&	93.3 	&		&	0.788 	&	18.2 	&	18.9 	&	93.6 	&		&	0.850 	&	18.2 	&	18.8 	&	93.8 	\\
	&	$\tau_1$	&	-0.061 	&	12.5 	&	12.6 	&	94.0 	&		&	0.065 	&	12.5 	&	12.8 	&	94.2 	&		&	0.152 	&	12.5 	&	12.8 	&	94.3 	\\
	&	$\tau_2$	&	-0.249 	&	11.3 	&	12.7 	&	91.9 	&		&	-0.841 	&	11.3 	&	12.4 	&	92.6 	&		&	-0.925 	&	11.3 	&	12.4 	&	92.5 	\\
	&	$\tau_3$	&	-0.208 	&	11.3 	&	11.8 	&	93.7 	&		&	-0.972 	&	11.2 	&	12.2 	&	92.3 	&		&	-0.967 	&	11.2 	&	12.2 	&	92.3 	\\
	&	$\tau_4$	&	0.344 	&	13.5 	&	14.2 	&	93.2 	&		&	0.694 	&	13.6 	&	14.4 	&	92.8 	&		&	0.626 	&	13.5 	&	14.4 	&	92.8 	\\
	&	$\beta_2$	&	-0.116 	&	5.06 	&	5.29 	&	94.1 	&		&	-0.122 	&	5.06 	&	5.33 	&	94.0 	&		&	-0.102 	&	5.06 	&	5.31 	&	94.0 	\\
	&	$\gamma_1$	&	-0.065 	&	3.46 	&	3.50 	&	94.7 	&		&	-0.065 	&	3.46 	&	3.50 	&	94.8 	&		&	-0.064 	&	3.46 	&	3.49 	&	94.8 	\\
\hline
\end{tabular}}											
\end{table}

\begin{table}
{\footnotesize
\caption*{Table S10: Simulation results (multiplied by 100) for Case 3 with $\nu=0.8$, $\alpha^*=(-1,3,-2,-2,3)^\top$, $C_n=\log\{\log (n)\}$
and $\epsilon_i\sim t(4)$.}							
\begin{tabular}{ccrccccrccccrcccc}
\hline
& &\multicolumn{4}{c}{Oracle}& &\multicolumn{4}{c}{SCAD}& &\multicolumn{4}{c}{MCP}\\				
\cline{3-6} \cline{8-11} \cline{13-16} 	
$n$& $\nu$ 	& Bias	&	SE	&	SD	&	CP	&		&	Bias	&	SE	&	SD	&	CP	&		&	Bias	&	SE	&	SD	&	CP	\\	
\hline
1000&	$\alpha_0$	&	0.086 	&	4.28 	&	4.35 	&	95.3 	&		&	0.117 	&	4.29 	&	4.40 	&	95.1 	&		&	0.085 	&	4.29 	&	4.39 	&	95.2 	\\
	&	$\alpha_1$	&	2.123 	&	23.1 	&	24.5 	&	92.5 	&		&	3.283 	&	23.4 	&	25.3 	&	92.2 	&		&	3.614 	&	23.4 	&	25.1 	&	92.1 	\\
	&	$\alpha_2$	&	-2.257 	&	22.4 	&	24.4 	&	92.2 	&		&	-2.761 	&	22.7 	&	24.8 	&	92.4 	&		&	-3.020 	&	22.7 	&	24.8 	&	92.4 	\\
	&	$\alpha_3$	&	-2.765 	&	22.6 	&	23.1 	&	93.0 	&		&	-1.171 	&	22.2 	&	24.1 	&	92.4 	&		&	-1.325 	&	22.2 	&	24.0 	&	92.4 	\\
	&	$\alpha_4$	&	3.856 	&	26.1 	&	27.2 	&	93.7 	&		&	1.807 	&	25.9 	&	27.9 	&	92.7 	&		&	1.964 	&	25.9 	&	27.8 	&	92.7 	\\
	&	$\tau_1$	&	-0.706 	&	17.6 	&	18.3 	&	93.5 	&		&	-0.632 	&	17.7 	&	18.8 	&	93.2 	&		&	-0.425 	&	17.6 	&	18.7 	&	93.3 	\\
	&	$\tau_2$	&	-0.579 	&	15.9 	&	16.7 	&	93.4 	&		&	-2.040 	&	16.0 	&	17.5 	&	92.2 	&		&	-2.166 	&	16.0 	&	17.4 	&	92.3 	\\
	&	$\tau_3$	&	1.006 	&	15.9 	&	16.4 	&	92.8 	&		&	-1.136 	&	15.9 	&	18.1 	&	90.1 	&		&	-1.193 	&	15.9 	&	18.1 	&	90.1 	\\
	&	$\tau_4$	&	0.950 	&	19.1 	&	19.6 	&	94.8 	&		&	2.099 	&	19.2 	&	20.3 	&	93.8 	&		&	1.964 	&	19.2 	&	20.3 	&	93.8 	\\
	&	$\beta_2$	&	-0.515 	&	7.16 	&	7.24 	&	94.9 	&		&	-0.614 	&	7.18 	&	7.44 	&	94.3 	&		&	-0.559 	&	7.17 	&	7.40 	&	94.4 	\\
	&	$\gamma_1$	&	-0.024 	&	4.88 	&	4.99 	&	94.7 	&		&	-0.018 	&	4.88 	&	5.00 	&	94.7 	&		&	-0.020 	&	4.88 	&	5.00 	&	94.7 	\\	
2000	&	$\alpha_0$	&	-0.051 	&	3.04 	&	3.12 	&	93.7 	&		&	-0.032 	&	3.04 	&	3.17 	&	93.8 	&		&	-0.047 	&	3.04 	&	3.16 	&	93.7 	\\
	&	$\alpha_1$	&	0.986 	&	16.2 	&	16.1 	&	94.6 	&		&	1.741 	&	16.4 	&	16.5 	&	94.1 	&		&	1.901 	&	16.4 	&	16.4 	&	94.2 	\\
	&	$\alpha_2$	&	-0.884 	&	15.7 	&	16.0 	&	94.3 	&		&	-1.274 	&	15.8 	&	16.2 	&	93.9 	&		&	-1.406 	&	15.9 	&	16.2 	&	94.1 	\\
	&	$\alpha_3$	&	-0.973 	&	15.7 	&	16.5 	&	93.3 	&		&	-0.062 	&	15.5 	&	16.9 	&	92.8 	&		&	-0.141 	&	15.6 	&	16.8 	&	92.9 	\\
	&	$\alpha_4$	&	1.440 	&	18.2 	&	18.5 	&	93.7 	&		&	0.333 	&	18.1 	&	18.8 	&	93.4 	&		&	0.399 	&	18.1 	&	18.7 	&	93.8 	\\
	&	$\tau_1$	&	-0.087 	&	12.5 	&	12.4 	&	94.5 	&		&	0.001 	&	12.5 	&	12.7 	&	94.6 	&		&	0.095 	&	12.5 	&	12.6 	&	94.4 	\\
	&	$\tau_2$	&	-0.039 	&	11.2 	&	12.4 	&	92.7 	&		&	-0.954 	&	11.3 	&	12.3 	&	92.2 	&		&	-1.015 	&	11.2 	&	12.3 	&	92.0 	\\
	&	$\tau_3$	&	-0.209 	&	11.2 	&	11.5 	&	94.3 	&		&	-1.268 	&	11.2 	&	12.2 	&	91.9 	&		&	-1.273 	&	11.2 	&	12.2 	&	91.9 	\\
	&	$\tau_4$	&	0.285 	&	13.5 	&	13.9 	&	93.7 	&		&	0.952 	&	13.5 	&	14.3 	&	93.1 	&		&	0.880 	&	13.5 	&	14.3 	&	93.1 	\\
	&	$\beta_2$	&	-0.091 	&	5.05 	&	5.22 	&	94.6 	&		&	-0.153 	&	5.06 	&	5.30 	&	94.2 	&		&	-0.129 	&	5.06 	&	5.28 	&	94.2 	\\
	&	$\gamma_1$	&	-0.065 	&	3.46 	&	3.50 	&	94.6 	&		&	-0.066 	&	3.46 	&	3.51 	&	94.7 	&		&	-0.066 	&	3.46 	&	3.51 	&	94.7 	\\
\hline
\end{tabular}}											
\end{table}

\end{document}